\documentclass[a4paper,11pt,twoside]{thesis}

\usepackage{amsmath}
\usepackage{amsthm}
\usepackage{graphicx}
\usepackage{amssymb}
\usepackage[breaklinks=true]{hyperref}

\newcommand{\comment}[1]{}
\newcommand{\ket}[1]{|#1\rangle}
\newcommand{\bra}[1]{\langle #1|}
\newcommand{\braket}[2]{\langle #1 | #2 \rangle}
\newcommand{\proj}[1]{\ket{#1} \bra{#1}}
\newcommand{\norm}[1]{\left\|#1\right\|}
\newcommand{\vectornorm}[1]{\left|\left|#1\right|\right|}
\newcommand{\eps}{\epsilon}
\renewcommand{\Pr}{\mathbb{P}}
\newcommand{\Expect}{\mathbb{E}}
\newcommand{\tr}{\operatorname{tr}}
\newcommand{\swap}{\mathcal{F}}
\newcommand{\ot}{\otimes}
\newcommand{\bbI}{\mathbb{I}}
\newcommand{\cC}{\mathcal{C}}
\newcommand{\cE}{\mathcal{E}}
\newcommand{\cG}{\mathcal{G}}
\newcommand{\cF}{\mathcal{F}}
\newcommand{\cV}{\mathcal{V}}
\newcommand{\cU}{\mathcal{U}}
\newcommand{\cS}{\mathcal{S}}
\newcommand{\cP}{\mathcal{P}}
\newcommand{\cN}{\mathcal{N}}
\newcommand{\cH}{\mathcal{H}}
\newcommand{\bfm}{\mathbf{m}}
\newcommand{\bfn}{\mathbf{n}}

\newcommand{\bfq}{\mathbf{q}}
\newcommand{\bfp}{\mathbf{p}}
\newcommand{\tpiN}{\frac{2 \pi i}{N}}
\newcommand{\be}{\begin{equation}}
\newcommand{\ee}{\end{equation}}
\newcommand{\bes}{\begin{equation*}}
\newcommand{\ees}{\end{equation*}}
\def\ba#1\ea{\begin{align}#1\end{align}}
\def\bas#1\eas{\begin{align*}#1\end{align*}}
\def\bit{\begin{itemize}}
\def\eit{\end{itemize}}
\def\l{\left}
\def\r{\right}
\def\<{\langle}
\def\>{\rangle}
\newcommand{\stocdomr}{\trianglelefteq}
\newcommand{\stocdoml}{\trianglerighteq}
\newcommand{\majr}{\preceq}
\newcommand{\majl}{\succeq}
\renewcommand{\Re}{{\rm Re}}

\newtheorem{theorem}{Theorem}
\newtheorem{lemma}[theorem]{Lemma}
\newtheorem{definition}[theorem]{Definition}
\newtheorem{corollary}[theorem]{Corollary}
\newtheorem{conjecture}[theorem]{Conjecture}

\numberwithin{equation}{section}
\numberwithin{theorem}{section}

\newcommand{\nn}{\nonumber\\}
\newcommand{\eq}[1]{Eqn.~\ref{eq:#1}}
\newcommand{\eqn}[1]{Eqn.~\ref{eq:#1}}
\newcommand{\peq}[1]{(Eqn.~\ref{eq:#1})}
\newcommand{\thmref}[1]{Theorem \ref{thm:#1}}
\newcommand{\corref}[1]{Corollary \ref{cor:#1}}
\newcommand{\lemref}[1]{Lemma \ref{lem:#1}}
\newcommand{\chapref}[1]{Chapter \ref{chap:#1}}
\newcommand{\partref}[1]{Part \ref{part:#1}}
\newcommand{\secref}[1]{Section \ref{sec:#1}}
\newcommand{\defref}[1]{Definition \ref{def:#1}}

\newcommand{\figref}[1]{Figure \ref{fig:#1}}
\DeclareMathOperator{\supp}{supp}
\DeclareMathOperator{\Span}{span}
\DeclareMathOperator{\poly}{poly}
\DeclareMathOperator{\Par}{Par}
\renewcommand{\iff}{\ensuremath{\mathrm{\,iff\,}}}
\def\ra{\rightarrow}
\def\bbC{\mathbb{C}}
\def\bbE{\mathbb{E}}
\def\Pr{\mathbb{P}}
  
\title{Pseudo-randomness and Learning in Quantum Computation}
\author{Richard Andrew Low}
\thesisdec{A dissertation submitted to the University of Bristol in accordance with the requirements for award of the degree of Doctor of Philosophy in the Faculty of Engineering, Department of Computer Science, September 2009}
\wordcount{40,000}

\begin{document}

\maketitle

\begin{abstract}

This thesis discusses the young fields of quantum pseudo-randomness and quantum learning algorithms.  We present techniques for derandomising algorithms to decrease randomness resource requirements and improve efficiency.  One key object in doing this is a $k$-design, which is a distribution on the unitary group whose $k^{\rm th}$ moments match those of the unitarily invariant Haar measure.  We show that for a natural model of a random quantum circuit, the distribution of random circuits quickly converges to a 2-design.  We then present an efficient unitary $k$-design construction for any $k$, provided the number of qubits $n$ satisfies $k = O(n / \log n)$.  In doing this, we provide an efficient construction of a quantum tensor product expander, which is a generalisation of a quantum expander which in turn generalises classical expanders.  We then discuss applications of $k$-designs.  We show that they can be used to improve the efficiency of many existing algorithms and protocols and also find new applications to derandomising large deviation bounds.  In particular, we show that many large deviation bound results for Haar random unitaries carry over to $k$-designs for $k = \poly(n)$.

In the second part of the thesis, we present some learning and testing algorithms for the Clifford group.  We find an optimal algorithm for identifying an unknown Clifford operation.  We also give an algorithm to test if an unknown operation is close to a Clifford or far from every Clifford.

\end{abstract}

\begin{acknowledgements}

I am greatly indebted to my supervisor Aram Harrow for his teaching, support and inspiration throughout my PhD and for this I greatly thank him.  I would also like to thank the whole quantum information group at Bristol and members of the Computer Science department for their support and help.  In particular, I thank Mick Bremner, Rapha\"{e}l Clifford, Toby Cubitt, Richard Jozsa, Will Matthews, Ashley Montanaro,  Ben Sach, Dan Shepherd and Andreas Winter.

I also thank my family and friends for their encouragement and interest and members of the Hornstars, Dr Doctor, Cattle Market and Les Rosbifs for welcome distraction from this thesis.

I also acknowledge funding from the ARO grant ASTQIT and the EPSRC grant QIP-IRC.

\end{acknowledgements}

\setcounter{page}{1}

\tableofcontents

\setlength{\evensidemargin}{0.1cm}
%\addtolength{\oddsidemargin}{-0.3cm}
\setlength{\oddsidemargin}{1.6cm}
%\addtolength{\evensidemargin}{-0.6cm}

\chapter{Introduction}

Landauer famously said that information is physical \cite{LandauerInfoPhysical}.  A corollary of this is that computation is a physical process.  It is simply the evolution of a physical state, governed by the laws of physics.  A classical computer is therefore an information processor where the physical evolution is restricted to that of classical physics.  A quantum computer is more general: quantum evolution is allowed.  One might therefore reasonably expect that quantum computers are more powerful.  It might be that the extra possibilities allowed by quantum evolution allow states to be processed more efficiently to speed up the computation.  Determining which problems a quantum computer can solve faster than a classical computer is the central problem in the theory of quantum computation.

Significant progress has already been made in answering this question.  Shor's algorithm \cite{ShorsAlgorithm} shows that factoring of integers is possible in polynomial time on a quantum computer.  In contrast, it is not known if factoring is possible in polynomial time on a classical computer.  Also, Grover's unstructured search algorithm \cite{GroversAlgorithm} allows a marked item in an unsorted database to be found using only the square root of the time required on a classical computer.  Finding other algorithms and provable separations between quantum and classical computation is an important area of current research.

This thesis makes some progress towards finding such new algorithms.  In classical computer science, randomness and pseudo-randomness have been key tools in the development of new and faster algorithms.  In the first part of this thesis, we discuss applications and constructions of quantum analogues of these pseudo-random objects.  Whilst we do not come up with new algorithms based on these, we hope that in the future the tools we build will find application in this area.  In the second part of this thesis, we discuss problems in the theory of machine learning, which is an area in which quantum computers could outperform their classical counterparts.

A side theme in this thesis is the idea that computational complexity must be considered in physical models.  The converse of our opening statement is also true: physical systems store and process information; they are computers.  Therefore physical systems that could solve problems that are provably difficult do not exist in nature.  This can rule out models that provide too much computational power.

In \partref{PR}, we discuss quantum pseudo-randomness.  We introduce the subject in \chapref{PRIntro} and provide motivation from the classical computer science literature.  The main idea is to use pseudo-randomness instead of full randomness to decrease the amount of randomness required.  This is desirable in classical computing because random bits are expensive to produce.  In quantum computing random bits can be obtained by measurement but uniformly random unitaries and states (formally defined in \secref{RandomUnitaries}) cannot be produced efficiently so pseudo-randomness is necessary if efficiency is desired.  In classical computing random bits are often saved by limiting dependence, for example by using $k$-wise independent random variables.  These are variables where the distribution of any $k$ variables is the same as for fully independent random variables but dependencies become apparent when observing more than $k$ of the variables.  We discuss a quantum analogue of this known as a $k$-design.

In \chapref{RandomCircuits} we show that short random quantum circuits (see \secref{QuantumIntro} for background on quantum circuits) are 2-designs, giving an efficient method for producing a 2-design.  Then in \chapref{TPE} we provide an efficient $k$-design construction for all $k$, giving the first construction for $k>2$.  In order to do this, we present an efficient construction of a quantum $k$-tensor product expander, which is a quantum analogue of a classical tensor product expander which in turn is a generalisation of the standard expander used in classical computer science.  We then summarise known applications of $k$-designs in \chapref{DesignApplications} as well as providing our own to show that $k$-designs exhibit measure concentration which in some cases is almost as strong as for uniformly random unitaries.

In \partref{Learning} we turn to problems in learning theory.  In particular, we consider the problem of identifying a given black box unitary with as few queries as possible.  We find an algorithm with optimal asymptotic query complexity to identify an unknown unitary from the Clifford group (defined in \chapref{LearningCliffords}).  We also show how this can be done if the unitary only approximately implements a Clifford and we also present a testing algorithm to determine if a given operation is close to a Clifford or far from every Clifford.

\section{Brief Introduction to Quantum Mechanics}
\label{sec:QuantumIntro}

We now briefly mention some key concepts in quantum mechanics and define some notation.  For a more complete introduction see \cite{NielsenChuang}.

The state of a $d$-dimensional quantum system is represented by a vector in the complex space $\mathbb{C}^{d}$.  If $d=2$, we call the system a \emph{qubit} and often we will take $d=2^n$ and say the system has $n$ qubits.

We will normally use \emph{Dirac notation} for quantum states.  We write column vectors as $\ket{\psi}$, with the associated conjugate row vector as $\bra{\psi}$.  The inner product between two states is written as $\braket{\psi_1}{\psi_2}$.  We will write $\psi$ for the projector $\proj{\psi}$.  States can also be probabilistic mixtures of pure states.  If the state is $\ket{\psi_i}$ with probability $p_i$ then it has \emph{density matrix} $\sum_i p_i \proj{\psi_i}$.

It is often convenient to break the space up into different components, for example the system and its environment.  Mathematically, systems are combined by using the \emph{tensor product}.  The combined state of system $A$ in state $\ket{\psi_A}$ and system $B$ in state $\ket{\psi_B}$ is written $\ket{\psi_A} \ot \ket{\psi_B}$.  This leads to the phenomenon of \emph{entanglement},  which is when the combined state cannot be written in this product form.  For example, the state $\frac{1}{\sqrt{2}} \l( \ket{0_A} \ot \ket{0_B} + \ket{1_A} \ot \ket{1_B} \r)$ cannot be written in the product form $\ket{\psi_A} \ot \ket{\psi_B}$ and so is entangled.  To find the state of a subsystem $A$ from a density matrix $\rho_{AB}$, we take the \emph{partial trace}.  Write $\rho_{AB} = \sum_{ijkl} \rho_{ijkl} \ket{i_A} \bra{j_A} \ot \ket{k_B} \bra{l_B}$.  Then the reduced state is $\rho_A = \tr_B \rho_{AB} = \sum_{ijk} \rho_{ijkk} \ket{i_A} \bra{j_A}$.

Measurement of a quantum system can be written mathematically in terms of a POVM (positive operator valued measure).  This is a set of positive semi-definite operators $P_i$ such that $\sum_i P_i = I$.  Then the measurement outcomes are the labels $i$ and outcome $i$ occurs with probability $\tr P_i \rho$ if the state being measured is $\rho$.

The evolution of a closed quantum system is unitary.  That is, the quantum state at a later time $t$ is related by a unitary to the initial quantum state: $\rho_t = U_t \rho_0 U_t^\dagger$.  If the system of interest is part of some larger system then the dynamics need not be unitary.  The most general form of evolution can be written in the Kraus decomposition: $\rho_t = \sum_i A_i \rho_0 A_i^\dagger$ where $A_i$ are any operators normalised so that $\sum_i A_i^\dagger A_i = I$.

Finally we mention that it is often convenient to think of unitary evolution as a quantum circuit, built up of smaller elementary unitary gates.  This is in direct analogy to the use of circuits in classical computing.  Classically, a NAND gate suffices to produce any other gate so all classical circuits can be made up of just NAND gates.  Similarly, there exist sets of unitary gates from which any unitary can be built.  An example is the following three gates:
\bes
H = \frac{1}{\sqrt{2}}\begin{pmatrix}1 & 1 \\ 1 & -1\end{pmatrix} 
\qquad
R_{\pi/4} = \begin{pmatrix}1 & 0 \\ 0 & e^{i \pi / 4}\end{pmatrix}
\ees
\bes
CNOT = \begin{pmatrix}1 & 0 & 0 & 0 \\ 0 & 1 & 0 & 0 \\ 0 & 0 & 0 & 1 \\ 0 & 0 & 1 & 0 \end{pmatrix}.
\ees
We often seek to build a family of circuits that act on $n$ qubits for all $n$, where the gates are chosen from some elementary set, such as that above.  We say that the circuits are efficient if the number of gates grows only polynomially with $n$.

\section{Preliminaries}

Here we define some notation and concepts that are used throughout this thesis.

\subsection{Pauli Matrices}

We will often use the Pauli matrices:
\ba
\sigma_0 = \sigma_I &= \begin{pmatrix}1 & 0 \\ 0 & 1\end{pmatrix} 
\qquad
\sigma_1 = \sigma_x = \begin{pmatrix} 0& 1 \\ 1 & 0\end{pmatrix} \nonumber \\
\sigma_2 = \sigma_y &= \begin{pmatrix} 0& -i \\ i & 0\end{pmatrix} 
\qquad
\sigma_3 = \sigma_z = \begin{pmatrix} 1& 0 \\ 0 & -1\end{pmatrix} 
\ea
We can extend these to matrices on $n$ qubits by taking tensor products.  Let $p \in \{0, 1, 2, 3\}^n$ and $\sigma_p = \sigma_{p_1} \ot \sigma_{p_2} \ot \cdots \ot \sigma_{p_n}$ where $p_i$ is the value at the $i^{\text th}$ position in the string $p$.  We will sometimes use the alternative notation of $p \in \{I, x, y, z\}^n$.  We will refer to $\sigma_p$ as Pauli matrices on $n$ qubits.  There are $4^n$ Pauli matrices and they are orthogonal i.e.~$\tr \sigma_p \sigma_q = 2^n \delta_{pq}$.  Note also that $\sigma_p^2 = \sigma_0$, the identity.  Also, Pauli matrices either commute or anticommute.

The Pauli matrices form an orthogonal basis for matrices in $\bbC^{2^n \times 2^n}$.  Therefore any such matrix $A$ can be written in the form $\sum_p \gamma(p) \sigma_p$, with $\gamma(p) = \frac{1}{2^n} \tr \sigma_p A$.  Sometimes we will choose a different normalisation for the Pauli coefficients $\gamma(p)$ but will make this clear from the context.

\subsection{The Symmetric Group and Permutation Operators}

The symmetric group is the group of all permutations.  Let $\cS_N$ be the symmetric group on $N$ objects.  Then for $\pi \in \cS_N$ define the corresponding permutation operator
\be
B(\pi):=\sum_{i=1}^N \ket{\pi(i)}\bra{i}
\ee
to be the matrix that permutes the basis states $\ket{1},\ldots,\ket{N}$ according to $\pi$.

On the other hand, if we have $k$ $N$-dimensional systems then for $\pi\in\cS_k$ define the subsystem permutation operator $S(\pi)$ by
\be
S(\pi) := \sum_{n_1=1}^N\cdots \sum_{n_k=1}^N \ket{n_{\pi^{-1}(1)},\ldots
n_{\pi^{-1}(k)}}\bra{n_1,\ldots,n_k}.
\ee

Now we present two useful lemmas about subsystem permutation operators.
\begin{lemma}
\label{lem:TraceCycles}
Let $C$ be a cycle of length $c$ in $S_c$.  Then
\begin{equation*}
\tr \l( C \l( A_1 \otimes A_2 \otimes \ldots \otimes A_c \r)\r) = \tr \l(A_{C(1)} A_{C^{\circ 2}(1)} A_{C^{\circ 3}(1)} \ldots A_1 \r).
\end{equation*}
\end{lemma}
\begin{proof}
We have
\bas
\tr ( C ( A_1 &\otimes A_2 \otimes \ldots \otimes A_c )) \\
&= \sum_{i_1, i_2, \ldots, i_c} \bra{i_1 i_2 \ldots i_c} C \l( A_1 \otimes A_2 \otimes \ldots \otimes A_c \r) \ket{i_1 i_2 \ldots i_c} \\
&= \sum_{i_1, i_2, \ldots, i_c} \bra{i_1} A_{C(1)} \ket{i_{C(1)}} \bra{i_2} A_{C(2)} \ket{i_{C(2)}} \ldots \bra{i_c} A_{C(c)} \ket{i_{C(c)}} \\
&= \sum_{i_1, i_2, \ldots, i_c} \bra{i_1} A_{C(1)} \ket{i_{C(1)}} \bra{i_{C(1)}} A_{C^{\circ 2}(1)} \ket{i_{C^{\circ 2}(1)}} \ldots \bra{i_{C^{\circ c-1}(1)}} A_{1} \ket{i_1}
\eas
since $C^{\circ c}(1)=1$.  Evaluate the sum using the resolution of the identity to get the result.
\end{proof}
A simple example of this Lemma is that
\be
\tr (\swap (A \ot B)) = \tr A B
\ee
where $\swap$ is the swap operator.

We also work out the Pauli expansion of the swap operator.  To stress that this result does not depend on the choice of orthogonal basis we prove it in full generality.
\begin{lemma}
\label{lem:Swap}
The swap operator $\swap$ on two $d$-dimensional systems can be written as
\begin{equation*}
\frac{1}{d} \sum_{p} \sigma_p \otimes \sigma_p.
\end{equation*}
where $\{ \sigma_p \}$ form a Hermitian orthogonal basis with $\tr \sigma_p^2 = d$.
\end{lemma}
\begin{proof}
Expand $\swap$ in the basis and use \lemref{TraceCycles}:
\begin{align*}
\tr \l(\l(\sigma_p \otimes \sigma_q\r) \swap\r) &= \tr \sigma_p \sigma_q \\
&=
\begin{cases}
d	&	p=q \\
0	&	{\rm otherwise}.
\end{cases}
\end{align*}
The given sum has the correct coefficients in the basis therefore $\frac{1}{d} \sum_{p} \sigma_p \otimes \sigma_p=\swap$.
\end{proof}

\subsection{Asymptotic Notation}

We will use the following standard asymptotic notation.

\begin{definition}
\label{def:BigO}
$f(n) = O(g(n))$ if there exists $c, n_0 > 0$ such that $0 \le f(n) \le c g(n)$ for all $n \ge n_0$.
\end{definition}
\begin{definition}
\label{def:BigOmega}
$f(n) = \Omega(g(n))$ if there exists $c, n_0 > 0$ such that $f(n) \ge c g(n) \ge 0$ for all $n \ge n_0$.
\end{definition}
\begin{definition}
\label{def:BigTheta}
$f(n) = \Theta(g(n))$ if $f(n) = O(g(n))$ and $f(n) = \Omega(g(n))$.
\end{definition}
\begin{definition}
\label{def:LittleO}
$f(n) = o(g(n))$ if $\lim_{n \rightarrow \infty} f(n)/g(n) = 0$.
\end{definition}

\subsection{Norms and Superoperator Norms}

\subsubsection{Norms}

We will make heavy use of Schatten $p$-norms:
\begin{definition}
\label{def:SchattenPNorms}
For $A$ a $d \times d$ matrix, the Schatten $p$-norm is given by
\be
|| A ||_p = \l( \sum_{i = 1}^d \sigma_i^p \r)^{1/p}
\ee
where $\sigma_i$ are the singular values of $A$.
\end{definition}
In particular, $|| A ||_1 = \sum_{i = 1}^d \sigma_i = \tr \sqrt{A^\dagger A}$, $|| A ||_2 = \sqrt{\sum_{i = 1}^d \sigma_i^2} = \sqrt{\tr A^\dagger A}$ and $|| A ||_\infty = \max_i \sigma_i$.

These norms satisfy the following simple relationships:
\ba
|| A ||_2 &\le || A ||_1 \le \sqrt{d} || A ||_2 \\
|| A ||_\infty &\le || A ||_1 \le d || A ||_\infty \\
|| A ||_\infty &\le || A ||_2 \le \sqrt{d} || A ||_\infty.
\ea

\subsubsection{Superoperator Norms}

Just as state norms can be used to bound the distinguishability of states, superoperator norms bound how easy it is to tell different superoperators apart.  We start with the 1-norm:
\begin{definition}
\label{def:1normsuper}
The 1-norm of a superoperator $\cE$ is given by
\bes
\vectornorm{\cE}_{1 \rightarrow 1} = 
\sup_{X \ne 0} \frac{\vectornorm{\cE(X)}_{1}}{\vectornorm{X}_{1}}.
\ees
\end{definition}
The main problem with this definition is that the 1-norm is not stable under tensoring with the identity i.e.~there exist channels with $||\cE \ot \text{\rm id}_d||_{1 \rightarrow 1} > || \cE ||_{1 \rightarrow 1}$, where $\text{\rm id}_d$ is the identity channel on $d$ dimensions.  This means that some channels are easier to distinguish by inputting entangled states.  If the norm is to measure the distinguishability of channels it should take this into account.  To overcome this problem, the diamond norm is defined:
\begin{definition}[\cite{KSV02}]
\label{def:diamondNorm}
The diamond norm of a superoperator $\cE$ is given by
\bes
\vectornorm{\cE}_{\diamond} = \sup_d \vectornorm{\cE \otimes \text{\rm id}_d}_{1 \rightarrow 1} = 
\sup_d \sup_{X \ne 0} \frac{\vectornorm{(\cE \otimes
    \text{\rm id}_d)X}_{1}}{\vectornorm{X}_{1}}.
\ees
\end{definition}
If follows immediately that $|| \cE ||_{1 \rightarrow 1} \le || \cE ||_\diamond$.  Also it is shown in \cite{KSV02} that the diamond norm satisfies $||\cE \ot \text{\rm id}_d||_\diamond = || \cE ||_\diamond$ for all channels $\cE$ and dimensions $d$ and that the dimension $d$ in the supremum can be taken to be the same as the dimension of the system $\cE$ acts on.  Operationally, the diamond norm of the difference between two quantum operations tells us the largest possible probability of distinguishing the two operations if we are allowed to have them act on part of an arbitrary, possibly entangled, state.

We will also use the 2-norm:
\begin{definition}
\label{def:2normsuper}
The 2-norm of a superoperator $\cE$ is given by
\bes
\vectornorm{\cE}_{2 \rightarrow 2} = 
\sup_{X \ne 0} \frac{\vectornorm{\cE(X)}_{2}}{\vectornorm{X}_{2}}.
\ees
\end{definition}

In \cite{VanDamThesis} Appendix C, the following relationships between the superoperator norms are proven:
\ba
|| \cE ||_{2 \rightarrow 2} &\le \sqrt{d} || \cE ||_{1 \rightarrow 1} \label{eq:2normsuper1normsuper}\\
|| \cE ||_{1 \rightarrow 1} &\le \sqrt{d} || \cE ||_{2 \rightarrow 2} \label{eq:1normsuper2normsuper}\\
|| \cE ||_\diamond &\le d || \cE ||_{1 \rightarrow 1} \label{eq:diamondnorm1normsuper}\\
|| \cE ||_\diamond &\le d || \cE ||_{2 \rightarrow 2}. \label{eq:diamondnorm2normsuper}
\ea

\section{Previous Publications}

The majority of this thesis has been published previously and some is work in collaboration.

\chapref{RandomCircuits} is joint work with Aram Harrow and is available as ``Random Quantum Circuits are Approximate 2-designs'', Communications in Mathematical Physics, Volume 291, Number 1, Pages 257-302.  It is also available as a pre-print: arXiv:0802.1919.

\chapref{TPE} is also joint work with Aram Harrow and is available as ``Efficient Quantum Tensor Product Expanders and $k$-Designs'', Proceedings of RANDOM 2009, LNCS, Volume 5687, Pages 548-561.  It is also available as a pre-print: arXiv:0811.2597.

\chapref{DesignApplications} from \secref{LargeDeviations} onwards is available as ``Large deviation bounds for $k$-designs", Proceedings of the Royal Society A, Volume 465, Number 2111, Pages 3289-3308.  It is also available as a pre-print: arXiv:0903.5236.

\chapref{LearningCliffords} is available as ``Learning and Testing Algorithms for the Clifford Group'', Physical Review A, Volume 80, Number 5, Page 052314.  It is also available as a pre-print: arXiv:0907.2833.

\part{Quantum Pseudo-randomness}
\label{part:PR}

\chapter{Introduction to Quantum Pseudo-randomness}
\label{chap:PRIntro}

Randomness is an important resource in both classical and quantum computing.  It has applications in virtually all areas of computer science, including algorithms, cryptography and networking.  Randomness can improve efficiency or, as in the case of cryptography, allow us to perform tasks that we would not be able to do with deterministic resources.

An example of an algorithm where a randomised algorithm is faster than any known deterministic algorithm is polynomial identity testing.  Here, the task is to determine if two polynomials are identically equal.  By evaluating the polynomials on random inputs, identity testing can be done in polynomial time whereas no polynomial time deterministic algorithm is known.

Also, a commonly used randomised algorithm is that of randomised quicksort.  In quicksort, a pivot element is chosen and elements smaller than this are placed to the left and larger elements to the right.  Then these two parts are sorted recursively.  However, the choice of pivot element greatly affects the run-time of the algorithm.  If chosen poorly (for example so that there is only one element smaller than the pivot), the algorithm runs in $O(n^2)$ time.  If chosen well, the algorithm runs in $O(n \log n)$ time.  Choosing the pivot element randomly will be a good choice on average, giving expected run-time $O(n \log n)$ \cite{MotwaniRaghavan}.  However, this run-time can be achieved deterministically using deterministic median finding \cite{LinearTimeMedian} but in practice the randomised method is more efficient.

As another example, many primality testing algorithms are randomised because of their simplicity, even though a deterministic polynomial-time algorithm is now known.  Also, in the field of communication complexity, separations between deterministic and randomised algorithms can be proven.  The deterministic complexity of evaluating the equality function (to determine if Alice and Bob's strings are equal) is $\Theta(n)$, whereas the randomised complexity is $\Theta(\log n)$ \cite{KushilevitzNisan}.  As yet another example of randomness in classical computer science, in networking a random delay is often inserted after a collision so the nodes wait different times so are likely to avoid another collision.

In this part, we seek to extend some of these gains of using randomness to quantum computing.  We wish to find applications of randomness to find new quantum algorithms and constructions.

Besides the computer science applications, there are also physical reasons for studying randomness in quantum mechanics.  Some systems can be modelled as interacting randomly and it is interesting to ask what the limiting state (or distribution on states) is for such a system.  Also of great interest is how quickly the system reaches this stationary state.  If the time taken grows too quickly with the size of the system (for example, exponentially) then for any system apart from the most trivial, the stationary state will never be reached and will not be seen in physical systems.  However, if the time is small (for example, a small polynomial), then the stationary state can be reached quickly and will be observed in real systems.  It is in problems like this that physicists must consider the computer science aspects of their models.  We study problems of this kind in Chapters \ref{chap:RandomCircuits} and \ref{chap:DesignApplications}.

\section{Random Unitaries}
\label{sec:RandomUnitaries}

In quantum computing, operations are unitary gates and randomness is often used in the form of random unitary operations.  Random unitaries have algorithmic uses (e.g.~\cite{Sen05}), cryptographic applications (e.g.~\cite{AmbainisSmith04,RandomizingQuantumStates04}) and applications to fundamental quantum protocols (e.g.~\cite{RemoteStatePreparation05, SuperdenseCodingHHL}).  For information-theoretic applications, it is often convenient to use unitary matrices drawn from the uniform distribution on the unitary group, also known as the Haar measure.  This measure is the unique unitarily invariant measure i.e.~the only measure $dU$ on the unitary group $\cU(d)$ where $\int_{\cU(d)} f(U) dU = \int_{\cU(d)} f(UV) dU$ for all functions $f$ and unitaries $V$.  For random states, we write the unitarily invariant measure on $d$-dimensional states as $d\psi$.  This can be thought of as a Haar distributed unitary applied to any fixed pure state.  It is also known as the Fubini-Study metric.

However, in both classical and quantum computing, obtaining random bits is often expensive, and so it is often desirable to minimise their use.  For example, in classical computing, expanders (discussed in \chapref{TPE}) and $k$-wise independent functions (see \secref{k-designs}) have been developed for this purpose and have found wide application.  We will spend a great deal of time exploring quantum analogues of these: quantum expanders and $k$-designs.

In addition to randomness being expensive, there is an even more pressing problem when using random unitaries and states.  An $n$-qubit unitary is defined by $4^n$ real parameters, and so cannot even be approximated efficiently using a subexponential amount of time or randomness.  So any application that requires a random unitary cannot be efficient.  Instead, we will seek to construct efficient pseudo-random ensembles of unitaries which resemble the Haar measure for certain applications.  For example, a $k$-design (often referred to as a $t$-design, or a $(k,k)$-design), as mentioned above, is a distribution on unitaries which matches the first $k$ moments of the Haar distribution.  $k$-designs have found many uses which are explored in \chapref{DesignApplications}.

In \secref{k-designs}, we formally define $k$-designs and summarise known constructions.  Then in \chapref{RandomCircuits} we show that, for a natural model of a random quantum circuit, the distribution quickly converges to that of a 2-design.  This gives an efficient approximate 2-design construction and also has physical applications.  In \chapref{TPE}, we provide an efficient construction of a unitary $k$-design for any $k$ (although there are restrictions on the dimension, see later).  Then in \chapref{DesignApplications}, we discuss applications of designs, including to derandomising constructions that use large deviation bounds.

Parts of this chapter have been published previously in \cite{RandomCircuits,TPE,LargeDeviationskDesigns} and parts are joint work with Aram Harrow.

\section{\texorpdfstring{$k$}{k}-designs}
\label{sec:k-designs}

A unitary $k$-design is a distribution of unitaries that gives the same expectations of polynomials of degree at most $k$ as the Haar measure.  This is just like Gaussian quadrature, where integrals of polynomials are calculated by sums.  Gaussian quadrature says that there exist sample points $\{ x_i \}$ and weights $\{ w_i \}$ so that for all polynomials $f$ of degree at most $2b-1$,
\be
\sum_{i=1}^b w_i f(x_i) = \int_{p}^q dx f(x)
\ee
for some fixed limits $p$ and $q$.  This allows the integrals to be calculated much more efficiently.  A unitary $k$-design is the same, except the polynomial is on elements of unitary matrices from the unitary group rather than numbers on the real line.  The $k$ refers to the degree of the polynomial.  We will also discuss state designs, where the function is on coefficients of states rather than unitaries.

\subsection{\texorpdfstring{$k$}{k}-wise Independence}

$k$-designs can also be thought of as a quantum analogue of $k$-wise independence.  A sequence of random variables $X_1, \ldots, X_n$ is $k$-wise independent if, for any subset of size $j \le k$,
\be
\Pr(X_{i_1} = x_{i_1} , \ldots, X_{i_j}  = x_{i_j}) = \Pr(X_{i_1} = x_{i_1}) \ldots \Pr(X_{i_j}  = x_{i_j}).
\ee
As a simple example of how this can save randomness, consider the set
\be
\{000, 011, 101, 110\}.
\ee
If an element is chosen uniformly at random from this set, the probability distribution of the values of any two bits is the same as if all three bits were chosen independently.  This is therefore a 2-wise independent set, and saves one bit of randomness.  In general, if $k \ll n$, an exponential saving in randomness can be made in this way.  Efficient constructions of exactly $k$-wise independent sets are known \cite{ExactkWiseIndep} and more efficient approximate constructions are given in \cite{NaorNaorkWiseIndep}.

A related concept is that of $k$-wise independent permutations.  These are sets of permutations with the property that, when a permutation is chosen randomly from this set and applied to $n$ points, the distribution of the positions of any $k$ points is the same as if a uniformly random permutation was applied.  For example, a random cyclic shift is a 1-wise independent permutation.  Again, an exponential saving of randomness is possible \cite{KNRkWiseIndepPerms}.

We seek to construct quantum $k$-designs to achieve a similar saving of randomness for quantum algorithms.  We now formally define $k$-designs.

\subsection{Exact Designs}

We will use the following notation to distinguish the measure we are using.  Write $\bbE$ for the expectation with $\bbE_{U \sim \nu}$ meaning the expectation when $U$ is chosen from the measure $\nu$.  If the measure is the Haar measure in dimension $d$ we will write $\bbE_{U \sim \cU(d)}$.  We use the same subscripts for probabilities so $\Pr_{U \sim \cU(d)}$ denotes the probability when $U$ is chosen from the Haar measure, etc..  When considering random states, we will write $\bbE_{\ket{\psi} \sim \cS(d)}$, etc..

\subsubsection{State designs}

A $k$-design is an ensemble of states such that, when one state is
chosen from the ensemble and copied $k$ times, it is indistinguishable
from a uniformly random state.  The state $k$-design definition we use is due to Ambainis and Emerson \cite{AmbainisEmerson07}:
\begin{definition}[\cite{AmbainisEmerson07}, Definition 1]
\label{def:StateDesign}
An ensemble of quantum states $\nu = \{ p_i, \ket{\psi_i} \}$ is a state $k$-design if
\be
\bbE_{\ket{\psi} \sim \nu} \l[ \l( \ket{\psi} \bra{\psi} \r)^{\ot k} \r] = \bbE_{\ket{\psi} \sim \cS(d)} \l[ \l( \ket{\psi} \bra{\psi} \r)^{\ot k} \r]
\ee
\end{definition}
We can evaluate the integral on the right hand side:
\begin{lemma}
\label{lem:SymmetricStateAverage}
\be
\int_\psi \l( \ket{\psi} \bra{\psi} \r)^{\ot k} d\psi = \frac{\Pi_{+k}}{{k+d-1 \choose k}}
\ee
where $\Pi_{+k}$ is the
projector onto the symmetric subspace of $k$ $d$-dimensional spaces.
\end{lemma}
\begin{proof}
The standard proof (see e.g.~\cite{GoodmanWallach98} or \cite{BBDEJM97}) involves showing that $\int_\psi \l( \ket{\psi} \bra{\psi} \r)^{\ot k} d\psi$ commutes with all elements of an irreducible representation (irrep) of the unitary group that acts on the symmetric subspace so by Schur's lemma must be proportional to the projector onto the symmetric subspace.  However, here we give an alternative proof that introduces a technique we will use later.

By the unitary invariance of the Haar measure, $\int_\psi \l( \ket{\psi} \bra{\psi} \r)^{\ot k} d\psi$ commutes with $U^{\ot k}$ for all unitaries $U$.  By Schur-Weyl duality (see e.g.~\cite{GoodmanWallach98}), this implies that the integral is a linear combination of subsystem permutation operators.  Therefore we have
\be
\int_\psi \l( \ket{\psi} \bra{\psi} \r)^{\ot k} d\psi = \sum_{\pi \in S_k} \alpha_{\pi} S(\pi).
\ee
However, the integral is invariant under permutations so $\alpha_{\pi}$ must be the same for all permutations $\pi$.  Using $\Pi_{+k} = \frac{1}{k!} \sum_{\pi \in S_k} S(\pi)$ and finding the normalisation by taking the trace (the dimension of the symmetric subspace is ${k+d-1 \choose k}$) proves the result.
\end{proof}

We will also state equivalent definitions of designs in terms of polynomials of matrix elements of the unitary or coefficients of the state.  First we must define what we mean by the degree of a polynomial:
\begin{definition}
A monomial in elements of a matrix $U$ or state $\ket{\psi}$ is of degree $(k_1,k_2)$ if it contains $k_1$ conjugated elements and $k_2$ unconjugated elements.  We call it balanced if $k_1=k_2$ and will simply say a balanced monomial has degree $k$ if it is degree $(k,k)$.  A balanced polynomial is of degree $k$ if it is a sum of balanced monomials of degree at most $k$, with at least one monomial with degree equal to $k$.
\end{definition}
So that, in this definition, $U_{pq} U^*_{rs}$ is a balanced monomial of degree $(1,1)$ and $U_{pq} U_{rs}$ is a monomial of degree $(2,0)$ and is unbalanced.  For the state $\ket{\psi} = \sum_i \alpha_i \ket{i}$, $\alpha_i \alpha_j^*$ is a balanced monomial of degree $(1, 1)$.

We can then define state $k$-designs in terms of monomials:
\begin{definition}[\cite{AmbainisEmerson07}, Definition 3]
\label{def:StateDesignMonomials}
An ensemble of quantum states $\nu$ is a state $k$-design if, for all balanced monomials $M$ of degree at most $k$,
\be
\bbE_{\ket{\psi} \sim \nu} M(\ket{\psi}) = \bbE_{\ket{\psi} \sim \cS(d)} M(\ket{\psi})
\ee
\end{definition}
This is an equivalent definition to \defref{StateDesign}:
\begin{lemma}[\cite{AmbainisEmerson07}, Theorem 5]
The state design definitions \ref{def:StateDesign} and \ref{def:StateDesignMonomials} are equivalent.
\end{lemma}
\begin{proof}
Firstly, we only need to prove the result for $M$ of degree exactly $k$, since by partial tracing this implies the result for any smaller $k$.

Each entry in the matrix $\bbE_{\ket{\psi} \sim \nu} \l[ \l( \ket{\psi} \bra{\psi} \r)^{\ot k} \r]$ is the expectation of a monomial of degree $k$, with the state chosen from the design.  Further, the corresponding entry in $\bbE_{\ket{\psi} \sim \cS(d)} \l[ \l( \ket{\psi} \bra{\psi} \r)^{\ot k} \r]$ is the expectation of the same monomial but with the state chosen from the Haar measure.  If the ensemble of states satisfies \defref{StateDesignMonomials} then these are equal, so the ensemble also satisfies \defref{StateDesign}.

On the other hand, for every balanced monomial of degree $k$, there is an entry in $\bbE_{\ket{\psi} \sim \nu} \l[ \l( \ket{\psi} \bra{\psi} \r)^{\ot k} \r]$ equal to its expectation.  Therefore, if the ensemble of states satisfies \defref{StateDesign} then it also satisfies \defref{StateDesignMonomials}.
\end{proof}

\subsubsection{Unitary designs}

Consider having $k$ $d$-dimensional systems in any initial state.  A unitary $k$-design is an ensemble of unitaries such that when a unitary is randomly selected from it and applied to each of the $k$ systems, the overall state is indistinguishable from choosing a uniformly random unitary.  This can be seen as a generalisation of state designs in that any column of a unitary $k$-design is a state $k$-design.  Formally, we have:
\begin{definition}
\label{def:UnitaryDesign}
Let $\nu$ be an ensemble of unitary operators.  Define
\be
\label{eq:UnitaryDesign}
\cG_{\nu}(\rho) = \bbE_{U \sim \nu} \l[ U^{\ot k} \rho (U^\dagger)^{\ot k} \r]
\ee
and
\be
\cG_H(\rho) = \bbE_{U \sim \cU(d)} \l[ U^{\ot k} \rho (U^\dagger)^{\ot k} \r]
\ee
Then the ensemble is a unitary $k$-design if $\cG_\nu(\rho) = \cG_H(\rho)$ for all $d^k \times d^k$ matrices $\rho$ (not necessarily physical states).
\end{definition}
For convenience we have defined this for all matrices $\rho$ although it is equivalent to only require equality for physical states, since all matrices can be obtained from linear combinations of physical states.

Like state designs, unitary designs can also be defined in terms of polynomials:
\begin{definition}[\cite{DCEL06}]
\label{def:UnitaryDesignMonomials}
$\nu$ is a unitary $k$-design if, for all balanced monomials $M$ of degree $k$,
\be
\bbE_{U \sim \nu} M(U) = \bbE_{U \sim \cU(d)} M(U).
\ee
\end{definition}
Again, these definitions are equivalent:
\begin{lemma}
The unitary design definitions \ref{def:UnitaryDesign} and \ref{def:UnitaryDesignMonomials} are equivalent.
\end{lemma}
\begin{proof}
The proof is very similar to the state design case.  Again, we only consider monomials of degree $k$ since by partial tracing this implies the result for smaller $k$.

Consider matrices $\rho$ of the form $\ket{i_1, i_2, \ldots, i_k} \bra{j_1, j_2, \ldots, j_k}$ in \defref{UnitaryDesign}.  Then each element of $U^{\ot k} \rho \left(U^\dagger\right)^{\ot k}$ is a balanced monomial of degree $k$ and, for some choice of indices in $\ket{i_1, i_2, \ldots, i_k} \bra{j_1, j_2, \ldots, j_k}$, each balanced monomial of degree $k$ appears.
\end{proof}

\subsection{Approximate \texorpdfstring{$k$}{k}-designs}

While exact designs have desirable properties, it is often much easier to construct approximate designs which, for many applications, are sufficient.  Also, approximate designs can have fewer unitaries than exact designs.  For example, it was shown in \cite{AMTW00} that $2^{2n}$ unitaries are necessary and sufficient for an exact unitary 1-design.  However, an approximate 1-design can be implemented with only $2^{n+o(n)}$ unitaries which gives almost a factor of 2 saving in random bits.

\subsubsection{Approximate state designs}

Our approximate state design definition is as follows:
\begin{definition}
\label{def:ApproxStateDesign}
$\nu$ is an $\eps$-approximate state $k$-design if
\be
\left|\left|\bbE_{\ket{\psi} \sim \nu} \l[ \left( \ket{\psi} \bra{\psi} \right)^{\ot k} \r] - \bbE_{\ket{\psi} \sim \cS(d)} \l[ \left( \ket{\psi} \bra{\psi} \right)^{\ot k} \r]\right|\right|_\infty \le \frac{\eps}{{k+d-1 \choose k}}.
\ee
\end{definition}
${k+d-1 \choose k}$ appears because it is the dimension of the symmetric subspace.  In \cite{AmbainisEmerson07}, a similar definition was proposed but
with the additional requirement that the ensemble also forms a
1-design (exactly), i.e.
\bes
\bbE_{\ket{\psi} \sim \nu} \ket{\psi} \bra{\psi} = \bbE_{\ket{\psi} \sim \cS(d)} \ket{\psi} \bra{\psi}
\ees
This requirement was necessary there only so that a suitably normalised
version of the ensemble would form a POVM. We will not use it.

By taking the partial trace one can show that a
$k$-design is a $k'$-design for $k' \le k$.  Thus approximate
$k$-designs are always at least approximate 1-designs.

\subsubsection{Approximate unitary designs}

We have many choices to make when defining an approximate design.  Here we give four definitions which are convenient in different contexts.  In \lemref{ApproxUnitaryDesignEquiv} we show that they are all equivalent, up to polynomial dimension factors.

If the unitary design is considered a quantum channel that applies a random unitary from the distribution to the input, then a relevant measure is the diamond norm difference between the approximate design and an exact design.  Because the diamond norm is related to the distinguishability of channels, having a low diamond norm distance means that it is difficult to detect that an approximate design was given rather than exact.  One approximate design definition is therefore:
\begin{definition}[DIAMOND\footnote{We name the definitions to help distinguish them}, See \chapref{RandomCircuits}]
\label{def:ApproxUnitaryDesignDiamond}
$\nu$ is an $\eps$-approximate unitary $k$-design if
\begin{equation}
\label{eq:ApproxUnitaryDesignDiamond}
\vectornorm{\cG_\nu - \cG_H}_{\diamond} \le \eps,
\end{equation} 
where $\cG_\nu$ and $\cG_H$ are defined in \defref{UnitaryDesign}.
\end{definition}
In \cite{DCEL06}, they consider approximate twirling, which is implemented using an approximate 2-design.  They give an alternative definition of closeness which is more convenient for this application:
\begin{definition}[TWIRL, \cite{DCEL06}]
\label{def:ApproxUnitaryDesignDankert}
$\nu$ is an $\eps$-approximate twirl if
\begin{equation}
\max_\Lambda \vectornorm{\bbE_{U \sim \nu} U^\dagger \Lambda(U \rho U^\dagger)U - \Expect_{U \sim \cU(d)} U^\dagger \Lambda(U \rho U^\dagger)U}_{\diamond} \le \frac{\eps}{d^2}.
\end{equation}
The maximisation is over channels $\Lambda$ and $d$ is the dimension.
\end{definition}

In \chapref{TPE}, unitary designs are constructed from quantum tensor product expanders.  A quantum $k$-TPE is defined as an ensemble $\nu$ of unitaries such that
\be
\left\|\bbE_{U\sim \nu} \l[U^{\ot k,k}\r] - \bbE_{U \sim \cU(d)}
\l[U^{\ot k,k}\r] \right\|_\infty \le \lambda
\ee
for $\lambda < 1$ and $U^{\ot k, k} = U^{\ot k} \ot \l( U^* \r)^{\ot k}$ (the motivation for this definition is explained in \chapref{TPE}).  From this a natural $k$-design definition follows:
\begin{definition}[TRACE, See \chapref{TPE}]
$\nu$ is an $\eps$-approximate unitary $k$-design if
\label{def:ApproxUnitaryDesignkk}
\be
\left\| \bbE_{U\sim \nu} \l[U^{\ot k,k}\r]  - \bbE_{U \sim \cU(d)} \l[U^{\ot k,k}\r] \right\|_1 \le \eps.
\ee
\end{definition}
In \thmref{DesignFromTPE} we prove the simple result that a unitary design can be constructed by iterating the TPE.

We will also need a definition in terms of monomials:
\begin{definition}[MONOMIAL, See \chapref{DesignApplications}]
\label{def:ApproxUnitaryDesignMonomials}
$\nu$ is an $\eps$-approximate unitary $k$-design if, for all balanced monomials $M$ of degree $\le k$,
\begin{equation}
\label{eq:ApproximateUnitarykdesignMonomials}
\l| \bbE_{U \sim \nu} M(U) - \bbE_{U \sim \cU(d)} M(U) \r| \le \frac{\eps}{d^k}
\end{equation}
\end{definition}

We would now like to show that all these definitions are equivalent.  By equivalent, we mean that, if $\nu$ is an $\eps$-approximate unitary design by one definition, then it is an $\eps'$-approximate unitary design by any other definition, where $\eps' = \poly(d^k) \eps$.

\begin{lemma}
\label{lem:ApproxUnitaryDesignEquiv}
Definitions \ref{def:ApproxUnitaryDesignDiamond} (DIAMOND), \ref{def:ApproxUnitaryDesignkk} (TRACE) and \ref{def:ApproxUnitaryDesignMonomials} (MONOMIAL) are all equivalent.  Also \defref{ApproxUnitaryDesignDankert} (TWIRL) is equivalent to the other definitions for an approximate 2-design only.
\end{lemma}
\begin{proof}
To prove this, we will consider yet another possible definition (OPERATOR-2-NORM):
\be
\label{eq:ApproxUnitaryDesign2Norm}
\l|\l| \cG_{\nu} - \cG_H \r|\r|_{2 \rightarrow 2} \le \eps.
\ee
Note that this is equivalent to
\be
\label{eq:ApproxUnitaryDesignInftyNorm}
\l\| \bbE_{U\sim \nu} [U^{\ot k,k}]  - \bbE_{U \sim \cU(d)} [U^{\ot k,k}] \r\|_\infty \le \eps
\ee
which is the same as \defref{ApproxUnitaryDesignkk} (TRACE) except the norm is the $\infty$-norm rather than the 1-norm.  We shall prove that the other $k$-design definitions are equivalent to this.   We then show that \defref{ApproxUnitaryDesignDankert} (TWIRL) is equivalent to \defref{ApproxUnitaryDesignMonomials} (MONOMIAL) for $k=2$.  We use notation $A \xrightarrow{s}B$ to mean that if $\nu$ is an $\eps$-approximate unitary $k$-design according to definition A then it is a $s \eps$-approximate unitary $k$-design according to definition B.  If $s=1$ we omit the superscript.

A diagram showing the different parts to the proof is given in \figref{ApproxDesignEquivDiag}.  We remark that direction 2 is unneeded but is included since it provides tighter bounds and has a simple proof.

\begin{figure}[ht]
  \begin{center}
    \includegraphics[width=8cm]{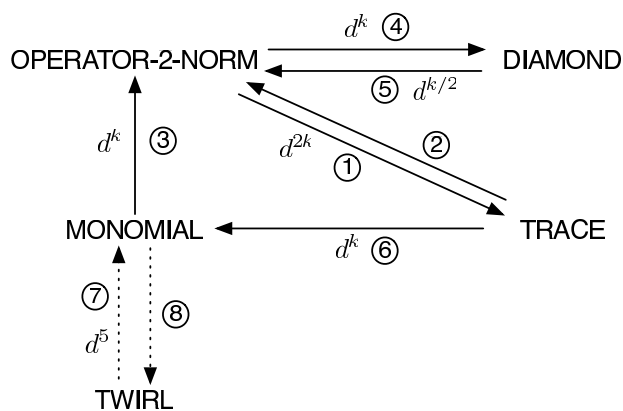}
    \caption[A diagram showing the parts of the proof of \lemref{ApproxUnitaryDesignEquiv}]{A diagram showing the different parts of the proof of \lemref{ApproxUnitaryDesignEquiv}.  The dotted arrows show the correspondence is only for $k=2$.  The circled digits refer to the enumerated items below and the factors by the arrows indicate the precision lost in the approximation when converting between the definitions.}
    \label{fig:ApproxDesignEquivDiag}
  \end{center}
\end{figure}

\begin{enumerate}

\item
OPERATOR-2-NORM $\xrightarrow{d^{2k}}$ TRACE:

Use the equivalence between Equations \ref{eq:ApproxUnitaryDesign2Norm} and \ref{eq:ApproxUnitaryDesignInftyNorm} and that \bes \l\| \bbE_{U\sim \nu} [U^{\ot k,k}]  - \bbE_{U \sim U(N)} [U^{\ot k,k}] \r\|_1 \le d^{2k} \l\| \bbE_{U\sim \nu} [U^{\ot k,k}]  - \bbE_{U \sim \cU(d)} [U^{\ot k,k}] \r\|_\infty. \ees

\item
TRACE $\rightarrow$ OPERATOR-2-NORM:

Use \bes \l\| \bbE_{U\sim \nu} [U^{\ot k,k}]  - \bbE_{U \sim \cU(d)} [U^{\ot k,k}] \r\|_\infty \le \l\| \bbE_{U\sim \nu} [U^{\ot k,k}]  - \bbE_{U \sim \cU(d)} [U^{\ot k,k}] \r\|_1 \ees and the equivalence between Equations \ref{eq:ApproxUnitaryDesign2Norm} and \ref{eq:ApproxUnitaryDesignInftyNorm}.

\item
MONOMIAL $\xrightarrow{d^k}$ OPERATOR-2-NORM:

Choose any $\rho \in \bbC^{d^k \times d^k}$ and write it as $\rho = \sum_{ij} \rho_{ij} \ket{i} \bra{j}$.  Then
\bas
\l\| \cG_\nu(\rho) - \cG_H(\rho) \r\|_2 &\le \sum_{ij} | \rho_{ij} | \l\| \cG_\nu( \ket{i} \bra{j} ) - \cG_H( \ket{i} \bra{j} ) \r\|_2 \\
&= \sum_{ij} | \rho_{ij} | \sqrt{ \sum_{kl} \l| \l( \cG_\nu( \ket{i} \bra{j} ) - \cG_H( \ket{i} \bra{j} ) \r)_{kl} \r|^2 }
\eas
using the fact that the 2-norm squared is the sum of the squares of the matrix elements.  Now, we have a bound on the matrix elements of $\cG_\nu( \ket{i} \bra{j} ) - \cG_H( \ket{i} \bra{j} )$ from \defref{ApproxUnitaryDesignMonomials} (MONOMIAL):
\bes
| \l(\cG_\nu( \ket{i} \bra{j} ) - \cG_H( \ket{i} \bra{j} )\r)_{kl} | \le \frac{\eps}{d^k}
\ees
so
\bas
\l\| \cG_\nu(\rho) - \cG_H(\rho) \r\|_2 &\le \eps \sum_{ij} | \rho_{ij} | \\
&\le d^k \eps || \rho ||_2.
\eas

\item
OPERATOR-2-NORM $\xrightarrow{d^k}$ DIAMOND:

This follows from the superoperator norm relationship given in \eq{diamondnorm2normsuper}.

\item
DIAMOND $\xrightarrow{d^{k/2}}$ OPERATOR-2-NORM:

This uses the operator norm inequalities $|| \phi ||_{1 \rightarrow 1} \le || \phi ||_\diamond$ and \eq{2normsuper1normsuper}.

\item
TRACE $\xrightarrow{d^k}$ MONOMIAL:

Let $M$ be a balanced monomial of degree $k$ and write it as \bes M = U_{p_1 q_1} \ldots U_{p_k q_k} U^*_{r_1 s_1} \ldots U^*_{r_k s_k}. \ees  Then let $\hat{M} = \ket{p_1} \bra{q_1} \ot \ldots \ot \ket{p_k} \bra{q_k} \ot \ket{r_1} \bra{s_1} \ot \ldots \ot \ket{r_k} \bra{s_k}$.  Then $M(U) = \tr \hat{M} U^{\ot k,k}$ and $\| \hat{M} \|_\infty = 1$.  Now we use the fact that for any operator $A$
\be
\| A \|_1 = \max_B \{ \tr A B : \| B \|_\infty \le 1 \}
\ee
to rewrite the TRACE definition:
\bas
\| \bbE_{U \sim \nu} [ U^{\ot k,k} ] - &\bbE_{U \sim \cU(d)} [ U^{\ot k,k} ] \|_1 = \\
&\max_B \l\{ \tr \l( \bbE_{U \sim \nu} [U^{\ot k,k}] - \bbE_{U \sim \cU(d)} [U^{\ot k,k}] \r) B : \| B \|_\infty \le 1 \r\} \\
&\ge \l| \tr \l(\bbE_{U \sim \nu} [U^{\ot k,k}] - \bbE_{U \sim \cU(d)} [U^{\ot k,k}] \r) \hat{M} \r| \\
&= \l| \bbE_{U \sim \nu} M(U) - \bbE_{U \sim \cU(d)} M(U) \r|.
\eas

\item
MONOMIAL $\xrightarrow{d^5}$ TWIRL (for $k=2$):

Write $\Lambda(\rho)$ in the Kraus decomposition as
\be
\Lambda(\rho) = \sum_k A_k \rho A_k^\dagger
\ee
with
\be
\label{eq:KrausNorm}
\sum_k A_k^\dagger A_k = I.
\ee
Let $\Lambda_U(\rho) = \sum_k U^\dagger A_k U \rho U^\dagger A_k^\dagger U$.  Then the $p,q$ matrix element of $\Lambda_U(\rho)$ is
\be
\sum_{krstuij} \rho_{ij} U_{si} U_{uq} U^*_{rp} U^*_{tj} A_{krs} A^*_{kut}.
\ee

From \defref{ApproxUnitaryDesignMonomials} (MONOMIAL) we have that \bes \l| (\bbE_{U \sim \nu} - \bbE_{U \sim \cU(d)}) U_{si} U_{uq} U^*_{rp} U^*_{tj} \r| \le \eps/d^2 \ees (treating expectation as an operator).  This implies that
\be
\label{eq:LambdaMatrixElements}
\l| \sum_{krstuij} \rho_{ij} A_{krs} A^*_{kut} (\bbE_{U \sim \nu} - \bbE_{U \sim \cU(d)}) U_{si} U_{uq} U^*_{rp} U^*_{tj} \r| \le \sum_{krstuij} | \rho_{ij} A_{krs} A^*_{kut} | \eps/d^2.
\ee
Now, $\sum_{rs} |A_{krs}| \le d || A_k ||_2$ and $\sum_{ij} |\rho_{ij}| \le d || \rho ||_2$ and, taking the trace of the normalisation condition \eq{KrausNorm} we find
\bes
d = \sum_k \tr A_k^\dagger A_k = \sum_k || A_k ||_2^2.
\ees
So we find \eq{LambdaMatrixElements} is upper bounded by
\bes
\frac{\eps}{d^2} d || \rho ||_2 \sum_k d^2 || A_k ||_2^2 \le \eps d^2 || \rho ||_2.
\ees
Using the fact that the 2-norm squared is the sum of the squares of the matrix elements we find that
\bes
|| \bbE_{U \sim \nu} \Lambda_U(\rho) - \bbE_{U \sim \cU(d)} \Lambda_U(\rho) ||_2 \le \eps d^4 || \rho ||_2.
\ees
Using $|| \cdot ||_\diamond \le d || \cdot ||_2$ (\eq{diamondnorm2normsuper}) we prove the result.

\item
TWIRL $\rightarrow$ MONOMIAL for $k=2$:

Let $A_\sigma = \ket{p}\bra{q} + \sigma \ket{r} \bra{s}$ where $\sigma \in \{+1, -1, +i, -i\}$.  Let $B = I - \ket{q}\bra{q} - \ket{s}\bra{s}$.  Then $A_\sigma$ and $B$ are the Kraus operators of a valid channel, provided $p \ne r$, which we assume for now.  Further, let
\be
\Lambda_{U,\sigma}(\rho) = U^\dagger \Lambda_\sigma (U \rho U^\dagger) U
\ee
where $\Lambda_\sigma$ is the channel with Kraus operators $A_\sigma$ and $B$.  Now let
\be
\Lambda_{U,s}(\rho) = \Lambda_{U ,+1}(\rho)-\Lambda_{U ,-1}(\rho)+i\Lambda_{U, +i}(\rho)-i\Lambda_{U, -i}(\rho).
\ee
We see that
\be
\Lambda_{U,s}(\rho) = 4 U^\dagger \ket{p} \bra{q} U \rho U^\dagger \ket{s} \bra{r} U.
\ee
Now, from \defref{ApproxUnitaryDesignDankert} (TWIRL) and the triangle inequality (using $|| \cdot ||_2 \le || \cdot ||_1$), we have
\be
|| \bbE_{U \sim \nu} \Lambda_{U,s}(\rho) - \bbE_{U \sim \cU(d)} \Lambda_{U,s}(\rho) ||_2 \le \frac{4 \eps || \rho ||_1}{d^2}.
\ee
This implies that each matrix element is small i.e.
\be
| (\bbE_{U \sim \nu} - \bbE_{U \sim \cU(d)}) \bra{c} U^\dagger \ket{p} \bra{q} U \rho U^\dagger \ket{s} \bra{r} U \ket{d} | \le \frac{\eps || \rho ||_1}{d^2}.
\ee
Now let $\rho = \ket{e} \bra{f}$.  We do not have to choose a physical state since the diamond-norm bound is true for all matrices.  This gives us
\be
| (\bbE_{U \sim \nu} - \bbE_{U \sim \cU(d)}) U^*_{pc} U_{qe} U^*_{sf} U_{rd} | \le \frac{\eps}{d^2}
\ee
as required.

For $p=r$, we also assume that $s=q$ since if not, just take $p \ne r$ and $s=q$ and swap the labels.  Here take $A_\pm = \pm \ket{p} \bra{q}$ and $B = I - \ket{q}\bra{q}$ and consider $\Lambda_{U,+}(\rho) - \Lambda_{U,-}(\rho) = 2 U^\dagger \ket{p} \bra{q} U \rho U^\dagger \ket{q} \bra{p} U$.\qedhere
\end{enumerate}
\end{proof}

We remark that other types of approximate definitions are possible.  For cryptographic uses, a computationally secure approximate design may be sufficient, rather than the information theoretic security discussed above.  A computationally secure approximate design would be nearly indistinguishable from an exact design in polynomial time.  Applications and constructions of such objects remain open problems.

\subsubsection{Constructions}

Here we summarise the known constructions of unitary and state designs.  We will say that a $k$-design construction is efficient if the effort required to sample a state or unitary from the design is polynomial in $n$ and $k$.  Note that we do not require the number of states or unitaries to be polynomial because, even for approximate designs, an exponential number is required.  Rather, the number of random bits needed to specify an element of the design should be $\poly(n,k)$.

We start with state design constructions since these have been studied far more than unitary designs.  Firstly, exact efficient state 1-designs are trivial: simply choose a random state from any basis.  Numerous examples of exact efficient state 2-design constructions are known (e.g.~\cite{Barnum02}).  Hayashi et al.~\cite{HHM06} give an inefficient construction of state $k$-designs for any $n$ and $k$ but general exact constructions are not efficient in $n$ and $k$.  However, Ambainis and Emerson provide an efficient approximate construction for any $k$ with $d \ge 2k$.  Aaronson \cite{Aaronson07} also gives an efficient approximate construction.

Less is known about efficient constructions for unitary designs.  It is straightforward to prove that the Pauli matrices form an exact 1-design and in \cite{DLT02, Dankert05} it is shown that the Clifford group (see \chapref{LearningCliffords} for a definition) forms an exact 2-design although no efficient exact sampling method is known.  However, an approximate sampling method is given in \cite{DLT02} and a more efficient approximate 2-design construction is given in \cite{DCEL06}.  The structure of unitary 2-designs is considered in \cite{GAE07}, providing lower bounds on the number of unitaries in the design. 

In \chapref{TPE} we give the first efficient approximate unitary $k$-design construction for $k > 2$.  The construction works in $O(k n + \log 1/\eps)$ time for $k = O(n / \log n)$.  Through \lemref{ApproxUnitaryDesignEquiv}, the construction is efficient for all the equivalent definitions above.  We also conjecture in \chapref{RandomCircuits} that random quantum circuits of length $\poly(n, k)$ are approximate unitary $k$-designs although we only prove this for $k=2$.

\chapter{Random Quantum Circuits}
\label{chap:RandomCircuits}

\section{Introduction: Pseudo-random Quantum Circuits}

Random circuits are a natural object to consider when looking at the complexity of random operations.  They are circuits where the gates and their positions are chosen randomly from some given distribution.  If the gate set that the random circuit chooses from is universal then, as we show below, the random circuit will converge to the uniform Haar measure.  The advantage of considering a random circuit rather than a random unitary on the whole system is it is naturally efficient to implement, for polynomial length circuits.  Random circuits of some fixed length are also a new measure on the unitary group which, as we show later, reproduces some of the properties of the Haar measure for polynomial length.  As well as the computer science aspects, this has applications in physics since randomly interacting systems could be modelled as a random circuit.  These systems will only reach their equilibrium if the random circuit converges quickly.  Thus proving convergence of the random circuit shows that some physical systems will have some properties of Haar random systems after evolving for a short amount of time.

We consider a general class of random circuits where a series of
two-qubit gates are chosen from a universal gate set.  We give a
framework for analysing the $k^{\text{th}}$ moments of these circuits.
Our conjecture, based on an analogous classical result
\cite{HooryBrodsky04}, is that a random circuit on $n$ qubits of
length $\poly(n,k)$ is an approximate $k$-design.  While we do not
prove this, we instead give a tight analysis of the $k=2$ case.  We
find that in a broad class of natural random circuit models (described
in \secref{RandomCircuits}), a circuit of length $O(n(n+\log 1/\eps))$
yields an $\eps$-approximate 2-design.  The approximate design definition used in this section is the diamond-norm definition given in \defref{ApproxUnitaryDesignDiamond} and, through \lemref{ApproxUnitaryDesignEquiv}, applies to the alternative definitions given above.  Moreover, our results also apply to random
stabiliser circuits, meaning that a random stabiliser circuit of
length $O(n(n+\log 1/\eps))$ will be an
$\eps$-approximate 2-design.  This both simplifies the construction
and tightens the efficiency of the approach of \cite{DLT02}, which
constructed $\eps$-approximate 2-designs in time $O(n^6(n^2+\log
1/\eps))$ using $O(n^3)$ elementary quantum gates.

\subsection{Random Circuits}
\label{sec:RandomCircuits}

The random circuit we will use is the following.  Choose a 2-qubit
gate set that is universal on $U(4)$ (or on the stabiliser subgroup of
$U(4)$).  One example of this is the set of all one qubit gates
together with the controlled-NOT gate.  Another is simply the set of
all of $U(4)$.  Then, at each step, choose a random pair of qubits and
apply a gate from the universal set chosen uniformly at random.  For
the $U(4)$ case, the distribution will be the Haar measure on $U(4)$.
One such circuit is shown in Fig.~\ref{figRandomCircuit} for $n=4$
qubits.  This is based on the approach used in \cite{ODP06,DOP07} but our analysis is both simpler and more general.

\begin{figure}[h]
  \begin{center}
    \includegraphics[width=12cm]{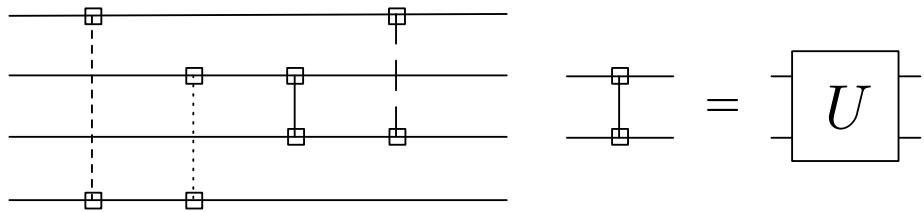}
    \caption[An example of a random circuit]{An example of a random circuit.  Different lines indicate
      a different gate is applied at each step.}
    \label{figRandomCircuit}
  \end{center}
\end{figure}
Since the universal set can generate the whole of $U(2^n)$ in this
way, such random circuits can produce any unitary.  Further, since
this process converges to a unitarily invariant distribution and the
Haar distribution is unique, the resulting unitary must be uniformly
distributed amongst all unitaries \cite{ELL05}.  Therefore this
process will eventually converge to a Haar distributed unitary from
$U(2^n)$.  This is proven rigorously in \lemref{GEigenvalues}.  
However, since a Haar unitary cannot be produced in polynomial time, this process will not converge in polynomial time.  
We address this problem by considering only the lower-order moments of the
distribution and showing these are nearly the same for random circuits
as for Haar-distributed unitaries.  This claim is formally described
in \thmref{Main2Design}.

This chapter is organised as follows.  In \secref{Preliminaries} we explain how a random circuit could be used to construct a $k$-design.  We then summarise the results of this chapter in \secref{RandomCircuitsResults}.  In \secref{Moments} we work out how the state evolves after a single step of the random circuit.  We then extend this to multiple steps in \secref{Convergence} and prove our general convergence results.  A key simplification will be (following \cite{ODP06}) to map the evolution of the second moments of the quantum circuit onto a classical Markov chain.  We then prove a tight convergence result for the case where the gates are chosen from $U(4)$ in \secref{U4Convergence}.  This section contains most of the technical content of the chapter.  Using our bounds on mixing time we put together the proof that random circuits yield approximate unitary 2-designs in \secref{MainResult}.  \secref{ConclusionRandomCircuits} concludes with some discussion of applications.

The majority of this chapter, with the exception of \secref{ZeroChainGap}, has been published previously as \cite{RandomCircuits} and is joint work with Aram Harrow.

\section{Preliminaries}
\label{sec:Preliminaries}
\subsection{Pauli expansion}
Much of the following will be done in the Pauli basis.  In this chapter, we choose the normalisation so that $\rho$ is written in the Pauli basis as
\be
\rho = 2^{-n/2} \sum_p \gamma(p) \sigma_p.
\ee
With this normalisation,
\be
\sum_p \gamma(p)^2 = \tr \rho^2
\ee
which is 1 for pure $\rho$.  In general,
\bes
\sum_p \gamma^2(p) \le 1
\ees
with equality if and only if $\rho$ is pure.  Note also that $\tr \rho = 1$ is equivalent to $\gamma(0) = 2^{-n/2}$.

This notation is extended to states on $nk$ qubits by treating $\gamma$ as a
function of $k$ strings from $\{0,1,2,3\}^n$.  Thus a 
 state $\rho$ on $nk$ qubits is written as
\begin{equation}
\label{eq:PauliBasisGeneral}
\rho = 
2^{-nk/2} \sum_{p_1, \ldots, p_k} \gamma_0(p_1, \ldots, p_k) \sigma_{p_1} \otimes \ldots \otimes \sigma_{p_k}.
\end{equation}

\subsection{Random Circuits as \texorpdfstring{$k$}{k}-designs}

If a random circuit is to be an approximate $k$-design then
\eq{ApproxUnitaryDesignDiamond} must be satisfied where the unitaries in $\cG_\nu$
are the different possible random circuits.  We can think of this as
applying the random circuit not once but $k$ times to $k$ different
systems.

Suppose that applying $t$ random gates yields the random
circuit $W$.  If $W^{\ot k}$ acts on an $nk$-qubit state $\rho$, then
the resulting state is 
\begin{equation}
\label{eqn:StateEvolutionPauliBasis}
\rho_W := W^{\ot k}\rho(W^\dag)^{\ot k} =
2^{-nk/2} \sum_{p_1, \ldots, p_k} \gamma_0(p_1, \ldots, p_k) W
\sigma_{p_1} W^{\dagger} \otimes \ldots \otimes W \sigma_{p_k}
W^{\dagger}. 
\end{equation}

For this to be a $k$-design, the expectation over
all choices of random circuit should match the expectation over
Haar-distributed $W\in U(2^n)$.

We are now ready to state our main results.  Our results apply to a
large class of gate sets which we define below: 
\begin{definition}
\label{def:2copy-gapped}
Let $\cE=\{p_i,U_i\}$ be a discrete ensemble of elements from $U(d)$.
Define an operator $G_\cE$ by
\be G_\cE := \sum_i p_i U_i^{\ot k,k}
\label{eq:gapped-condition}\ee
where $U^{\ot k,k} = U^{\ot k} \ot \l( U^* \r)^{\ot k}$.
More generally, we can consider continuous distributions.  If $\mu$ is
a probability measure on $U(d)$ then we can define $G_\mu$ by analogy
as
\be G_\mu := \int_{U(d)} d\mu(U) U^{\ot k,k}
\label{eq:gapped-condition2}\ee
Then
$\cE$ (or $\mu$) is $k$-copy gapped if $G_\cE$ (or $G_\mu$)
has only $k!$ eigenvalues with absolute value equal to $1$.
\end{definition}
For any discrete ensemble $\cE=\{p_i,U_i\}$, we can define a measure
$\mu=\sum_i p_i \delta_{U_i}$.  Thus, it suffices to state our theorems in
terms of $\mu$ and $G_\mu$.  We also remark that the $k$-copy gapped property is the same as the $k$-tensor product expander property for any non-zero gap as defined in \chapref{TPE}.

The condition on $G_\mu$ in the above definition may seem somewhat
strange.  We will see in \secref{Moments} that when $d\geq k$ there is
a $k!$-dimensional subspace of $(\bbC^d)^{\ot 2k}$ that is acted upon
trivially by any $G_\mu$.  Additionally, when $\mu$ is the Haar
measure on $U(d)$ then $G_\mu$ is the projector onto this space.
Thus, the $k$-copy gapped condition implies that vectors orthogonal to
this space are shrunk by $G_\mu$.

We will see that $G_\mu$ is $k$-copy gapped in a  number of important cases.
First, we give a definition of universality that can apply not only to
discrete gates sets, but to arbitrary measures on $U(4)$.

\begin{definition}\label{def:dist-universal}
Let $\mu$ be a distribution on $U(4)$.  Suppose that for any open ball
$S\subset U(4)$ there exists a positive integer $\ell$
such that $\mu^{\star \ell}(S)>0$.  Then we say $\mu$ is universal
[for $U(4)$].
\end{definition}
  Here $\mu^{\star \ell}$ is the $\ell$-fold convolution
of $\mu$ with itself; i.e.
$$\mu^{\star \ell}= \int \delta_{U_1\cdots U_\ell} 
d\mu(U_1)\cdots d\mu(U_\ell).$$
When $\mu$ is a discrete distribution over a set $\{U_i\}$,
\defref{dist-universal} is equivalent to the usual definition of
universality for a finite set of unitary gates.

\begin{theorem}\label{thm:kCopyGappedExamples}
The following distributions on $U(4)$ are $k$-copy gapped:
\begin{itemize}
\item[(i)]{Any universal gate set.  Examples are $U(4)$ itself, any
    entangling gate together with all single qubit gates, or the gate
    set considered in \cite{ODP06}.}
\item[(ii)]{Any approximate (or exact) unitary $k$-design on 2 qubits, such as the
    uniform distribution over the 2-qubit Clifford group, which is an exact 2-design.} 
\end{itemize}
\end{theorem}
\begin{proof}
\mbox{}
\begin{itemize}
\item[(i)]{This is proven in \lemref{GEigenvalues}.}
\item[(ii)]{This follows straight from \defref{UnitaryDesign}.\qedhere}
\end{itemize}
\end{proof}

\section{Summary of Results}
\label{sec:RandomCircuitsResults}

\begin{theorem}
\label{thm:Main2Design}
Let $\mu$ be a 2-copy gapped distribution and $W$ be a random circuit on $n$ qubits
obtained by drawing $t$ random unitaries according to $\mu$ and
applying each of them to a random pair of qubits.  Then there exists $C$ (depending only on $\mu$) such that for any $\eps>0$ and any $t \geq C(n(n+\log 1/\eps))$, $\cG_W$ is an
  $\eps$-approximate unitary 2-design according to either 
  \defref{ApproxUnitaryDesignDiamond} (DIAMOND) or \defref{ApproxUnitaryDesignDankert} (TWIRL).
\end{theorem}

To prove \thmref{Main2Design}, we show that the second
moments of the random circuits converge quickly to those of a uniform
Haar distributed unitary.  For $W$ a circuit as in \thmref{Main2Design}, write $\gamma_W(p_1, p_2)$ for the Pauli coefficients of $\rho_W = W^{\ot 2} \rho \l(W^\dagger\r)^{\ot 2}$.  Then write $\gamma_t(p_1, p_2) = \Expect_W \gamma_W(p_1, p_2)$ where $W$ is a circuit of length $t$.  Then we have
\begin{lemma}
\label{lem:MainMixing}
Let $\mu$ and $W$ be as in \thmref{Main2Design}.  Let the initial state be $\rho$ with $\gamma_0(p,p) \ge 0$ and $\sum_p \gamma_0(p,p) = 1$ (for example the state $\proj{\psi} \otimes \proj{\psi}$ for any pure state $\ket{\psi}$).  Then there
exists a constant $C$ (possibly depending on $\mu$) such that for
any $\eps>0$
\begin{itemize}
\item[(i)]{\begin{equation}
\sum_{p_1, p_2 \atop p_1 p_2 \ne 00} \left( \gamma_t(p_1,
    p_2) - \delta_{p_1 p_2} \frac{1}{2^n(2^n+1)} \right)^2 \le
\eps 
\end{equation}
for $t \ge Cn \log 1/\eps$.}
\item[(ii)]{\begin{equation}
\sum_{p_1, p_2 \atop p_1 p_2 \ne 00} \left| \gamma_t(p_1, p_2)
  - \delta_{p_1 p_2} \frac{1}{2^n(2^n+1)} \right| \le \eps
\end{equation}
for $t \ge Cn(n + \log 1/\eps)$ or, when $\mu$ is the uniform
distribution on $U(4)$ or its stabiliser subgroup, $t \ge Cn \log
\frac{n}{\eps}$.} 
\end{itemize}
\end{lemma}
We can then extend this to all states by a simple corollary:
\begin{corollary}
\label{cor:MainMixing}
Let $\mu$, $W$ and $\gamma_W$ be as in \lemref{MainMixing}.  Then, for any initial state $\rho = \frac{1}{2^n}\sum_{p_1,p_2} \gamma_0 (p_1,p_2) \sigma_{p_1} \ot \sigma_{p_2}$, there
exists a constant $C$ (possibly depending on $\mu$) such that for
any $\eps>0$
\begin{itemize}
\item[(i)]{\begin{equation}
\sum_{p_1, p_2 \atop p_1 p_2 \ne 00} \left( \gamma_t(p_1,
    p_2) - \delta_{p_1 p_2} \frac{\sum_{p \ne 0} \gamma_0 (p,p)}{4^n-1} \right)^2 \le
\eps 
\label{eq:2-norm-converge} 
\end{equation}
for $t \ge Cn(n+\log 1/\eps)$.}
\item[(ii)]{\begin{equation}
\sum_{p_1, p_2 \atop p_1 p_2 \ne 00} \left| \gamma_t(p_1, p_2)
  - \delta_{p_1 p_2} \frac{\sum_{p \ne 0} \gamma_0 (p,p)}{4^n-1} \right| \le \eps
\label{eq:1-norm-converge}
\end{equation}
for $t \ge Cn(n + \log 1/\eps)$.} 
\end{itemize}
\end{corollary}

By the diamond-norm definition of an approximate design (\defref{ApproxUnitaryDesignDiamond}), we only need 
convergence in the 2-norm \peq{2-norm-converge}, which is implied by
1-norm convergence \peq{1-norm-converge} but 
weaker.  However, \defref{ApproxUnitaryDesignDankert} (TWIRL), which requires the map to be close to the
twirling operation, requires 1-norm convergence
(i.e. \eq{1-norm-converge}).  Thus, \thmref{Main2Design} for \defref{ApproxUnitaryDesignDiamond} (DIAMOND)
follows from \corref{MainMixing}(i) 
and \thmref{Main2Design} for \defref{ApproxUnitaryDesignDankert} (TWIRL) follows from \corref{MainMixing}(ii).
\thmref{Main2Design} is proved in \secref{MainResult} and \corref{MainMixing} in \secref{Convergence}.

We note that we do not need to separately prove the result for \defref{ApproxUnitaryDesignDankert} (TWIRL) since the result follows from the equivalence of the $k$-design definitions (\lemref{ApproxUnitaryDesignEquiv}).  However, we include the proof since, if our bounds were improved to show convergence in $O(n \log \frac{n}{\eps})$ time, if we simply applied \lemref{ApproxUnitaryDesignEquiv}, this would only imply that $O(n(n+ \log 1/\eps)$ time was needed for \defref{ApproxUnitaryDesignDankert} (TWIRL).

We also emphasise that, in the course of proving \lemref{MainMixing}, we prove that the eigenvalue gap (defined in \secref{MarkovChain}) of the Markov chain that gives the evolution of the $\gamma(p,p)$ terms is $O(1/n)$.  It is easy to show that this bound is tight for some gate sets.

{\em Related work:}  Here we compare our work with other related results and efficient
constructions of approximate unitary 2-designs.
\begin{itemize}
\item The uniform distribution over the Clifford group on $n$ qubits
  is an exact 2-design \cite{DLT02}.  Moreover, \cite{DLT02} described
  how to sample from the Clifford group using $O(n^8)$ classical gates
  and $O(n^3)$ quantum gates.  Our results show that applying
  $O(n(n+\log 1/\eps))$ random two-qubit Clifford gates also achieve
  an $\eps$-approximate 2-design (although not necessarily a
  distribution that is within $\eps$ of uniform on the Clifford
  group).
\item Dankert et al.~\cite{DCEL06} gave a specific circuit
  construction of an approximate 2-design.  To achieve small error in
  the sense of \defref{ApproxUnitaryDesignDiamond} (DIAMOND), their circuits require
  the same
  $O(n(n+\log 1/\eps))$ gates that our random circuits do.  However,
  when we use 
  \defref{ApproxUnitaryDesignDankert} (TWIRL), the circuits from \cite{DCEL06} only
  need $O(n \log 1/\eps)$ gates while we only show that random circuits of
  length $O(n(n+\log 1/\eps))$ suffice.
\item The closest results to our own are in the papers by Oliveira et
  al.~\cite{ODP06,DOP07}, which considered a specific gate set
  (random single qubit gates and a controlled-NOT) and proved that the
  second moments converge in time $O(n^2(n+ \log 1/\eps))$.  Our
  strategy of analysing random
  quantum circuits in terms of classical Markov chains is also adapted
  from \cite{ODP06,DOP07}.  In \secref{Moments}, we generalise this approach
   to analyse the $k^{\text{th}}$ moments for arbitrary $k$.

  Our main results extend the results of  \cite{ODP06,DOP07} to a larger class of gate sets and improve their
   convergence bounds.  Some of these improvements have been conjectured by
   \cite{Znidaric07}, where the author presented numerical evidence in support
   of them.
\item
   An algorithmic application of random circuits was given in
   \cite{HH08}, where they were used to construct a new class of
   superpolynomial quantum speedups.  In that paper, random circuits of
   length $O(n^3)$ were used in order to guarantee that they were
   so-called ``dispersing'' circuits.  Our results immediately imply that
   circuits of length $O(n^2)$ would instead suffice.   We believe that
   this could be further improved with a specialised argument, since
   \cite{HH08} assumed that the input to the random circuit was always a
   computational basis state.
\end{itemize}

\section{Analysis of the Moments}
\label{sec:Moments}

In order to prove our results, we need to understand how the state
evolves after each step of the random circuit.  
%We need to average over the choices of pairs and gates at each step to
%find out how the expected coefficients evolve.
In this section we consider just one step and a fixed pair of qubits.
Later on we will extend this to prove convergence results for multiple
steps with random pairs of qubits drawn at every step.  We consider first the Haar distribution over the full unitary group
and then will discuss the more general case of any 2-copy gapped
distribution. 

In this section, we work in general dimension $d$ and with a general
Hermitian orthogonal basis $\sigma_0,\ldots,\sigma_{d^2-1}$.  Later we
will take $d$ to be either 4 or $2^n$ and the $\sigma_i$ to be Pauli
matrices.  However, in this section we keep the discussion general to
emphasise the potentially broader applications.

Fix an orthonormal basis for $d\times d$ Hermitian matrices:
$\sigma_0, \ldots, \sigma_{d^2-1}$, normalised so that $\tr
\sigma_p\sigma_q = d\,\delta_{p,q}$.   Let $\sigma_0$ be the identity.
We need to evaluate the quantity 
\begin{equation}
\label{eq:Tk}
\Expect_U \left(U^{\otimes k} \sigma_{p_1} \otimes \ldots \otimes \sigma_{p_k} (U^\dagger)^{\otimes k}\right) =: T(\bfp)
\end{equation}
where the expectation is over Haar distributed $U \in U(d)$.  We will need this quantity in two cases.  Firstly, for $d=2^n$, these are the moments obtained after applying a uniformly distributed unitary so we know what the random circuit must converge to.  Secondly, for $d=4$, this tells us how a random $U(4)$ gate acts on any chosen pair.

Call the quantity in \eq{Tk} $T(\bfp)$ (we use \textbf{bold} to indicate a $k$-tuple of coefficients; take $\bfp = (p_1, \ldots, p_k)$) and write it in the $\sigma_p$ basis as
\begin{equation}
T(\bfp) = \sum_{\bfq} \hat{G}(\bfq; \bfp) \sigma_{q_1} \otimes \ldots \otimes \sigma_{q_k}.
\end{equation}
Here, $\hat{G}(\bfq; \bfp)$ is the coefficient in the Pauli expansion of $T(\bfp)$ and we define $\hat{G}$ as the matrix with entries equal to $\hat{G}(\bfq; \bfp)$.  We have left off the usual normalisation factor because, as we shall see, with this normalisation $\hat{G}$ is a projector.  Inverting this, we have
\begin{align}
\hat{G}(\bfq; \bfp) 
&= d^{-k} \tr \left(\sigma_{q_1} \otimes \ldots \otimes \sigma_{q_k}
  T(\bfp)\right) \nonumber\\
& = d^{-k} \bbE_U \tr \left(  (\sigma_{q_1}\ot \cdots \ot \sigma_{q_k}) 
U^{\ot  k}  (\sigma_{p_1}\ot \cdots \ot \sigma_{p_k}) (U^\dag)^{\ot k}\right)
\label{eq:G}\end{align}
Note that $\hat{G}$ is real since $T$ and the basis are Hermitian.

We can gain all the information we need about the Haar integral in \eq{Tk} with the following observations:
\begin{lemma}
\label{lem:HaarIntegralCommutes}
$T(\bfp)$ commutes with $U^{\otimes k}$ for any unitary $U$.
\end{lemma}
\begin{proof}
Follows from the invariance of the Haar measure on the unitary group.
\end{proof}

\begin{corollary}
\label{cor:IsPerms}
$T(\bfp)$ is a linear combination of permutations from the symmetric group $S_k$.
\end{corollary}
\begin{proof}
This follows from Schur-Weyl duality (see e.g.~\cite{GoodmanWallach98}).
\end{proof}

From this, we can prove that $\hat{G}$ is a projector and find its eigenvectors.
\begin{theorem}
\label{thm:SymmetryG}
$\hat{G}$ is symmetric, i.e. $\hat{G}(\bfq; \bfp) = \hat{G}(\bfp; \bfq)$.
\end{theorem}
\begin{proof}
Follows from the invariance of the trace under cyclic permutations.
\end{proof}

\begin{theorem}
\label{thm:EigenvectorsG}
$S(\pi)$ is an eigenvector of $\hat{G}$ with eigenvalue $1$ for any subsystem permutation operator $S(\pi)$ i.e.
\begin{equation*}
\sum_{\bfq} \hat{G}(\bfp; \bfq) \tr (\sigma_{q_1} \otimes \ldots \otimes \sigma_{q_k} S(\pi)) = \tr (\sigma_{p_1} \otimes \ldots \otimes \sigma_{p_k} S(\pi)).
\end{equation*}
Further, any vector orthogonal to this set has eigenvalue $0$.
\end{theorem}
\begin{proof}
For the first part,
\begin{align}
\sum_{\bfq} &\hat{G}(\bfp; \bfq)
\tr (\sigma_{q_1} \otimes \ldots \otimes \sigma_{q_k} S(\pi)) \nn
&= d^{-k} \sum_{\bfq} \mathbb{E}_U \tr \left(\sigma_{q_1} U \sigma_{p_1} U^\dagger\right) \ldots 
\tr \l(\sigma_{q_k} U
  \sigma_{p_k} U^\dagger\right) \tr\left(\sigma_{q_1} \otimes \ldots
  \otimes \sigma_{q_k} S(\pi) \right)\nn
&= d^{-k} \tr \left( S(\pi) \mathbb{E}_U \sum_{q_1} \tr \left(\sigma_{q_1} U
  \sigma_{p_1} U^\dagger\right) \sigma_{q_1} \otimes \ldots 
\otimes \sum_{q_k} \tr
\left(\sigma_{q_k} U \sigma_{p_k} U^\dagger\right) \sigma_{q_k} \r)
\label{eq:G-perm-evector}
\end{align}
Writing $U^\dagger \sigma_p U$ in the $\sigma_p$ basis, we find
\begin{equation*}
\frac{1}{d} \sum_q \tr\left( \sigma_q U \sigma_p U^\dagger \right) \sigma_q = U \sigma_p U^\dagger.
\end{equation*}
Therefore \eq{G-perm-evector} becomes %(using \lemref{TraceCycles})
\begin{equation*}
\tr \left( S(\pi) \mathbb{E}_U U^\dagger \sigma_{p_1} U \otimes \ldots \otimes U^\dagger \sigma_{p_k} U \right) =\tr \left( \sigma_{p_1} \otimes \hdots \otimes \sigma_{p_k} S(\pi) \right).
\end{equation*}
For the second part, consider any vector $v$ which is orthogonal to the permutation operators (we can neglect the complex conjugate because $S(\pi)$ is real in this basis), i.e.
\begin{equation}
\label{eq:OrthToPerms}
\sum_{\bfq} \tr \left( \sigma_{q_1} \otimes \ldots \otimes \sigma_{q_k} S(\pi) \right) v(\bfq) = 0
\end{equation}
for any permutation $\pi$.  Then
\begin{equation*}
\sum_{\bfq} \hat{G}(\bfp; \bfq) v(\bfq) = d^{-k} \sum_{\bfq} \tr \left(\sigma_{q_1} \otimes \ldots \otimes \sigma_{q_k} T(\bfp) \right) v(\bfq)
\end{equation*}
which is zero since $T(\bfp)$ is a linear combination of permutations and $v$ is orthogonal to this by \eq{OrthToPerms}.
\end{proof}

\begin{theorem}
$\hat{G}^2 = \hat{G}$, i.e. $\sum_{\mathbf{q'}} \hat{G}(\bfp; \mathbf{q'}) \hat{G}(\mathbf{q'}; \bfq) = \hat{G}(\bfp; \bfq)$.
\end{theorem}
\begin{proof}
Using \eq{G},
\begin{equation*}
\sum_{\mathbf{q'}} \hat{G}(\bfp; \mathbf{q'}) \hat{G}(\mathbf{q'}; \bfq) = \sum_{\mathbf{q'}} \hat{G}(\bfp; \mathbf{q'}) d^{-k} \tr \left(\sigma_{q'_1} \otimes \ldots \otimes \sigma_{q'_k} T(\bfq)\right).
\end{equation*}
From \corref{IsPerms}, $T(\bfq)$ is a linear combination of permutations.  This implies, using \thmref{EigenvectorsG} that
\begin{align*}
\sum_{\mathbf{q'}} \hat{G}(\bfp; \mathbf{q'}) d^{-k} \tr \left(\sigma_{q'_1} \otimes \ldots \otimes \sigma_{q'_k} T(\bfq)\right) &= d^{-k} \tr \left(\sigma_{p_1} \otimes \ldots \otimes \sigma_{p_k} T(\bfq)\right) \\
&= \hat{G}(\bfp; \bfq)
\end{align*}
as required.
\end{proof}

\begin{corollary}
$\hat{G}$ is a projector so has eigenvalues $0$ and $1$.
\end{corollary}

We now evaluate $\hat{G}$ and $T$ for the cases of $k=1$ and $k=2$ since these are the cases we are interested in for the remainder of the chapter.

\subsection{\texorpdfstring{$k=1$}{k=1}}

The $k=1$ case is clear: the random unitary completely randomises the state.  Therefore all terms in the expansion are set to zero apart from the identity i.e.
\begin{equation}
T(p) =
\begin{cases}
\sigma_0	&	p=0\\
0	&	p \ne 0.
\end{cases}
\end{equation}

\subsection{\texorpdfstring{$k=2$}{k=2}}
\label{sec:GforK2}

For $k=2$, there are just two permutation operators, identity $I$ and swap $\mathcal{F}$.  Therefore there are just two eigenvectors with non-zero eigenvalue ($n > 1$).  In normalised form, taking them to be orthogonal, their components are
\begin{align*}
f_1(q_1, q_2) &= \delta_{q_1 0} \delta_{q_2 0} \\
f_2(q_1, q_2) &= \frac{1}{d^2-1} \delta_{q_1 q_2} (1 - \delta_{q_1 0})
\end{align*}
We will now prove three properties of $\hat{G}$ that we need:
\begin{enumerate}
\item{$\hat{G}(p_1, p_2; q_1, q_2) = 0$ if $p_1 \ne p_2$ or $q_1 \ne q_2$.
\begin{proof}
Consider the function $f(q_1, q_2) = \delta_{q_1 a} \delta_{q_2 b}$ with $a \ne b$.  This function has zero overlap with the eigenvectors $f_1$ and $f_2$ so it goes to zero when acted on by $\hat{G}$.  Therefore $\hat{G}(p_1, p_2; a, b) = 0$.  The claim follows from the symmetry property (\thmref{SymmetryG}).
\end{proof}}
With this we will write $\hat{G}(p; q) \equiv \hat{G}(p_1, p_2; q_1, q_2)$.
\item{$\hat{G}(p; 0) = \delta_{p 0}$.
\begin{proof}
Let $\hat{G}$ act on eigenvector $f_1$.
\end{proof}}
\item{$\hat{G}(p; a) = \frac{1}{d^2-1}$ for $a, p \ne 0$.
\begin{proof}
Let $\hat{G}$ act on the input $\delta_{q a}$.  This has zero overlap with $f_1$ and overlap $\frac{1}{d^2-1}$ with $f_2$.
\end{proof}}
\end{enumerate}
Therefore we have
\begin{equation}
\hat{G}(p_1, p_2; q_1, q_2) =
\begin{cases}
0	& p_1 \ne p_2 \rm{\, or\,} q_1 \ne q_2 \\
1	& p_1 = p_2 = q_1 = q_2 = 0 \\
\frac{1}{d^2-1}	& p_1 = p_2 \ne 0, q_1 = q_2 \ne 0 \\
\end{cases}
\end{equation}

Since $T(p_1, p_2) = \sum_{q_1, q_2} \hat{G}(p_1, p_2; q_1, q_2) \sigma_{q_1} \otimes \sigma_{q_2}$, we have
\begin{equation}
T(p_1, p_2) =
\begin{cases}
0	&	p_1 \ne p_2 \\
\sigma_0 \otimes \sigma_0	&	p_1 = p_2 = 0 \\
\frac{1}{d^2-1} \sum_{p' \ne 0} \sigma_{p'} \otimes \sigma_{p'}	&	p_1 = p_2 \ne 0.
\end{cases}
\end{equation}
Therefore the terms $\sigma_{p_1} \otimes \sigma_{p_2}$ with $p_1 \ne p_2$ are set to zero.  Further, the sum of the diagonal coefficients $\gamma(p, p)$ is conserved.  This allows us to identify this with a probability distribution (after renormalising) and use Markov chain analysis.  To see this, write again the starting state
\begin{equation*}
\rho = \frac{1}{d} \sum_{q_1, q_2} \gamma_0(q_1, q_2) \sigma_{q_1} \otimes \sigma_{q_2}
\end{equation*}
with state after application of any unitary $W$
\begin{equation*}
\rho_W = \frac{1}{d} \sum_{q_1, q_2} \gamma_W(q_1, q_2) \sigma_{q_1} \otimes \sigma_{q_2} = 2^{-n} \sum_{q_1, q_2} \gamma(q_1, q_2) \left(W \sigma_{q_1} W^{\dagger}\right) \otimes \left(W \sigma_{q_2} W^{\dagger}\right).
\end{equation*}
Then
\begin{align*}
\sum_q \gamma_W(q, q) &= \frac{1}{d} \sum_q \tr \left(\sigma_q \otimes \sigma_q \rho_W\right) \\
&= \tr \left(\swap \rho_W\right) \\
&= \frac{1}{d} \sum_{q_1, q_2} \gamma(q_1, q_2) \tr \left( \swap \left(W \sigma_{q_1} W^{\dagger}\right) \otimes \left(W \sigma_{q_2} W^{\dagger}\right) \right)\\
&= \frac{1}{d} \sum_{q_1, q_2} \gamma(q_1, q_2) \tr \left( \sigma_{q_1} \sigma_{q_2} \right)\\
&= \sum_q \gamma(q, q)
\end{align*}
as required, where $\swap$ is the swap operator and we have used
Lemmas \ref{lem:Swap} and \ref{lem:TraceCycles}. 

\subsection{Moments for General Universal Random Circuits}
\label{sec:MomentsGeneral}

We now consider universal distributions $\mu$ that in general may be
different from the uniform (Haar) measure on $U(d)$.  Our main result in this section
will be to show that a universal distribution on $U(4)$ is also 2-copy
gapped.  In fact, we will phrase this result in slightly more general
terms and show that a universal distribution on $U(d)$ is also
$k$-copy gapped for any $k$.  Universality (\defref{dist-universal}) generalises in
the obvious way to $U(d)$, whereas when we say that $\mu$ is $k$-copy
gapped, we mean that \be \|G_\mu - G_{U(d)}\|_\infty < 1,
\label{eq:k-copy-gapped}\ee
where $G_{?}=\bbE_U U^{\ot k,k}$, with the
expectation taken over $\mu$ for $G_\mu$ or over the Haar measure for
$G_{U(d)}$.  

The reason \eq{k-copy-gapped} represents our condition for $\mu$ to be
$k$-copy gapped is as follows:
Observe that $\hat{G}$ and $G$ are unitarily related, so the
definition of $k$-copy gapped could equivalently be given in terms of
$\hat{G}$.  We have shown above that $\hat{G}_{U(d)}$ (and thus
$G_{U(d)}$) has all eigenvalues equal to $0$ or $1$ i.e. it is a projector.  By
contrast, $G_\mu$ may not even be Hermitian.  However, we will prove
below that all eigenvectors of ${G}_{U(d)}$ with eigenvalue 1 are also
eigenvectors of ${G}_\mu$ with eigenvalue 1.  Thus, \eq{k-copy-gapped}
will imply that $\lim_{t\ra\infty} (\hat{G}_\mu)^t= \hat{G}_{U(d)}$,
just as we would expect for a gapped random walk.

 We would like to show that \eq{k-copy-gapped} holds whenever
$\mu$ is universal.  This result was proved in \cite{ArnoldKrylov62}
(and was probably known even earlier)
when $\mu$ had the form $(\delta_{U_1}+\delta_{U_2})/2$.  Here we show how
to extend the argument to any universal $\mu$.

\begin{lemma}
\label{lem:GEigenvalues}
Let $\mu$ be a distribution on $U(d)$. Then all eigenvectors of $G_{U(d)}$ with
eigenvalue 1 are eigenvectors of $G_\mu$ with eigenvalue 1.
Additionally, if $\mu$ is universal then $\mu$ is $k$-copy gapped for
any positive integer $k$ (cf. \eq{k-copy-gapped}).
\end{lemma}
In particular, if $k=2$ this Lemma implies that $\mu$ is 2-copy gapped
(cf. \thmref{kCopyGappedExamples}).
\begin{proof}
Let $V\cong \bbC^{d}$ be the fundamental representation of $U(d)$,
where the action of $U\in U(d)$ is simply $U$ itself.  Let $V^*$ be
its dual representation, where $U$ acts as $U^*$.  The operators
$G_\mu$ and $G_{U(d)}$ act on the space $V^{\ot k} \ot (V^*)^{\ot
  k}$.  We will see that $G_{U(d)}$ is completely determined by the
decomposition of $V^{\ot k} \ot (V^*)^{\ot k}$ into irreducible representations (irreps).  
 Suppose that
the multiplicity of $(r_\lambda,V_\lambda)$ in $V^{\ot k} \ot
(V^*)^{\ot k}$ is $m_\lambda$,
where the $V_\lambda$'s are the irrep spaces and $r_\lambda(U)$ the
corresponding representation matrices.  In other words
\begin{align}
V^{\ot k} \ot (V^*)^{\ot k} 
& \cong \bigoplus_{\lambda} V_\lambda \ot \bbC^{m_\lambda} \\
U^{\ot k} \ot (U^*)^{\ot k}
& \sim \sum_{\lambda} \proj{\lambda} \ot r_\lambda(U) \ot
I_{m_\lambda}
\label{eq:irrep-basis}
\end{align}
Here $\sim$ indicates that the two sides are related by conjugation by
a fixed ($U$ independent) unitary.

Let
$\lambda=0$ denote the trivial irrep: i.e. $V_0=\bbC$ and $r_0(U)=1$
for all $U$. 
We claim that $\bbE_U r_\lambda(U)=0$ whenever 
$\lambda\neq 0$ and the expectation is taken over the Haar measure.
To show this, note that $\bbE_U
r_\lambda(U)$ commutes with $r_\lambda(V)$ for all $V\in U(d)$ and
thus, by Schur's Lemma, we must have $\bbE_U
r_\lambda(U)= cI$ for some $c\in \bbC$.  However, by the
translation-invariance of 
the Haar measure we have $cI = \bbE_U
r_\lambda(U) = \bbE_U
r_\lambda(UV) = c\, r_\lambda(V)$ for all $V\in U(d)$.  Since
$\lambda\neq 0$, we cannot have $r_\lambda(V)= I$ for all $V$ and so
it must be that $c=0$.

Thus, if we write $G_{U(d)}$ and  $G_\mu$ using the basis on the RHS of
\eq{irrep-basis}, we have
\be {G}_{U(d)}  = \proj{0} \ot I_{m_0}\ee
where $\proj{0}$ is a projector onto the trivial irrep.
On the other hand,
\be {G}_{\mu} = \proj{0} \ot I_{m_0} + \sum_{\lambda \ne 0} \proj{\lambda} \ot \l(\int r_\lambda(U) d\mu(U)\r) \ot I_{m_\lambda}
\ee
Thus, every eigenvector of ${G}_{U(d)}$ with eigenvalue one is also
fixed by ${G}_\mu$.  For the remainder of the space, the direct sum
structure means that
\be \|G_{U(d)} - G_\mu\|_\infty = \max_{\substack{\lambda\neq 0\\
m_\lambda\neq 0}} \l\|\int r_\lambda(U) d\mu(U)\r\|_\infty
.\ee
Note that this maximisation only includes $\lambda$ with $\dim V_\lambda>1$.  This is because non-trivial one-dimensional irreps of $U(d)$ have the form $\det U^m$ for some non-zero integer $m$.  Under the map $U\mapsto e^{i\phi}U$, such irreps pick up a phase of $e^{im\phi}$.  However, $U^{\ot k}\ot (U^*)^{\ot k}$ is invariant under $U\mapsto e^{i\phi}U$.  Thus $V^{\ot k} \ot (V^*)^{\ot k}$  cannot contain any non-trivial one-dimensional irreps.

Now suppose by
contradiction that there exists $\lambda\neq 0$ 
with $m_\lambda\neq 0$
and \bes \l|\l|\int r_\lambda(U) d\mu(U)\r|\r|_\infty=1.\ees  (We do not need to consider the case $\|\int
r_\lambda(U) d\mu(U)\|_\infty>1$, since
$\|r_\lambda(U)\|_\infty=1$ for all $U$ and $\|\cdot\|_\infty$ obeys the triangle inequality.)
Indeed, the triangle inequality further implies that  there exists a unit vector $\ket{v}\in
V_\lambda$ such that  
$$\int d\mu(U)\, r_\lambda(U)\ket{v} = \omega\ket{v},$$
for some $\omega\in\bbC$ with $|\omega|=1$. 

By the above argument we can assume that $\dim V_\lambda > 1$.  Since $V_\lambda$ is irreducible, it cannot contain a one-dimensional invariant subspace, implying that there exists $U_0\in U(d)$ such that 
$$|\bra{v}r_\lambda(U_0)\ket{v}| = 1-\delta,$$
for some $\delta>0$.   Since $U\mapsto |\bra{v}r_\lambda(U)\ket{v}|$ is
continuous, there exists an open ball $S$ around $U_0$ such that
$|\bra{v}r_\lambda(U)\ket{v}| \leq  1-\delta/2$ for all $U\in S$.  
Define $\bar{S} := U(d)\backslash S$.

Now we use the fact that $\mu$ is universal to find an $\ell$ such
that $\mu^{\star \ell}(S)>0$.  Next, observe that $\int d\mu^{\star
  \ell}(U)\, \bra{v}r_\lambda(U)\ket{v} = \omega^\ell$.  Taking the
absolute value of both sides yields
\begin{align*}
1 & = 
\l| \int_{U(d)} d\mu^{\star \ell}(U)\, \bra{v}r_\lambda(U)\ket{v}\r|
\\ & \leq
 \int_{U(d)} d\mu^{\star \ell}(U)\, \l|\bra{v}r_\lambda(U)\ket{v}\r|
\\ & =
 \int_{S} d\mu^{\star \ell}(U)\, \l|\bra{v}r_\lambda(U)\ket{v}\r|
+ \int_{\bar{S}} d\mu^{\star \ell}(U)\, \l|\bra{v}r_\lambda(U)\ket{v}\r|
\\ & \leq
\mu^{\star \ell}(S)\l(1-\frac{\delta}{2}\r) + 
\l(1-\mu^{\star \ell}(S)\r)
\\ & < 1,
\end{align*}
a contradiction.  We conclude that $\|G_{U(d)} - G_\mu\|_\infty<1$.
\end{proof}

\section{Convergence}
\label{sec:Convergence}

In the previous section we saw that iterating any universal gate set on $U(d)$ eventually converges to the uniform distribution on $U(d)$.  Since the set of all two-qubit unitaries is universal on $U(2^n)$, this implies that random circuits eventually converge to the Haar measure.  In this section, we turn to proving upper bounds on this convergence rate, focusing on the first two moments.

Let $\hat{G}^{(ij)}$ be the matrix with $\hat{G}$ (with $d=4$) acting on qubits $i$ and $j$ and the identity on the others.  Then, if the pair $(i, j)$ is chosen at step $t$, we can find the expected coefficients at step $t+1$ by multiplying by $\hat{G}^{(ij)}$.  In general, a random pair is chosen at each step.  So
\be
\gamma_{t+1}(\mathbf{p}) = \sum_{\mathbf{q}} \frac{1}{n(n-1)} \sum_{i \ne j} \hat{G}^{(ij)}(\mathbf{p}; \mathbf{q}) \gamma_t(\mathbf{q})
\ee
where $\gamma_{t+1}$ are the expected coefficients at step $t$.  We can think of this evolution as repeated application of the matrix
\begin{equation}
\label{eq:GeneralTransitionMatrix}
P = \frac{1}{n(n-1)} \sum_{i \ne j} \hat{G}^{(ij)}.
\end{equation}

For $k=2$, the key idea of Oliveira et al.~\cite{ODP06} was to map the evolution of the $\gamma(p, p)$ coefficients to a Markov chain.  The $\gamma(p_1, p_2)$ coefficients with $p_1 \ne p_2$ just decay as each qubit is chosen and can be analysed directly.

However, we can only map the $\gamma(p, p)$ coefficients to a probability distribution when they are non-negative, which is not the case for general states.  Most of the rest of the chapter is dedicated to proving \lemref{MainMixing}, which only applies to states with $\gamma(p, p) \ge 0$ and normalised so their sum is $1$.  \corref{MainMixing} then extends this to all states:
\begin{proof}[Proof of \corref{MainMixing}]
\lemref{MainMixing} still applies to the $\gamma(p_1, p_2)$ terms with $p_1 \ne p_2$.  Therefore we just need to show how to apply \lemref{MainMixing} to states that initially have some negative $\gamma(p, p)$ terms.  

For the $\gamma(p, p)$ terms, \lemref{MainMixing} says that the random walk starting with any initial probability distribution converges to uniform in some bounded time $t$.  Let $g_t(p, p; q, q)$ be the coefficients after $t$ steps of the walk starting at a particular point $q$ (i.e.~$g_0(p, p;q,q) = \delta_{p,q}$).  Now, for any starting state $\rho$, let the initial coefficients be $\gamma_0(p,p)$.  Then, by linearity, we can write the expected coefficients after $t$ steps $\gamma_t(p, p) := \Expect \gamma_W(p, p)$ as
\be
\gamma_t(p,p) = \sum_{q\ne0} \gamma_0(q,q) g_t(p,p;q,q)
\ee
for $p\ne0$.

We can now prove convergence rates for the expected coefficients $\gamma_t(p, p)$:
\begin{itemize}
\item[(i)]{
For the 2-norm, we have from \lemref{MainMixing} that for $t \ge C n \log 1/\eps$
\be
\sum_{p \ne 0} \left( g_t(p, p; q, q) - \frac{1}{4^n-1}\right)^2 \le \eps
\ee
for any $q$.  Note that the normalisation for the $\gamma(p,p)$ terms with $p\ne0$ has changed from \lemref{MainMixing} since we are neglecting the $\gamma(0,0)$ term here.  Now
\begin{align*}
&\sum_{p \ne 0} \left( {\gamma}_t(p,p) - \frac{\sum_{q\ne0} \gamma_0(q,q)}{4^n-1}\right)^2 \\
&= \sum_{p \ne 0} \left( \sum_{q \ne 0} \gamma_0(q,q) \left(g_t(p,p; q,q) - \frac{1}{4^n-1}\right)\right)^2 \\
&\le \sum_{q \ne 0} \gamma_0(q,q)^2 \sum_{q' \ne 0} \sum_{p\ne0} \left(g_t(p,p;q',q') - \frac{1}{4^n-1}\right)^2 \\
&\le (4^n-1) \eps \sum_{q \ne 0} \gamma_0(q,q)^2 \\
&\le 4^n \eps \sum_{q_1, q_2} \gamma_0(q_1,q_2)^2 \\
&= 4^n \eps \, \tr \rho^2 \\
&\le 4^n \eps
\end{align*}
where the first inequality is the Cauchy-Schwarz inequality.  Therefore for $t \ge C n (n + \log 4^n/\eps)$, the 2-norm distance from stationarity for the $\gamma(p,p)$ terms is at most $\eps$.  Choose $C'$ such that $C' n (n + \log 1/\eps) \ge C n (n + \log 4^n/\eps)$ to obtain the result.
}
\item[(ii)]{
For the 1-norm, \lemref{MainMixing} says that for $t \ge C n (n + \log 1/\eps)$
\be
\sum_{p \ne 0} \left| g_t(q;p,p) - \frac{1}{4^n-1}\right| \le \eps.
\ee
We can then proceed much as for the 2-norm case:
\begin{align*}
&\sum_{p \ne 0} \left| \gamma_t(p,p) - \frac{\sum_{q\ne0} \gamma_0(q,q)}{4^n-1}\right| \\
&= \sum_{p \ne 0} \left| \sum_{q \ne 0} \gamma_0(q,q) \left(g_t(p,p;q,q) - \frac{1}{4^n-1}\right)\right| \\
&\le \sum_{q \ne 0} |\gamma_0(q,q)| \sum_{p\ne0} \left|g_t(p,p;q,q) - \frac{1}{4^n-1}\right| \\
&\le \eps \sum_{q \ne 0} |\gamma_0(q,q)| \\
&\le 2^n \eps \sum_{q \ne 0} \gamma_0^2(q,q) \\
&\le 2^n \eps.
\end{align*}
Therefore for $t \ge C n (n + \log 2^n/\eps)$, the 1-norm distance from stationarity for the $\gamma(p,p)$ terms is at most $\eps$.\qedhere
}
\end{itemize}
\end{proof}

We now proceed to prove \lemref{MainMixing}.  Firstly, we will consider the simple case of $k=1$ to prove this process forms a 1-design as this will help us to understand the more complicated case of $k=2$.

\subsection{First Moments Convergence}

Recall that $\rho = 2^{-n/2} \sum_p \gamma(p) \sigma_p$ and we wish to evaluate the moments of the coefficients.  So for the first moments to converge, we want to know $\Expect \gamma(p)$.

For $k=1$, the $U(4)$ random circuit uniformly randomises each pair that is chosen.  More precisely, a pair of sites $i, j$ are chosen at random and all the coefficients with $p_i \ne 0$ or $p_j \ne 0$ are set to zero.  Thus we get an exact 1-design when all sites have been hit.  For other gate sets, the terms do not decay to zero but decay by a factor depending on the gap of $\hat{G}$.  Call the gap $\Delta$; for $U(4)$ $\Delta=1$ and for others $0 < \Delta \le 1$ and $\Delta$ is independent of $n$.  Therefore once each site has been hit $m$ times the terms have decayed by a factor $(1-\Delta)^m$.

For a bound like the mixing time (see \secref{MarkovChain} for definition), we want to bound the quantity $\sum_{p \ne 0} | \Expect_W \gamma_W(p) |$ where $\gamma_W(p)$ is the Pauli coefficient after applying the random circuit $W$.  We also want 2-norm bounds, so we bound $\sum_{p \ne 0} (\Expect_W \gamma_W(p))^2$ too.  We will in fact find bounds on \bes \sum_{p \ne 0} \Expect_W | \gamma_W(p) | \ees and \bes \sum_{p \ne 0} \l(\Expect_W |\gamma_W(p)|\r)^2,\ees which are stronger.

A standard problem in the theory of randomised algorithms is the \emph{coupon collector problem}.  If a magazine comes with a free coupon, which is chosen uniformly randomly from $n$ different types, how many magazines should you buy to have a high probability of getting all $n$ coupons?  It is not hard to show that $n \ln \frac{n}{\eps}$ samples (magazines) have at least a $1-\eps$ probability of including all $n$ coupons.  Using this, we expect all sites to be hit with probability at least $1-\eps$ after $\Theta(n \log \frac{n}{\eps})$ steps.  This argument can be made precise in this context by bounding the non-identity coefficients.  We find, as expected, that the sum is small after $O(n \log n)$ steps:
\begin{lemma}
\label{lem:CoefficientDecayGeneral}
After $O(n \log 1/\eps)$ steps
\begin{equation*}
\sum_{p \ne 0} \left( \Expect_W |\gamma_W(p) | \right)^2 \le \eps
\end{equation*}
and after $O(n\log \frac{n}{\eps})$ steps,
\begin{equation}
\label{eq:TermsDecayGeneral}
\sum_{p \ne 0} \Expect_W | \gamma_W(p) | \le \eps.
\end{equation}
\end{lemma}
\begin{proof}
At each step, a pair of sites is chosen at random and any terms with non-identity coefficients for this pair decay by a factor $(1-\Delta)$.  For example, the term $\sigma_1 \otimes \sigma_0^{\otimes (n-1)}$ decays whenever the first site is chosen.  Thus the probability of each term decaying depends on the number of zeroes.  We start with the 1-norm bound.

Suppose the circuit applied after $t$ steps is $W_t$.  Consider $\Expect_{W_t} | \gamma_{W_t}(p) |$ for any $p$ with $d$ non-zeroes.  Since the state $\rho$ is physical, $\tr \rho^2 \le 1$ so $\sum_p \gamma^2_0(p)  \le 1$.  Now, in each step, if any site is chosen where $p$ is non-zero, this term decays by a factor $(1-\Delta)$.  This occurs with probability $1-\frac{(d-n)(d-n-1)}{n(n-1)} \ge d/n$, the probability of choosing a pair where at least one site is non-zero.  Therefore
\begin{equation*}
\Expect | \gamma_{W_t}(p) | \le \left( (1-\Delta)d/n + (1-d/n) \right) | \gamma_{W_{t-1}}(p) |
\end{equation*}
where the expectation is over the circuit applied at step $t$.  If we iterate this $t$ times we find
\begin{equation*}
\Expect_{W} | \gamma_{W}(p) | \le \exp(-\Delta t d/n) | \gamma_{0}(p) |
\end{equation*}
where the expectation here is over all random circuits for the $t$ steps.  We now sum over all $p$:
\begin{equation*}
\sum_{p \ne 0} \Expect_W | \gamma_W(p) | \le \sum_{d=1}^{n} \exp(-\Delta t d/n) \sum_{d(p) = d} | \gamma_{0}(p) |
\end{equation*}
where $d(p)$ is the number of non-zeroes in $p$.  For the 1-norm bound, we can simply bound $| \gamma_{0}(p) | \le 1$ to give $\sum_{d(p) = d} | \gamma_0(p) | \le {n \choose d} 3^d$ so
\begin{equation*}
\sum_{p \ne 0} \Expect_W | \gamma_W(p) | \le (1+3 \exp(-\Delta t /n))^n - 1
\end{equation*}
where we have used the binomial theorem.  Now let $t= \frac{n}{\Delta} \ln \frac{3n}{\eps}$.  This gives
\begin{equation*}
\sum_{p \ne 0} \Expect_W | \gamma_W(p) | \le (1+\eps/n)^n - 1 = O(\eps).
\end{equation*}
For the 2-norm bound,
\begin{align*}
\sum_{p \ne 0} (\Expect_W | \gamma_W(p) |)^2 \le& \sum_{p \ne 0} \exp(-2\Delta t d/n) \gamma^2_{0}(p) \\
=& \sum_{d=1}^n \exp(-2\Delta t d/n)\sum_{d(p) = d} \gamma^2_{0}(p) \\
\le& \sum_{d=1}^n \exp(-2\Delta t d/n) \\
\le& \frac{\exp(-2\Delta t /n)}{1-\exp(-2\Delta t /n)}
\end{align*}
where we have used $\sum_p \gamma^2_0(p) \le 1$.  We find after $\frac{n}{2\Delta} \ln 1/\eps$ steps that
\bes
\sum_{p \ne 0} (\Expect_W |\gamma_W(p)|)^2 \le \frac{\eps}{1-\eps}\qedhere
\ees
\end{proof}

\subsection{Second Moments Convergence}

Firstly, the $\sigma_{p_1} \otimes \sigma_{p_2}$ terms for $p_1 \ne p_2$ decay in a similar way to the non-identity terms in the 1-design analysis.  In fact, the proof of \lemref{CoefficientDecayGeneral} carries over almost identically to this case to give
\begin{lemma}
\label{lem:CoefficientsDecay2General}
After $O(n \log 1/\eps)$ steps
\begin{equation*}
\sum_{p_1 \ne p_2} (\Expect_W |\gamma_W(p_1, p_2)|)^2 \le \eps
\end{equation*}
and after $O(n(n+\log 1/\eps))$ steps
\begin{equation*}
\sum_{p_1 \ne p_2} \Expect_W | \gamma_W(p_1, p_2) | \le \eps.
\end{equation*}
\end{lemma}
\begin{proof}
Instead of the number of zeroes governing the decay rate, we need to count the number of places where $p_1$ and $p_2$ differ.  This gives
\begin{equation*}
\Expect | \gamma_{W_t}(p_1, p_2) | \le \left( (1-\Delta)d/n + (1-d/n) \right) | \gamma_{W_{t-1}}(p_1, p_2) |
\end{equation*}
where now $d$ is the number of differing sites.  There are ${n \choose d} 12^d 4^{n-d}$ states that differ in $d$ places so we find
\begin{equation*}
\sum_{p_1 \ne p_2} \Expect_W | \gamma_W(p_1, p_2) | \le 4^n[(1+3 \exp(-\Delta t /n))^n - 1].
\end{equation*}
Set $t =\frac{n}{\Delta} (n \ln 4 + \ln 1/\eps)$ to make this $O(\eps)$.  The 2-norm bound follows in the same way as for \lemref{CoefficientDecayGeneral}.
\end{proof}
We now need to prove the $\gamma(p, p)$ terms converge quickly.  We have seen above that the sum of the terms $\gamma(p, p)$ is conserved and, for the purposes of proving \lemref{MainMixing}, we assume the sum is $1$ and $\gamma(p,p) \ge 0$ for all $p$.

To illustrate the evolution, consider the simplest case when the gates are chosen from $U(4)$.  We have evaluated $\hat{G}$ in \secref{GforK2} for $k=2$ for this case.  Translated into coefficients this yields the following update rule, where we have written it for the case when qubits 1 and 2 are chosen:
\begin{multline}
\label{eq:FullChain}
\gamma_{t+1}(r_1, r_2, r_3, \ldots, r_n, s_1, s_2, s_3, \ldots, s_n) \\
=\begin{cases}
0	&	(r_1, r_2) \ne (s_1, s_2) \\
\gamma_{t}(0, 0, r_3, \ldots, r_n, 0, 0, s_3, \ldots, s_n)	&	(r_1, r_2) = (s_1, s_2) = (0, 0) \\
\frac{1}{15} \sum_{r'_1, r'_2 \atop r'_1 r'_2 \ne 0} \gamma_{t}(r'_1, r'_2, r_3, \ldots, r_n, r'_1, r'_2, s_3, \ldots, s_n) &	(r_1, r_2) = (s_1, s_2) \ne (0, 0).
\end{cases}
\end{multline}
The key idea of Oliveira et al.~\cite{ODP06} was to map the evolution of the $\gamma(p, p)$ coefficients to a Markov chain.  We can apply this here to get, on state space $\{0, 1, 2, 3\}^n$, the evolution:
\begin{enumerate}
\item{Choose a pair of sites uniformly at random.}
\item{If the state is $00$ it remains $00$.}
\item{Otherwise, choose the state uniformly at random from $\{0,1,2,3\}^2 \backslash \{00\}$.}
\end{enumerate}
This is the correct evolution since, if the initial state is distributed according to $\gamma_t(q, q)$, the final state is distributed according to $\gamma_{t+1}(p, p)$.

The evolution for other gate sets will be similar, but the states will not be chosen uniformly randomly in the third step.  However, the state $00$ will remain $00$ and the stationary distribution on the other 15 states is the same.  We will find the convergence times for general gate sets and then consider the $U(4)$ gate set since we can perform a tight analysis for this case.

\subsection{Markov Chain Analysis}
\label{sec:MarkovChain}

Before finding the convergence rate for our problem, we will briefly introduce the basics of Markov chain mixing time analysis.  All of these standard results can be found in \cite{MontenegroTetali06} and references therein.

A process is Markov if the evolution only depends on the current state rather than the full state history.  Therefore the evolution of the state can be thought of as a matrix, the \emph{transition matrix}, acting on a vector which represents the current distribution.  We will only be interested in discrete time processes so the state after $t$ steps is given by the $t^{\text{th}}$ power of the transition matrix acting on the initial distribution.

We say a Markov chain is \emph{irreducible} if it is possible to get from one state to any other state in some number of steps.  Further, a chain is \emph{aperiodic} if it does not return to a state at regular intervals.  If a chain is both irreducible and aperiodic then it is said to be \emph{ergodic}.  A well known result of Markov chain theory is that all ergodic chains converge to a unique stationary distribution.  In matrix language this says that the transition matrix $P$ has eigenvalue $1$ with no multiplicity and all other eigenvalues have absolute value strictly less than 1.  We will also need the notion of \emph{reversibility}.  A Markov chain is reversible if the time reversed chain has the same transition matrix, with respect to some distribution.  This condition is also known as \emph{detailed balance}:
\begin{equation}
\label{eq:DetailedBalance}
\pi(x) P(x, y) = \pi(y) P(y, x).
\end{equation}
It can be shown that a reversible ergodic Markov chain is only reversible with respect to the stationary distribution.  So above $\pi(x)$ is the stationary distribution of $P$.  An immediate consequence of this is that for a chain with uniform stationary distribution, it is reversible if and only if it is symmetric (i.e.~$P(x, y) = P(y, x)$).  Note also that reversible chains have real eigenvalues, since they are similar to the symmetric matrix $\sqrt{\frac{\pi(x)}{\pi(y)}}P(x, y)$ (using the similarity transform $\delta_{xy} \sqrt{\pi(x)}$).

With these definitions and concepts, we can now ask how quickly the Markov chain converges to the stationary distribution.  This is normally defined in terms of the 1-norm mixing time.  We use (half the) 1-norm distance to measure distances between distributions:
\begin{equation}
\vectornorm{s - t} = \frac{1}{2} \vectornorm{s - t}_1 = \frac{1}{2} \sum_i |s_i - t_i|.
\end{equation}
We assume all distributions are normalised so then $0 \le \vectornorm{s - t} \le 1$.  We can now define the mixing time:
\begin{definition}
Let $\pi$ be the stationary distribution of $P$.  Then if $P$ is ergodic the mixing time $\tau$ is
\begin{equation}
\tau(\eps) = \max_s \min_t \{ t \ge 0 : \vectornorm{P^t s - \pi} \le \eps \}.
\end{equation}
\end{definition}
We will also use the (weaker) 2-norm mixing time (note this is not the same as $\tau_2$ in \cite{MontenegroTetali06}):
\begin{definition}
Let $\pi$ be the stationary distribution of $P$.  Then if $P$ is ergodic the 2-norm mixing time $\tau_2$ is
\begin{equation}
\tau_2(\eps) = \max_s \min_t \{ t \ge 0 : \vectornorm{P^t s - \pi}_2 \le \eps \}.
\end{equation}
\end{definition}
Unless otherwise stated, when we say mixing time we are referring to the 1-norm mixing time.

There are many techniques for bounding the mixing time, including finding the second largest eigenvalue of $P$.  This gives a good measure of the mixing time because components parallel to the second largest eigenvector decay the slowest.  We have (for reversible ergodic chains)
\begin{theorem}[see \cite{MontenegroTetali06}, Corollary 1.15]
\label{thm:MixingTimeGap}
\begin{equation*}
\tau(\eps) \le \frac{1}{\Delta} \ln \frac{1}{\pi_* \eps}
\end{equation*}
where $\pi_* = \min \pi(x)$ and $\Delta = \min(1-\lambda_2, 1+\lambda_{min})$ where $\lambda_2$ is the second largest eigenvalue and $\lambda_{min}$ is the smallest.  $\Delta$ is known as the \emph{gap}.
\end{theorem}
If the chain is irreversible, it may not even have real eigenvalues.  However, we can bound the mixing time in terms of the eigenvalues of the reversible matrix $PP^*$ where $P^*(x, y) = \frac{\pi(y)}{\pi(x)} P(y, x)$.  In this case we have (\cite{MontenegroTetali06}, Corollary 1.14)
\begin{equation}
\tau(\eps) \le \frac{2}{\Delta_{PP^*}} \ln \frac{1}{\pi_* \eps}
\end{equation}
where now $\Delta_{PP^*}$ is the gap of the chain $PP^*$.  Note that for a reversible chain $P = P^*$ and $\Delta_{PP^*} \approx 2\Delta$ so the bounds are approximately the same.

This can also be converted into a 2-norm mixing time bound:
\begin{equation}
\label{eq:2normMixing}
\tau_2(\eps)\le \frac{2}{\Delta_{PP^*}} \ln 1/\eps.
\end{equation}
To bound the gap, we will use the comparison theorem in \thmref{Comparison} below.  In this Theorem, we are thinking of the Markov chain as a directed graph where the vertices are the states and there are edges for allowed transitions (i.e.~transitions with non-zero probability).  For irreducible chains, it is possible to make a path from any vertex to any other; we call the path length the number of transitions in such a path (which will in general depend on the choice of path).
\begin{theorem}[see \cite{MontenegroTetali06}, Theorem 2.14]
\label{thm:Comparison}
Let $P$ and $\hat{P}$ be two Markov chains on the same state space $\Omega$ with the same stationary distribution $\pi$.  Then, for every $x \ne y \in \Omega$ with $\hat{P}(x, y) > 0$ define a directed path $\gamma_{xy}$ from $x$ to $y$ along edges in $P$ and let its length be $| \gamma_{xy} |$.  Let $\Gamma$ be the set of all such paths.  Then
\begin{equation*}
\Delta \ge \hat{\Delta}/A
\end{equation*}
for the gaps $\Delta$ and $\hat{\Delta}$ where
\begin{equation*}
A = A(\Gamma) = \max_{a \ne b, P(a, b) \ne 0} \frac{1}{\pi(a) P(a, b)} \sum_{x \ne y : (a, b) \in \gamma_{xy}} \pi(x) \hat{P}(x, y) | \gamma_{xy} |.
\end{equation*}
\end{theorem}
For example, when comparing 1-dimensional random walks there is no choice in the paths; they must pass through every point between $x$ and $y$.  Further, the walk can only progress one step at a time so (without loss of generality, for reversible chains) let $b = a+1$ to give
\begin{align}
\label{eq:ComparisonWalk}
A &= \max_a \frac{1}{\pi(a) P(a, a+1)} \sum_{x \le a} \sum_{y \ge a+1} \pi(x) \hat{P}(x, y) (y-x) \nonumber \\
&= \max_a \frac{\hat{P}(a, a+1)}{P(a, a+1)}.
\end{align}
A generalisation of the comparison theorem involves constructing flows, which are weighted sets of paths between states.  This can give a tighter bound since bottlenecks are averaged over.  This gives a modified comparison theorem:
\begin{theorem}[\cite{DiaconisSaloff-Coste93}, Theorem 2.3]
\label{thm:ComparisonFlows}
Let $P$ and $\hat{P}$ be two Markov chains on the same state space $\Omega$ with the same stationary distribution $\pi$.  Then, for every $x \ne y \in \Omega$ with $\hat{P}(x, y) > 0$, construct a set of directed paths $\mathcal{P}_{xy}$ from $x$ to $y$ along edges in $P$.  We define the flow function $f$ which maps each path $\gamma_{xy} \in \mathcal{P}_{xy}$ to a real number in the interval $[0, 1]$ such that
\begin{equation*}
\sum_{\gamma_{xy} \in \mathcal{P}_{xy}} f(\gamma_{xy}) = \hat{P}(x, y).
\end{equation*}
Again, let the length of each path be $| \gamma_{xy} |$.  Then
\begin{equation*}
\Delta \ge \hat{\Delta}/A
\end{equation*}
for the gaps $\Delta$ and $\hat{\Delta}$ where
\begin{equation}
\label{eq:ComparisonFlowsA}
A = A(f) = \max_{a \ne b, P(a, b) \ne 0} \frac{1}{\pi(a) P(a, b)} \sum_{x \ne y, \gamma_{xy} \in \mathcal{P}_{xy} : (a, b) \in \gamma_{xy}} \pi(x) f(\gamma_{xy}) | \gamma_{xy} |.
\end{equation}
\end{theorem}
Note that we recover the comparison theorem when there is just one path between each $x$ and $y$.

Yet another generalisation is to allow general length functions instead of simply counting the edges.  This only appears in the literature as a comparison to the chain $\hat{P}(x, y) = \pi(y)$ although it can easily be generalised to allow comparison with any chain.
\begin{theorem}[\cite{KahaleMixing}, Proposition 1]
\label{thm:ComparisonLengthFunction}
Let $P$ be a Markov chain on the state space $\Omega$ with stationary distribution $\pi$.  Then, for every $x \ne y \in \Omega$ define a directed path $\gamma_{xy}$ from $x$ to $y$ along edges in $P$ and let its length be
\be
| \gamma_{xy} |_l = \sum_{(a,b) \in \gamma_{xy}} l(a,b)
\ee
for any positive length function $l(a,b)$, defined on the edges of the path.  Let $\Gamma$ be the set of all such paths.  Then
\bes
\Delta \ge 1/A
\ees
where
\bes
A = A(\Gamma) = \max_{a \ne b, P(a, b) \ne 0} \frac{1}{l(a, b) \pi(a) P(a, b)} \sum_{x \ne y : (a, b) \in \gamma_{xy}} \pi(x) \hat{P}(x, y) | \gamma_{xy} |_l.
\ees
\end{theorem}

\subsubsection{Decomposition}

For some Markov chains, it is easier to consider different parts of the chain separately to prove convergence results.  This allows, for example, different convergence techniques to be used on different parts of the chain.  The separate parts are combined using the decomposition theorem:
\begin{theorem}[\cite{MartinRandallDecomposition}, Theorem 4.2]
\label{thm:Decomposition}
Let $P(x, y)$ be the transition matrix for a reversible Markov chain with state space $\Omega$ and stationary distribution $\pi(x)$.  Then let $\Omega_i$ be disjoint subsets of $\Omega$ such that $\cup_i \Omega_i = \Omega$.  Let
\begin{equation}
P_i(x, y) = \begin{cases}
P(x, y) & x, y \in \Omega_i, x \ne y \\
1 - \sum_{y' \in \Omega_i, y' \ne x} P(x, y') & x = y \in \Omega_i \\
0 & otherwise.
\end{cases}
\end{equation}
Further let $w_i = \sum_{x \in \Omega_i} \pi(x)$ and
\begin{equation}
\bar{P}(i, j)  = \frac{1}{w_i} \sum_{x \in \Omega_i, y \in \Omega_j} \pi(x) P(x, y).
\end{equation}
Then
\begin{equation}
\Delta \ge \frac{1}{2} \bar{\Delta} \min_{i} \Delta_i
\end{equation}
where $\Delta$ is the gap of $P$, $\bar{\Delta}$ for $\bar{P}$ and $\Delta_i$ for $P_i$.
\end{theorem}

\subsubsection{log-Sobolev Constant}

We will need tighter, but more complicated, mixing time results to prove the tight result for the $U(4)$ case.  We use the log-Sobolev constant:
\begin{definition}
\label{def:LogSobolev}
The log-Sobolev constant $\rho$ of a chain with transition matrix $P$ and stationary distribution $\pi$ is
\begin{equation*}
\rho = \min_f \frac{\sum_{x \ne y} (f(x) - f(y))^2 P(x, y) \pi(y)}{\sum_x \pi(x) f(x)^2 \log \frac{f(x)^2}{\sum_y \pi(y) f(y)^2}}.
\end{equation*}
\end{definition}
The mixing time result is:
\begin{lemma}[see \cite{DiaconisSaloff-Coste96}, Theorem 3.7']
The mixing time of a finite, reversible, irreducible Markov chain is
\begin{equation}
\label{eq:SobolevMixingTime}
\tau(\eps) = O\left(\frac{1}{\rho} \log \log \frac{1}{\pi_*} + \frac{1}{\Delta}{\log \frac{d}{\eps}}\right)
\end{equation}
where $\rho$ is the Sobolev constant, $\pi_*$ is the smallest value of the stationary distribution, $\Delta$ is the gap and $d$ is the size of the state space.
\end{lemma}
Further, the comparison theorem (\thmref{Comparison}) works just the same to give
\begin{equation*}
\rho \ge \hat{\rho}/A.
\end{equation*}
We will need one more result, due to Diaconis and Saloff-Coste:
\begin{lemma}[\cite{DiaconisSaloff-Coste96}, Lemma 3.2]
\label{lem:ProductChain}
Let $P_i$, $i = 1, \ldots, d$, be Markov chains with gaps $\Delta_i$ and Sobolev constants $\rho_i$.  Now construct the product chain $P$.  This chain has state space equal to the product of the spaces for the chains $P_i$ and at each step one of the chains is chosen at random and run for one step.  Then $P$ has spectral gap given by:
\begin{equation*}
\Delta = \frac{1}{d} \min_i \Delta_i
\end{equation*}
and Sobolev constant:
\begin{equation*}
\rho = \frac{1}{d} \min_i \rho_i.
\end{equation*}
\end{lemma}

\subsection{Convergence Proof}

We now prove the Markov chain convergence results to show that the $\gamma(p, p)$ terms converge quickly.  We have already shown that the $\gamma(p_1, p_2)$ terms with $p_1 \ne p_2$ converge quickly and that there is no mixing between these terms and the $\gamma(p, p)$ terms.  Therefore, in this section, we remove such terms from $\hat{G}$.

We want to prove the Markov chain with transition matrix (\eq{GeneralTransitionMatrix})
\begin{equation*}
P = \frac{1}{n(n-1)} \sum_{i \ne j} \hat{G}^{(ij)}
\end{equation*}
converges quickly.  Firstly, we know from \secref{MomentsGeneral} that $P$ has two eigenvectors with eigenvalue $1$.  The first is the identity state ($\sigma_0 \otimes \sigma_0$) and the second is the uniform sum of all non-identity terms ($\frac{1}{4^n-1}\sum_{p \ne 0} \sigma_p \otimes \sigma_p$).  From now on, we remove the identity state.  This makes the chain irreducible.  Since we know it converges, it must be aperiodic also so the chain is ergodic and all other eigenvalues are strictly between $1$ and $-1$.

We show here that the gap of this chain, up to constants, does not depend on the choice of 2-copy gapped gate set.  In the second half of the chapter we find a tight bound on the gap for the $U(4)$ case which consequently gives a tight bound on the gap for all universal sets.
 
Since the stationary distribution is uniform, the chain is reversible
if and only if $P$ is a symmetric matrix.  A sufficient condition for
$P$ to be symmetric is for $\hat{G}^{(ij)}$ to be symmetric.  We saw
in \thmref{SymmetryG} that for the $U(4)$ gate set case
$\hat{G}^{(ij)}$ is symmetric.  In fact, the proof works identically
to show that $\hat{G}^{(ij)}$ is symmetric for any gate set, provided
the set is invariant under Hermitian conjugation.  However, 2-copy gapped gate
sets do not necessarily have this property so the Markov chain is not
necessarily reversible.  We will find equal bounds (up to constants)
for the gaps of both $P$ (if $\hat{G}$ is symmetric) and $P P^*$ (if
$\hat{G}$ is not symmetric) below: 
\begin{theorem}
\label{thm:GapGeneralUniversal}
Let $\mu$ be any 2-copy gapped distribution of gates.  If $\mu$ is
invariant under Hermitian conjugation then let $\Delta_P$ be the
eigenvalue gap of the resulting Markov chain matrix $P$.  Then
\begin{equation}
\Delta_P = \Omega(\Delta_{U(4)})
\end{equation}
where $\Delta_{U(4)}$ is the eigenvalue gap of the $U(4)$ chain.  If
$\mu$ is not invariant under Hermitian conjugation then let $\Delta_{P
  P^*}$ be the eigenvalue gap of the resulting Markov chain matrix $P
P^*$.  Then 
\begin{equation}
\Delta_{P P^*} = \Omega(\Delta_{U(4)}).
\end{equation}
\end{theorem}
\begin{proof}
We will use the comparison method with flows (\thmref{ComparisonFlows}).  Firstly consider the case where $\mu$ is closed under Hermitian conjugation i.e.~$\hat{G}$ is symmetric.

We will compare $P$ to the $U(4)$ chain, which we call $P_{U(4)}$.  Recall that this chain chooses a pair at random and does nothing if the pair is $00$ and chooses a random state from $\{0,1,2,3\}^2 \backslash \{00\}$ otherwise.

To apply \thmref{ComparisonFlows}, we need to construct the flows between transitions in $P_{U(4)}$.  We will choose paths such that only one pair is modified throughout.  For example (with $n=4$), the transition $1000 \rightarrow 2000$ is allowed in $P_{U(4)}$.  To construct a path in $P$, we need to find allowed transitions between these two paths in $P$.  $\hat{G}$ may not include the transition $10 \rightarrow 20$ directly, however, $\hat{G}$ is irreducible on this subspace of just two pairs.  This means that a path exists and can be of maximum length $14$ if it has to cycle through all intermediate states (in fact, since $\hat{G}$ is symmetric the maximum path length is $8$; all that is important here is that it is constant).  For example, the transitions $10 \rightarrow 11 \rightarrow 20$ might be allowed.  Then we could choose the full path to be $1000 \rightarrow 1100 \rightarrow 2000$.  In this case we have chosen the path to involve transitions pairing sites 1 and 2.  However, we could equally well have chosen any pairing; we could pair the first site with any of the others.  We can choose 3 paths in this way.  For this example, the flow we want to choose will be all 3 of these paths equally weighted.  We now use this idea to construct flows between all transitions in $P_{U(4)}$ to prove the result.

Let $x \ne y \in \Omega$ and let $d(x, y)$ be the Hamming distance between the states ($d(x, y)$ gives the number of places at which $x$ and $y$ differ).  There are two cases where $P_{U(4)}(x, y) \ne 0$:
\begin{enumerate}
\item{$d(x, y) = 2$.  Here we must choose a unique pairing, specified by the two sites that differ.  Make all transitions in $P$ using this pair giving just one path.}
\item{$d(x, y) = 1$.  For this case, choose all possible pairings of the changing site that give allowed transitions in $P_{U(4)}$.  For each pairing, construct a path in $P$ modifying only this pair.  If the differing site is initially non-zero then there are $n-1$ such pairings; if the differing site is initially zero then there are $n-z(x)$ pairings where $z(x)$ is the number of zeroes in the state $x$.}
\end{enumerate}
All the above paths are of constant length since we have to (at most) cycle through all states of a pair.  We must now choose the weighting $f(\gamma_{xy})$ for each path such that
\be
\sum_{\mathcal{P}_{xy}} f(\gamma_{xy}) = P_{U(4)}(x, y)
\ee
where $\mathcal{P}_{xy}$ is the set of all paths from $x$ to $y$ constructed above.  We choose the weighting of each path to be uniform.  We just need to calculate the number of paths in $\mathcal{P}_{xy}$ to find $f$:
\begin{enumerate}
\item{$d(x, y) = 2$.  There is just one path so $f(\gamma_{xy}) = P_{U(4)}(x, y) = \Theta(1/n^2)$.}
\item{$d(x, y) = 1$.  If the differing site is initially non-zero then $P_{U(4)}(x, y) = \Theta(1/n)$ and there are $n-1$ paths so $f(\gamma_{xy}) = \frac{P_{U(4)}(x,y)}{n-1} = \Theta(1/n^2)$.  If the differing site is initially zero then $P_{U(4)}(x, y) = \Theta\left(\frac{n-z(x)}{n^2}\right)$ and there are $n-z(x)$ paths so $f(\gamma_{xy}) = \frac{P_{U(4)}(x, y)}{n-z(x)} = \Theta(1/n^2)$.}
\end{enumerate}
So for all paths, $f = \Theta(1/n^2)$.    We now just need to know how many times each edge $(a, b)$ in $P$ is used to calculate $A$:
\begin{equation}
A = \max_{a \ne b, P(a, b) \ne 0} A(a, b)
\end{equation}
where
\begin{equation}
A(a, b) = \frac{1}{P(a, b)} \sum_{x \ne y, \gamma_{xy} \in \mathcal{P}_{xy} : (a, b) \in \gamma_{xy}} f(\gamma_{xy}).
\end{equation}
We have cancelled the factors of $\pi(x)$ because the stationary distribution is uniform.  We have also ignored the lengths of the paths since they are all constant.  

To evaluate $A(a, b)$, we need to know how many paths pass through each edge $(a, b)$.  We again consider the two possibilities separately:
%.  For each possible pairing in $P$ that can give the transition from $a$ to $b$, there is a constant number of $x, y$ pairs with paths through it, since at least one of these sites is changing.  For any $x, y$ there is at most one path $\gamma_{xy}$ using each transition pair.  We can now calculate $A$ (\eq{ComparisonFlowsA}).  We need to evaluate:
%Now we evaluate $A(a, b)$ for all pairs with $P(a, b) \ne 0$:
\begin{enumerate}
\item{$d(a, b) = 2$.  Suppose $a$ and $b$ differ at sites $i$ and $j$.  Firstly, we need to count how many transitions from $x$ to $y$ in $P_{U(4)}$ could use this edge, and then how many paths for each transition actually use the edge. 

To find which $x$ and $y$ could use the edge, note that $x$ and $y$ must differ at sites $i$, $j$ or both.  Furthermore, the values at the sites other than $i$ and $j$ must be the same as for $a$ (and therefore $b$).  There is a constant number of $x, y$ pairs that satisfy this condition.  Now, for each $x, y$ pair satisfying this, paths that use this edge must use the pairing $i, j$ for all transitions.  Since in the paths we have chosen above there is a unique path from $x$ to $y$ for each pairing, there is at most one path for each $x, y$ pair that uses edge $a, b$.

For $d(a, b) = 2$, $P(a, b) = \Theta(1/n^2)$ so $A(a, b)$ is a constant for this case.}
\item{$d(a, b) = 1$.  Let there be $r$ pairings that give allowed transitions in $P$ between $a$ and $b$.  As above, each pairing gives a constant number of paths.  So the numerator is $\Theta(r/n^2)$.  Further, $P(a, b) = \Theta(r/n^2)$.  So again $A(a, b)$ is constant.}
\end{enumerate}
Combining, $A$ is a constant so the result is proven for the case $\hat{G}$ is symmetric.

We now turn to the irreversible case.  We now need to bound the gap of $P P^* = P P^T$.  This chain selects two (possibly overlapping) pairs at random and applies $\hat{G}$ to one of them and $\hat{G}^T$ to the other.  We can use the above exactly by choosing $\hat{G}$ to perform the transitions above and $\hat{G}^T$ to just loop the states back to themselves.  By aperiodicity (the greatest common divisor of loop lengths is $1$), we can always find constant length paths that do this.
\end{proof}

Now we need to know the gap of the $U(4)$ chain.  We can, by a simple application of the comparison theorem, show it is $\Omega(1/n^2)$.  However, in the second half of this chapter we show it is $\Theta(1/n)$.  This gives us (using \thmref{MixingTimeGap}):
\begin{corollary}
\label{cor:SecondMomentsMixing}
The Markov chain $P$ has mixing time $O(n(n+ \log 1/\eps))$ and 2-norm mixing time $O(n\log 1/\eps)$.
\end{corollary}
We conjecture that the mixing time (as well as \lemref{CoefficientsDecay2General}) can be tightened to $\Theta(n\log\frac{n}{\eps})$, which is asymptotically the same as for the $U(4)$ case:
\begin{conjecture}
\label{conj:Mixing}
The second moments for the case of general 2-copy gapped distributions
have 1-norm mixing time $\Theta(n\log\frac{n}{\eps})$. 
\end{conjecture}
It seems likely that an extension of our techniques in \secref{U4Convergence} could be used to prove this.

Combining the convergence results we have proved our general result \lemref{MainMixing}:
\begin{proof}[Proof of \lemref{MainMixing}]
Combining \corref{SecondMomentsMixing} (for the $\gamma(p, p)$ terms) and \lemref{CoefficientsDecay2General} (for the $\gamma(p_1, p_2)$, $p_1 \ne p_2$ terms) proves the result.
\end{proof}

We have now shown that the first and second moments of random circuits converge quickly.  For the remainder of the chapter we prove the tight bound for the gap and mixing time of the $U(4)$ case and show how mixing time bounds relate to the closeness of the 2-design to an exact design.  Only for the $U(4)$ case is the matrix $\hat{G}$ a projector so in this sense the $U(4)$ random circuit is the most fundamental.  While we expect the above mixing time bound is not tight, we can prove a tight mixing time result for the $U(4)$ case.  However, using our definition of an approximate $k$-design, the gap rather than the mixing time governs the degree of approximation.

\section{Tight Analysis for the \texorpdfstring{$U(4)$}{U(4)} Case}
\label{sec:U4Convergence}

We have already found tight bounds for the first moments in \lemref{CoefficientDecayGeneral}: just set $\Delta = 1$.

\subsection{Second Moments Convergence}

We need to prove a result analogous to \lemref{CoefficientsDecay2General} for the terms $\sigma_{p_1} \otimes \sigma_{p_2}$ where $p_1 \ne p_2$.  We already have a tight bound for the 2-norm decay, by setting $\Delta = 1$ into \lemref{CoefficientsDecay2General}.  We tighten the 1-norm bound:
\begin{lemma}
\label{lem:CoefficientsDecay2}
After $O(n \log \frac{n}{\eps})$ steps
\begin{equation}
\label{eq:TermsDecay}
\sum_{p_1 \ne p_2} \Expect_W | \gamma_W(p_1, p_2) | \le \eps
\end{equation}
\end{lemma}
\begin{proof}
We will split the random circuits up into classes depending on how many qubits have been hit.  Let $H$ be the random variable giving the number of different qubits that have been hit.  We can work out the distribution of $H$ and bound the sum of $| \gamma_W(p_1, p_2) |$ for each outcome.

Firstly we have, after $t$ steps,
\begin{equation*}
\Pr(H \le h) \le {n \choose h} \left(\frac{h(h-1)}{n(n-1)}\right)^t \le {n \choose h} (h/n)^t.
\end{equation*}
Now, for each qubit hit, each coefficient which has $p_1$ and $p_2$ differing in this place is set to zero.  So after $h$ have been hit, there are only (at most) $16^{(n-h)}$ terms in the sum in \eq{TermsDecay}.  As before, the state is a physical state, $\tr \rho^2 \le 1$ so $\sum_{p_1 p_2} \gamma^2(p_1, p_2) \le 1$ so $\sum_{p_1 p_2} | \gamma(p_1, p_2) | \le \sqrt{N}$ if there are at most $N$ non-zero terms in the sum.  Therefore we have, after $t$ steps,
\begin{align*}
\sum_{p_1 \ne p_2} \Expect_W | \gamma_W(p_1, p_2) | &\le \sum_{h=1}^{n-1} \Pr(H = h) 16^{(n-h)/2} \\
&\le \sum_{h=1}^{n-1} \Pr(H \le h) 4^{(n-h)} \\
&\le \sum_{h=1}^{n-1} {n \choose h} (h/n)^t 4^{(n-h)} \\
&= \sum_{h=1}^{n-1} {n \choose h} (1-h/n)^t 4^{h} \qquad h \rightarrow n-h  \\
&\le \sum_{h=1}^{n-1} {n \choose h} \exp(-ht/n) 4^{h}.
\end{align*}
Now, let $t = n \ln \frac{n}{\eps}$:
\begin{align*}
\sum_{p_1 \ne p_2} \Expect_W | \gamma_W(p_1, p_2) | &\le \sum_{h=1}^{n-1} {n \choose h} \left(\frac{4 \eps}{n} \right)^h \\
&= \left( 1 + \frac{4\eps}{n} \right)^n -1 -\left( \frac{4\eps}{n} \right)^n = O(\eps)
\end{align*}
where the last line follows from the binomial theorem.
\end{proof}
This, combined with the mixing time result we prove below, completes the proof that the second moments of the random circuit converge in time $O(n \log \frac{n}{\eps})$.

\subsection{Markov Chain of Coefficients}

The Markov chain acting on the coefficients is reducible because the state $\{ 0 \}^n$ is isolated.  However, if we remove it then the chain becomes irreducible.  The presence of self loops implies aperiodicity therefore the chain is ergodic.  We have already seen that the chain converges to the Haar uniform distribution (in \secref{RandomCircuits}) therefore the stationary state is the uniform state $\pi(x) = 1/(4^n-1)$.  Further, since the chain is symmetric and has uniform stationary distribution, the chain satisfies detailed balance (\eq{DetailedBalance}) so is reversible.  We now turn to obtaining bounds on the mixing time of this chain.

We want to show that the full chain converges to stationarity in time $\Theta(n \log \frac{n}{\eps})$.  To prove this, we will construct another chain called the zero chain.  This is the chain that counts the number of zeroes in the state.  Since it is the zeroes that slow down the mixing, this chain will accurately describe the mixing time of the full chain.
\begin{lemma}
\label{lem:ZeroChainTransitionMatrix}
The zero chain has transition matrix P on state space (we count non-zero positions) $\Omega = \{1,2, \ldots, n\}$.
\begin{equation}
P(x, y) =
\begin{cases}
  1 - \frac{2x(3n-2x-1)}{5n(n-1)}  & y = x \\
  \frac{2x(x-1)}{5n(n-1)}  & y = x-1 \\
  \frac{6x(n-x)}{5n(n-1)}  & y = x+1 \\
  0 & \rm{otherwise}
\end{cases}
\end{equation}
for $1 \le x, y \le n$.
\end{lemma}
\begin{proof}
Suppose there are $n-x$ zeroes (so there are $x$ non-zeroes).  Then the only way the number of zeroes can decrease (i.e.~for $x$ to increase) is if a non-zero item is paired with a zero item and one of the $9$ (out of $15$) new states is chosen with no zeroes.  The probability of choosing such a pair is $\frac{2x(n-x)}{n(n-1)}$ so the overall probability is $\frac{9}{15} \frac{2x(n-x)}{n(n-1)}$.

The number of zeroes can increase only if a pair of non-zero items is chosen and one of the $6$ states is chosen with one zero.  The probability of this occurring is $\frac{6}{15} \frac{x(x-1)}{n(n-1)}$.

The probability of the number of zeroes remaining unchanged is simply calculated by requiring the probabilities to sum to $1$.
\end{proof}
We see that the zero chain is a one-dimensional random walk on the line.  It is a lazy random walk because the probability of moving at each step is $<1$.  However, as the number of zeroes decreases, the probability of moving increases monotonically:
\begin{equation}
1-P(x,x) = \frac{2x(3n-2x-1)}{5n(n-1)} \ge 2x/5n.
\end{equation}

\begin{lemma}
\label{lem:ZeroStatDistrib}
The stationary distribution of the zero chain is
\begin{equation}
\label{eq:ZeroStatDistrib}
\pi_0(x) = \frac{3^x {n \choose x}}{4^n-1}.
\end{equation}
\end{lemma}
\begin{proof}
This can be proven by multiplying the transition matrix in \lemref{ZeroChainTransitionMatrix} by the state \eq{ZeroStatDistrib}.  Alternatively, it can be proven by counting the number of states with $n-x$ zeroes.  There are ${n \choose x}$ ways of choosing which sites to make non-zero and each non-zero site can be one of three possibilities: 1, 2 or 3.  The total number of states is $4^n-1$, which gives the result.
\end{proof}

Below we will prove the following theorem:
\begin{theorem}
\label{thm:ZeroChainMixing}
The zero chain mixes in time $\Theta(n \log \frac{n}{\eps})$.
\end{theorem}
We prove this using direct arguments about the convergence of the random walk.  However, we also include a less complex method that only bounds the gap:
\begin{theorem}
\label{thm:ZeroChainGap}
The zero chain has gap $\Omega(1/n)$.
\end{theorem}
This only implies the mixing time is $O(n(n+\log 1/\eps))$ which is weaker than \thmref{ZeroChainMixing}, although still sufficient to prove our main result \thmref{Main2Design}, using a modification of \corref{FullChainMixing} to show that the full chain mixing time is $O(n(n+\log 1/\eps))$.

Knowing the gap allows us to easily work out the 2-norm mixing time:
\begin{theorem}
\label{thm:ZeroChainMixing2Norm}
The zero chain has 2-norm mixing time $O(n \log 1/\eps)$.
\end{theorem}
\begin{proof}
Use the bound on the gap in \thmref{ZeroChainGap} and \eq{2normMixing}.
\end{proof}
Before proving \thmref{ZeroChainMixing}, we will show how the mixing time of the full chain follows from this.

\begin{corollary}
\label{cor:FullChainMixing}
The full chain mixes in time $\Theta(n \log \frac{n}{\eps})$.
\end{corollary}
\begin{proof}
Once the zero chain has approximately mixed, the distribution of zeroes is almost correct.  We need to prove that the distribution of non-zeroes is correct after $O(n \log \frac{n}{\eps})$ steps too.

Once each site of the full chain has been hit, meaning it is chosen and paired with another site so not both equal zero, the chain has mixed.  This is because, after each site has been hit, the probability distribution over the states is uniform.  When the zero chain has approximately mixed, a constant fraction of sites are zero so the probability of hitting a site at each step is $\Theta(1/n)$.  By the coupon collector argument, each site will have been hit with probability at least $1-\eps$ in time time $O(n \log \frac{n}{\eps})$.  Once the zero chain has mixed to $\eps'$, we can run the full chain this extra number of steps to ensure each site has been hit with high probability.  Since the mixing of the zero chain only increases with time, the distance to stationarity of the full chain is now $1-\eps-\eps'$.  We make this formal below.

After $t_0 = O(n \log \frac{n}{\eps'})$ steps, the number of zeroes is $\eps'$-close to the stationary distribution $\pi_0$ by \thmref{ZeroChainMixing} and only gets closer with more steps since the distance to stationarity decreases monotonically.  The stationary distribution \eq{ZeroStatDistrib} is approximately a Gaussian peaked at $3n/4$ with $O(n)$ variance.  This means that, with high probability, the number of non-zeroes is close to $3n/4$.  We will in fact only need that there is at least a constant fraction of non-zeroes; with probability at least $1-\eps'-\exp(-\Omega(n))$ there will be at least $n/2$.

To prove the mixing time, we run the chain for time $t_0$ so the zero chain mixes to $\eps'$.  Then run for $t_1$ additional steps.  Let $H_{i, t}$ be the event that site $i$ is hit at step $t$.  Let $H_i = \cup_{t = t_0+1}^{t_0 +t_1} H_{i, t}$ and $H = \cap_{i = 1}^n H_i$.  We want to show $\Pr(H)$ is close to 1, or, in other words, that all sites are hit with high probability.  Further let $X_t$ be the random variable giving the number of non-zeroes at step $t$.

If at step $t-1$ site $i$ is non-zero then the event $H_{i, t}$ occurs if the qubit is chosen, which occurs with probability $2/n$.  If, however, it was zero then it must be paired with a non-zero thing for $H_{i, t}$ to hold.  Conditioned on any history with $X_{t-1} \ge n/2$, this probability is $\ge 1/n$.  In particular, we can condition on not having previously hit $i$ and the bound does not change.  Combining we have
\begin{equation*}
\Pr\left(H_{i, t}^c \bigg| \left[ X_{t-1} \ge n/2 \right] \bigcap \left( \bigcap_{t' = t_0+1}^{t-1} H_{i, t'}^c \right)\right) \le 1 - 1/n.
\end{equation*}
Then, after $t_1$ extra steps,
\begin{equation*}
\Pr\left(H_i^c \bigg| \bigcap_{t = t_0}^{t_0+t_1-1} \left[ X_{t} \ge n/2 \right] \right) \le (1 - 1/n)^{t_1}
\end{equation*}
which, using the union bound, gives
\begin{equation*}
\Pr\left(H^c \bigg| \bigcap_{t = t_0}^{t_0+t_1-1} \left[ X_{t} \ge n/2 \right] \right) \le n(1 - 1/n)^{t_1}.
\end{equation*}
Now, since the zero chain has mixed to $\eps'$,
\begin{equation*}
\Pr\left(\overline{\bigcap_{t = t_0}^{t_0+t_1-1} \left[ X_{t} \ge n/2 \right]}\right) \le t_1 \l(\sum_{x=n/2}^{n-1} \pi_0(x) + \eps'\r)  \le t_1 \l(\exp(-O(n)) + \eps'\r)
\end{equation*}
so
\begin{equation*}
\Pr(H^c) \le n(1 - 1/n)^{t_1} + t_1 \l(\exp(-O(n)) + \eps'\r).
\end{equation*}
Now, choose $t_1 = n \ln \frac{2n}{\eps}$ so that $\Pr(H^c) \le \delta$ where $\delta = \eps + t_1 \l(\exp(-O(n)) + \eps'\r)$.  Choose $\eps = 1/n$ and $\eps' = 1/n^3$ so that $\delta$ is $1/\poly(n)$.  Now, using the bound on $\Pr(H^c)$, we can write the state $v$ after $t_1 = O(n \log n)$ steps as
\begin{equation*}
v = (1-\delta) \pi + \delta \pi'
\end{equation*}
where $\pi$ is the stationary distribution and $\pi'$ is any other distribution.  Using this,
\begin{equation*}
|| v- \pi || \le \delta.
\end{equation*}
We now apply \lemref{DistanceToStat} to show that after $O(n \log \frac{n}{\eps})$ steps the distance to stationarity of the full chain is $\eps$.
\end{proof}

\subsection{Proof of Theorem 3.6.4} %\thmref{ZeroChainMixing}}

We will now proceed to prove \thmref{ZeroChainMixing}.  We present an outline of the proof here; the details are in \secref{ZeroChainMixingProof}.

Firstly, note that by the coupon collector argument, the lower bound on the time is $\Omega (n \log n)$.  We need to prove an upper bound equal to this.  Intuition says that the mixing time should take time $O(n \log n)$ because the walk has to move a distance $\Theta(n)$ and the waiting time at each step is proportional to $n, n/2, n/3, \ldots$ which sums to $O(n \log n)$, provided each site is not hit too often.   We will show that this intuition is correct using Chernoff bound and log-Sobolev (see later) arguments.

We will first work out concentration results of the position after some number of \emph{accelerated} steps.  The zero chain has some probability of staying still at each step. The accelerated chain is the zero chain conditioned on moving at each step.  We define the accelerated chain by its transition matrix:
\begin{definition}
The transition matrix for the accelerated chain is
\begin{equation}
P_a(x, y) =
\begin{cases}
  0  & y = x \\
  \frac{x-1}{3n-2x-1} & y = x-1 \\
  \frac{3(n-x)}{3n-2x-1} & y = x+1 \\
  0 & \rm{otherwise}.
\end{cases}
\end{equation}
\end{definition}
We use the accelerated chain in the proof to firstly prove the accelerated chain mixes quickly, then to bound the waiting time at each step to obtain a mixing time bound for the zero chain.
  
To prove the mixing time bound, we will split the walk up into three phases.  We will split the state space into three (slightly overlapping) parts and the phase can begin at any point within that space. So each phase has a state space $\Omega_i \subset [1, n]$, an entry space $E_i \subset \Omega_i$ and an exit condition $T_i$.  We say that a phase completes successfully if the exit condition is satisfied in time $O(n \log n)$ for an initial state within the entry space.  When the exit condition is satisfied, the walk moves onto the next phase.

The phases are:
\begin{enumerate}
\item{$\Omega_1 = [1, n^\delta]$ for some constant $\delta$ with $0 < \delta < 1/2$.  $E_1 = \Omega_1$ (i.e.~it can start anywhere) and $T_1$ is satisfied when the walk reaches $n^\delta$.  For this part, the probability of moving backwards (gaining zeroes) is $O(n^{\delta -1})$ so the walk progresses forwards at each step with high probability.  This is proven in \lemref{Phase1MovesRight}.  We show that the waiting time is $O(n \log n)$ in \lemref{Phase1WaitingTime}.}
\item{$\Omega_2 = [n^\delta/2, \theta n]$ for some constant $\theta$ with $0 < \theta < 3/4$.  $E_2 = [n^\delta, \theta n]$ and $T_2$ is satisfied when the walk reaches $\theta n$.  Here the walk can move both ways with constant probability but there is a $\Omega(1)$ forward bias.  Here we use a monotonicity argument: the probability of moving forward at each step is
\begin{align*}
p(x) &= \frac{3(n-x)}{3n-2x-1} \\
& \ge \frac{3(n-x)}{3n-2x} \\
& \ge \frac{3(1-\theta)}{3-2\theta}.
\end{align*}
If we model this random walk as a walk with constant bias equal to $\frac{3(1-\theta)}{3-2\theta}$ we will find an upper bound on the mixing time since mixing time increases monotonically with decreasing bias.  Further, the waiting time at $x=a$ stochastically dominates the waiting time at $x=b$ for $b \ge a$.  The true bias decreases with position so the walk with constant bias spends more time at the early steps.  Thus the position of this simplified walk is stochastically dominated by the position of the real walk while the waiting time stochastically dominates the waiting time of the real walk.}
\item{$\Omega_3 = [\frac{\theta}{2} n, n]$ and $E_3 = [\theta n, n]$.  $T_3$ is satisfied when this restricted part of the chain has mixed to distance $\eps$.  Here the bias decreases to zero as the walk approaches $3n/4$ but the moving probability is a constant.  We show that this walk mixes quickly by bounding the log-Sobolev constant of the chain.}
\end{enumerate}
Showing these three phases complete successfully will give a mixing time bound for the whole chain.

We now prove in \secref{RandomCircuitsProofs} that the phases complete successfully with probability at least $1-1/\poly(n)$:
\begin{lemma}
\label{lem:Phase1Completes}
\begin{equation*}
\Pr(\text{\rm{Phase 1 completes successfully}}) \ge 1 - n^{2\delta-1} - 2 n^{-\delta}
\end{equation*}
\end{lemma}

\begin{lemma}
\label{lem:Phase2Completes}
\begin{multline*}
\Pr(\text{\rm{Phase 2 completes successfully}}) \ge \\ 1 - \exp\left(-\frac{2}{3} \mu \theta n\right) - \left(\frac{4}{\theta n}\right)^\frac{3}{2\mu} - \frac{2 \exp\left(\frac{-\mu n^\delta}{4}\right)}{1-\exp(-\mu/2)} - \left(q/p \right)^{n^\delta/2}
\end{multline*}
where $\mu = \frac{6(1-\theta)}{3-2\theta} -1$.
\end{lemma}

\begin{lemma}
\label{lem:Phase3Completes}
\begin{equation*}
\Pr(\text{\rm{Phase 3 completes successfully}})  \ge 1 - \left( \frac{\theta}{3(2-\theta)} \right)^{\theta n/2}
\end{equation*}
\end{lemma}

We can now finally combine to prove our result:
\begin{proof}[Proof of \thmref{ZeroChainMixing}]
The stationary distribution has exponentially small weight in the tail with lots of zeroes.  We show that, provided the number of zeroes is within phase 3, the walk mixes in time $O(n \log \frac{n}{\eps})$.  We also show that if the number of zeroes is initially within phase 1 or 2, after $O(n \log n)$ steps the walk is in phase 3 with high probability.  We can work out the distance to the stationary distribution as follows.

Let $p_f$ be the probability of failure.  This is the sum of the error probabilities in Lemmas \ref{lem:Phase1Completes}, \ref{lem:Phase2Completes} and \ref{lem:Phase3Completes}.  The key point is that $p_f = 1/\poly(n)$.  Then after $O(n \log \frac{n}{\eps})$ steps (the sum of the number of steps in the 3 phases), the state is equal to $(1-p_f)v_3 + p_f v'$ where $v_3$ is the state in the phase 3 space and $v'$ is any other distribution, which occurs if any one of the phases fails.  Since the distance to stationarity in phase 3 is $\eps$, $|| v_3 - \pi_3 || \le \eps$, where $\pi_3$ is the stationary distribution on the state space of phase 3.  In \lemref{MixingPhase3} we show that $\pi_3(x) = \pi(x)/(1-w)$ where $w = \sum_{x=1}^{\theta n / 2 -1} \pi(x)$.  Since $\pi(x)$ is exponentially small in this range, $w$ is exponentially small in $n$.  Now use the triangle inequality to find
\begin{equation}
||v_3 - \pi|| \le ||v_3 - \pi_3|| + ||\pi_3 - \pi||.
\end{equation}
Since the chain in phase 3 has mixed to $\eps$, the first term is $\le \eps$.  We can evaluate $|| \pi_3 - \pi||$:
\begin{align*}
|| \pi_3 - \pi || &= \frac{1}{2} \sum_{x=1}^n || \pi_3(x) - \pi(x) || \\
&= \frac{1}{2} \left( \sum_{x=1}^{\theta n/2 - 1} \pi(x) + \sum_{x = \theta n /2}^n(\pi(x)/(1-w) - \pi(x)) \right) \\
&= \frac{1}{2} \left( w + 1 - (1-w) \right) = w.
\end{align*}
So now,
\begin{align*}
||(1-p_f)v_3 + p_f v' - \pi|| &= ||(1-p_f)(v_3 - \pi) + p_f(v' - \pi)|| \\
&\le (1-p_f) || v_3 - \pi|| + p_f ||v' - \pi|| \\
&\le (1-p_f)(\eps + w) + p_f \\
&\le \delta
\end{align*}
where $\delta = \eps + w + p_f$.  We are free to choose $\eps$: choose it to be $1/n$ so that $\delta$ is $1/\poly(n)$.  So now the running time to get a distance $\delta$ is $t = O(n \log n)$.  We then apply \lemref{DistanceToStat} to obtain the result.

This concludes the proof of \thmref{ZeroChainMixing} so \corref{FullChainMixing} is proved.
\end{proof}
We have now proven \lemref{MainMixing} and consequently \corref{MainMixing}.  We are now ready to show how \thmref{Main2Design} follows, but first give the alternative proof that the zero chain gap is $\Omega(1/n)$.  The remainder of the proof of \thmref{Main2Design} is in \secref{MainResult}.

\subsection{Proof of Theorem 3.6.5} %\thmref{ZeroChainGap}}
\label{sec:ZeroChainGap}

Here we prove that the gap of the zero chain is $\Omega(1/n)$.  While this can be deduced from \thmref{ZeroChainMixing} and provides weaker mixing time bounds, this bound on the gap is sufficient to prove our main result so we present it as a simpler alternative proof.

We use the method of decomposition (\thmref{Decomposition}), whereby the Markov chain is split up into disjoint state spaces.  This works well here because, for the first part of the walk with many zeroes, the walker remains stationary most of the time whereas when there is a constant fraction of zeroes, the walker moves on most steps.  Using the decomposition method allows us to use different techniques in these different regimes.

We therefore divide the walk up into two parts, $P_1$ and $P_2$, which are shown in \figref{ZeroChainDecomp}.  The chain $\bar{P}$ is the chain that links the two parts, according to the decomposition theorem, \thmref{Decomposition}.

\begin{figure}[h]
  \begin{center}
    \includegraphics[width=12cm]{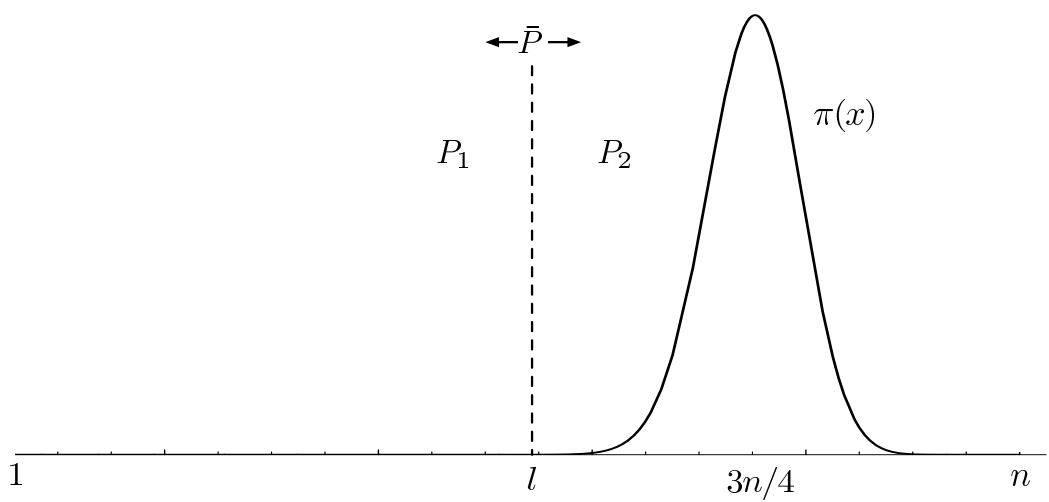}
    \caption[The decomposition of the zero chain]{The decomposition of the zero chain into $P_1$ and $P_2$.  The graph plotted is the zero chain stationary distribution $\pi(x)$.} 
    \label{fig:ZeroChainDecomp}
  \end{center}
\end{figure}

\begin{itemize}
\item{$P_1$: Let $P_1$ have state space $\Omega_1 = \{1, \ldots, m\}$.  This chain has transition matrix
\be
P_1(x,y) = 
\begin{cases}
0 & x > m \text{\, \rm or \,} y > m \\
1 - P(m, m-1) & x = y = m \\
P(x, y) & \text{\rm otherwise}
\end{cases}
\ee
and stationary distribution
\be
\pi_1(x) = \frac{\pi(x)}{b_m}
\ee
where
\be
b_m = \sum_{x=1}^{m} \pi(x).
\ee
}
\item{$P_2$: Let $P_2$ be on state space $\Omega_2 = \{m+1, \ldots, n\}$.  This chain has transition matrix
\be
P_2(x,y) = 
\begin{cases}
0 & x \le m \text{\, \rm or \,} y \le m \\
1 - P(m+1, m+2) & x = y = m+1 \\
P(x, y) & \text{\rm otherwise}
\end{cases}
\ee
and stationary distribution
\be
\pi_2(x) = \frac{\pi(x)}{c_m}
\ee
where
\be
c_m = \sum_{x=m+1}^{n} \pi(x) = 1 - b_m.
\ee
}
\end{itemize}

We will take $m = \theta n$ where $0 < \theta < 0.49$ (we could in principle just have $\theta<3/4$ but this restriction makes the calculations simpler; see \lemref{StatDistribExpSmall} for the origin of the upper bound on $\theta$).  Note that $P_2$ is the same as phase 3 used in the direct mixing time proof (up to relabelling $m+1$ to $m$).  Therefore we already have, from the proof of \lemref{MixingPhase3}, that the gap $\Delta_2$ of $P_2$ is $\Omega(1/n)$.  To find the gap of the whole zero chain we need to find the gaps $\Delta_1$ and $\bar{\Delta}$.  An ingredient to proving this is an exponential bound on the tail of the stationary distribution:
\begin{lemma}
\label{lem:StatDistribExpSmall}
\be
\pi(\theta n) \le 2\l(\frac{1}{4} \l( \frac{3 e}{\theta} \r)^{\theta} \r)^n
\ee
and for $\theta < \theta_0 \approx 0.49$, $\pi(\theta n) = e^{-\Omega(n)}$.
\end{lemma}
\begin{proof}
Use ${n \choose k} \le \l(\frac{n e}{k}\r)^k$ and $4^n -1 \ge \frac{4^n}{2}$ to prove the bound.  When $\frac{1}{4} \l( \frac{3 e}{\theta} \r)^{\theta} < 1$, $\pi(\theta n)$ is exponentially small.  $\theta_0$ is the solution to $\frac{1}{4} \l( \frac{3 e}{\theta} \r)^{\theta} = 1$.
\end{proof}

From this we can bound the gap of $\bar{P}$:
\begin{lemma}
\label{lem:GapPBar}
The gap of $\bar{P}$ is $\Omega(1/n)$.
\end{lemma}
\begin{proof}
We first need to work out the transition matrix for $\bar{P}$.  From the definition of $\bar{P}$ in the decomposition theorem,
\bas
\bar{P}(1, 2) &= \frac{1}{b_m} \pi(m) P(m, m+1)\\
\bar{P}(2, 1) &= \frac{1}{c_m} \pi(m+1) P(m+1, m).
\eas
We can find the two diagonal elements using the fact that the transition matrix is stochastic.  Because the zero chain is reversible, $\bar{P}$ is also reversible and by direct calculation of the eigenvalues has gap
\be
\bar{\Delta} = 1 - | 1- (\bar{P}(1, 2) + \bar{P}(2, 1)) |.
\ee
However, we can remove the modulus signs since, for $n$ large enough, $\bar{P}(1, 2) + \bar{P}(2, 1) \le 1$.  This is because, using reversibility and $\pi(m) \le b_m$
\be
\bar{P}(1, 2) + \bar{P}(2, 1) = \frac{\pi(m) P(m, m+1)}{c_m b_m} \le \frac{P(m, m+1)}{c_m}.
\ee
Using \lemref{StatDistribExpSmall} we find that $c_m \ge 1- e^{-\Omega(n)}$ for $m = \theta n$ and $\theta < \theta_0$.  Using $P(m, m+1) \le 3/5$, we find that for $n$ large enough, $\bar{P}(1, 2) + \bar{P}(2, 1) \le 1$.

Now we need to show that $\bar{P}(1, 2) + \bar{P}(2, 1) = \Omega(1)$.  Again using \lemref{StatDistribExpSmall}, we find $\bar{P}(2, 1) = e^{-\Omega(n)}$.  We just need a bound on $\bar{P}(1, 2)$.  First we bound $\pi(m)/b_m$:
\bas
\frac{\pi(m)}{b_m} &= \frac{\pi(m)}{\sum_{x=1}^m \pi(x)} \\
&= \frac{1}{\sum_{x=1}^m \frac{\pi(x)}{\pi(m)}} \\
&= \frac{1}{\sum_{x=1}^m 3^{x-m} \frac{{n \choose x}}{{n \choose m}}} \\
&\ge \frac{1}{\sum_{x=1}^m 3^{x-m}} \\
&= \frac{1}{\sum_{x=0}^{m-1} 3^{-x}} \\
&\ge \frac{1}{\sum_{x=0}^{\infty} 3^{-x}} \\
&= \frac{2}{3}
\eas
using ${n \choose x} \le {n \choose m}$ for $x \le m \le n/2$.  For $m = \theta n$, we have $P(m, m+1) = \Omega(1)$ so overall $\bar{P}(1, 2) = \Omega(1)$, proving the bound on the gap.
\end{proof}

\begin{lemma}
\label{lem:GapP1}
The gap of $P_1$ is $\Omega(1/n)$.
\end{lemma}
\begin{proof}
We use the comparison method with length functions as stated in \thmref{ComparisonLengthFunction}.  The length function we choose is, for $x \le y$, $l(x, y) = r^x$ for some constant $r$ satisfying $0 < r < 1$.

Let
\be
A_z = \frac{1}{l(z, z+1) \pi_1(z) P_1(z, z+1)} \sum_{x=1}^z \sum_{y=z+1}^m \pi_1(x) \pi_1(y) \sum_{s = x}^{y-1} l(s, s+1).
\ee
Then, according to \thmref{ComparisonLengthFunction}, $\Delta_1 \ge 1/A$ where $A = \max_z A_z$.  We need to find an upper bound for $A$:

\bas
A_z &= \frac{1}{r^z \pi(z) P(z, z+1) b_m} \sum_{x=1}^z \sum_{y=z+1}^m \pi(x) \pi(y) \frac{r^x - r^y}{1-r} \\
&= \frac{1}{r^z (1-r) \pi(z) P(z, z+1) b_m} \l( (b_m - b_z) \sum_{x=1}^z r^x \pi(x) - b_z \sum_{x={z+1}}^m r^x \pi(x) \r) \\
&= \frac{1}{r^z (1-r) \pi(z) P(z, z+1) b_m} \l( (1- b_z - (1 -b_m)) \sum_{x=1}^z r^x \pi(x) - (1-(1-b_z)) \sum_{x={z+1}}^m r^x \pi(x) \r) \\
&= \frac{1}{r^z (1-r) \pi(z) P(z, z+1) b_m} \l( (1- b_z) \sum_{x=1}^m r^x \pi(x) - (1 -b_m) \sum_{x=1}^z r^x \pi(x) -  \sum_{x={z+1}}^m r^x \pi(x) \r) \\
&\le \frac{1}{r^z (1-r) \pi(z) P(z, z+1) b_m} \l( \sum_{x=1}^m r^x \pi(x) - (1 -b_m) \sum_{x=1}^z r^x \pi(x) -  \sum_{x={z+1}}^m r^x \pi(x) \r) \\
&= \frac{\sum_{x=1}^z r^x \pi(x)}{r^z (1-r) \pi(z) P(z, z+1)}
\eas
where the inequality comes from $1-b_z\le 1$.  Now let $h_z(r) = \frac{1}{r^z \pi(z)} \sum_{x=1}^z r^x \pi(x)$ then, plugging in the value of $P(z, z+1)$, we find
\bes
A_z \le \frac{5n(n-1)}{6z(n-z)(1-r)} h_z(r).
\ees
Now, $\max_z \frac{5n(n-1)}{6z(n-z)(1-r)} = O(n)$ so showing $\max_z h_z(r) = O(1)$ is sufficient the prove the bound we require.  We evaluate $h_z$ recursively:

Firstly,
\bes
h_z(r) = 1 + \frac{1}{r^z \pi(z)} \sum_{x=1}^{z-1} r^x \pi(x).
\ees
Then evaluate the sum:
\bas
\sum_{x=1}^{z-1} r^x \pi(x) &= \sum_{x=1}^{z-1} \frac{r^x \pi(x)}{r^{x+1} \pi(x+1)} r^{x+1} \pi(x+1) \\
&= \frac{1}{3r} \sum_{x=1}^{z-1} \frac{x+1}{n-x} r^{x+1} \pi(x+1) \\
&= \frac{1}{3r} \sum_{x=2}^{z} \frac{x}{n-x+1} r^{x} \pi(x) \\
&< \frac{1}{3r} \sum_{x=1}^{z} \frac{x}{n-x+1} r^{x} \pi(x) \\
&< \frac{1}{3r} \frac{z}{n-z+1} \sum_{x=1}^{z} r^{x} \pi(x).
\eas
Combining,
\bes
h_z(r) < 1 + \frac{z}{3r(n-z+1)} h_z(r)
\ees
or
\bes
h_z(r) < \frac{1}{1-\frac{z}{3r(n-z+1)}}.
\ees
Since $1 \le z \le \theta n$, $h_z(r)$ is constant in this range, proving the result.
\end{proof}

We can now combine the results to prove the bound on the zero chain gap:
\begin{proof}[Proof of \thmref{ZeroChainGap}]
Using Lemmas \ref{lem:GapP1}, \ref{lem:MixingPhase3} and \ref{lem:GapPBar} together with \thmref{Decomposition} proves the result.
\end{proof}

\section{Main Result}
\label{sec:MainResult}

We will now show how the mixing time results imply that we have an approximate 2-design.

\begin{proof}[Proof of \thmref{Main2Design}:]
We will go via the 2-norm since this gives a tight bound when working
with the Pauli operators.  We write $\rho$ in the Pauli
basis as usual (as \eq{PauliBasisGeneral}) and note that $\rho$ is not necessarily a physical state so the coefficients may not be real.
\begin{align*}
\vectornorm{\cG_{W} - \cG_H}_{\diamond}^2 &= \sup_\rho \frac{1}{\vectornorm{\rho}_1^2} \vectornorm{(\cG_{W} \otimes \text{\rm id}_{2^{2n}}) (\rho) - (\cG_H \otimes \text{\rm id}_{2^{2n}}) (\rho)}_1^2 \\
&\le 2^{4n} \sup_\rho \frac{1}{\vectornorm{\rho}_1^2} \vectornorm{(\cG_{W} \otimes \text{\rm id}_{2^{2n}}) (\rho) - (\cG_H \otimes \text{\rm id}_{2^{2n}}) (\rho)}_2^2 \\
&= \sup_{\rho} \frac{1}{\vectornorm{\rho}_1^2} \bigg|\bigg| \sum_{p_1, p_2, p_3, p_4 \atop p_1 p_2 \ne 00} \gamma_0(p_1, p_2, p_3, p_4) ( \cG_{W}(\sigma_{p_1} \otimes \sigma_{p_2}) \otimes \sigma_{p_3} \otimes \sigma_{p_4}  \\
& \phantom{= \sum_{\rho} \bigg|\bigg|} -  \cG_H(\sigma_{p_1} \otimes \sigma_{p_2}) \otimes \sigma_{p_3} \otimes \sigma_{p_4} ) \bigg|\bigg|_2^2
\end{align*}
Now, write (for $p_1 p_2 \ne 00$) $\cG_{W}(\frac{1}{2^n}\sigma_{p_1} \otimes \sigma_{p_2}) = \frac{1}{2^n} \sum_{q_1, q_2 \atop q_1 q_2 \ne 00} g_t(q_1, q_2; p_1, p_2) \sigma_{q_1} \otimes \sigma_{q_2}$.  We get
\begin{align*}
& \sup_{\rho} \frac{1}{\vectornorm{\rho}_1^2} \bigg|\bigg| \sum_{p_1, p_2, p_3, p_4, q_1, q_2 \atop p_1 p_2 \ne 00, q_1 q_2 \ne 00} \gamma_0(p_1, p_2, p_3, p_4) \left( g_t(q_1, q_2; p_1, p_2) - \frac{\delta_{q_1 q_2} \delta_{p_1 p_2}}{2^n(2^n+1)} \right) \\
& \phantom{\sup_{\rho} \bigg|\bigg|} \sigma_{q_1} \otimes \sigma_{q_2} \otimes \sigma_{p_3} \otimes \sigma_{p_4} \bigg|\bigg|_2^2 \\
&= 2^{4n} \sup_{\rho} \frac{1}{\vectornorm{\rho}_1^2} \sum_{p_1, p_2, p_3, p_4, q_1, q_2 \atop p_1 p_2 \ne 00, q_1 q_2 \ne 00} |\gamma_0(p_1, p_2, p_3, p_4)|^2 \left( g_t(q_1, q_2; p_1, p_2) - \frac{\delta_{q_1 q_2} \delta_{p_1 p_2}}{2^n(2^n+1)} \right)^2  \\
&\le 2^{4n} \eps^2 \sup_{\rho} \frac{\sum_{p_1, p_2, p_3, p_4 \atop p_1 p_2 \ne 00} |\gamma_0(p_1, p_2, p_3, p_4)|^2}{\vectornorm{\rho}_1^2} \\
&\le 2^{4n} \eps^2 \sup_{\rho} \frac{\vectornorm{\rho}_2^2}{\vectornorm{\rho}_1^2} \\
&= 2^{4n} \eps^2
\end{align*}
where the first equality comes from the orthogonality of the Pauli operators under the Hilbert-Schmidt inner product.  This proves the result for the diamond norm, \defref{ApproxUnitaryDesignDiamond}.  For the distance measure defined in \defref{ApproxUnitaryDesignDankert} (TWIRL), the argument in \cite{DCEL06} can be used together with the 1-norm bound to prove the result.
\end{proof}

\section{Conclusions}
\label{sec:ConclusionRandomCircuits}

We have proved tight convergence results for the first two moments of
a random circuit.  We have used this to show that random circuits are
efficient approximate 1- and 2-unitary designs.  Our framework readily
generalises to $k$-designs for any $k$ and the next step in this
research is to prove that random circuits give approximate $k$-designs
for all $k$.

We have shown that, provided the random circuit uses gates from a
universal gate set that is also universal on $U(4)$, the circuit is
still an efficient 2-design.  We also see that the random circuit with
gates chosen uniformly from $U(4)$ is the most natural model.  We note that the gates from $U(4)$ can be replaced by
gates from any approximate 2-design on two qubits without any change to the asymptotic 
convergence properties.

Finally, random circuits are interesting physical models in their own
right.  The original purpose of \cite{ODP06} was to answer the
physical question of how quickly entanglement grows in a system with
random two party interactions.  \lemref{MainMixing}(i) shows that
$O(n(n + \log 1/\eps))$ steps suffice (in contrast to $O(n^2(n + \log
1/\eps))$ which they prove) to give almost maximal entanglement in
such a system.

\section{Proofs}
\label{sec:RandomCircuitsProofs}

\subsection{Zero chain mixing time proofs}
\label{sec:ZeroChainMixingProof}

\subsubsection{Asymmetric Simple Random Walk}
\label{sec:AsymRandomWalks}

We will use some facts about asymmetric simple random walks i.e.~a random walk on a 1D line with probability $p$ of moving right at each step and probability $q=1-p$ of moving left.

The position of the walk after $k$ steps is tightly concentrated around $k(p-q)$:
\begin{lemma}
\label{lem:BiasedWalkPosition}
Let $X_k$ be the random variable giving the position of a random walk after $k$ steps starting at the origin with probability $p$ of moving right and probability $q = 1-p$ of moving left.  Let $\mu = p - q$.  Then for any $\eta > 0$,
\begin{equation*}
\Pr(X_k \ge \mu k + \eta) \le \exp\left(-\frac{\eta^2}{2 k}\right)
\end{equation*}
and
\begin{equation*}
\Pr(X_k \le \mu k - \eta) \le \exp\left(-\frac{\eta^2}{2 k}\right).
\end{equation*}
\end{lemma}
\begin{proof}
The standard Chernoff bound for $0/1$ variables $\tilde{Y}_i$ gives, with $\tilde{Y}_i$ equal to $1$ with probability p and for $Y_k = \sum_{i=1}^k \tilde{Y}_i$,
\begin{align*}
\Pr(Y_k \ge k p + \eta) &\le \exp\left(-\frac{2 \eta^2}{k} \right) \\
\Pr(Y_k \le k p - \eta) &\le \exp\left(-\frac{2 \eta^2}{k} \right).
\end{align*}
For our case, set $\tilde{Y_i} = 2\tilde{X_i} - 1$ to give the desired result.
\end{proof}
This result is for a walk with constant bias.  We will need a result for a walk with varying (but bounded from below) bias:
\begin{lemma}
\label{lem:BiasedWalkPositionWithBoundedBias}
Let $X_k$ be the random variable giving the position of a random walk after $k$ steps starting at the origin with probability $p_i \ge p$ of moving right and probability $q_i \le p$ of moving left at step $i$.  Let $\mu = p - (1-p)$.  Then for any $\eta > 0$,
\begin{equation*}
\Pr(X_k \ge \mu k + \eta) \le \exp\left(-\frac{\eta^2}{2 k}\right)
\end{equation*}
and
\begin{equation*}
\Pr(X_k \le \mu k - \eta) \le \exp\left(-\frac{\eta^2}{2 k}\right).
\end{equation*}
\end{lemma}
\begin{proof}
Let $\tilde{Y}_i$ be a random variable equal to $1$ with probability $p$ and $0$ with probability $1-p$.  Then let $\tilde{Z}_i$ be a random variable equal to $1$ with probability $p_i$ and $0$ with probability $1-p_i$.  Let $Y_k = \sum_{i=1}^k \tilde{Y}_i$ and $Z_k = \sum_{i=1}^k \tilde{Z}_i$.  Then following the standard Chernoff bound derivation (for $\lambda > 0$),
\begin{align*}
\Pr(Z_k \ge k p + \eta) &= \Pr\left(e^{\lambda Z_k} \ge e^{\lambda(kp+\eta)}\right) \\
&\le \frac{e^{\lambda(k p + \eta)}}{\Expect e^{\lambda Z_k}} \\
&\le \frac{e^{\lambda(k p + \eta)}}{\Expect e^{\lambda Y_k}} \\
&\le \exp\left(-\frac{2 \eta^2}{k} \right).
\end{align*}
We can then, as above, set $\tilde{Z_i} = 2\tilde{X_i} - 1$.  The calculation is similar for the bound on $\Pr(X_k \le \mu k - \eta)$.
\end{proof}

From \lemref{BiasedWalkPosition} we can prove a result about how often each site is visited.  If the walk runs for $t$ steps the walk is at position $t \mu$ with high probability so we might expect from symmetry that each site will have been visited about $1/\mu$ times.  Below is a weaker concentration result of this form but is strong enough for our purposes.  It says that the amount of time spent $\le x$ is about $x/\mu$.
\begin{lemma}
\label{lem:NumberOfHits}
For $\gamma > 2$ and integer $x > 0$,
\begin{equation*}
\Pr\left( \sum_{k=1}^\infty \mathbb{I}(X_k \le x) \ge \gamma x/\mu \right) \le 2 \exp\left(-\frac{\mu x(\gamma-2)}{2}\right),
\end{equation*}
where $\mathbb{I}$ is the indicator function.
\end{lemma}
\begin{proof}
Let $Y_k = \mathbb{I}(X_k \le x)$.  From \lemref{BiasedWalkPosition},
\begin{equation*}
\Pr(Y_k = 0) \le \exp \left( -\frac{(k \mu - x)^2}{2 k} \right)
\end{equation*}
for $k \le x/\mu$
and
\begin{equation*}
\Pr(Y_k = 1) \le \exp \left( -\frac{(k \mu - x)^2}{2 k} \right)
\end{equation*}
for $k \ge x/\mu$.

Then the quantity to evaluate is
\begin{equation*}
\Pr\left( \sum_{k=1}^\infty Y_k \ge \gamma x / \mu \right).
\end{equation*}
We use a standard trick to split this into two mutually exclusive possibilities and then bound the probabilities separately.  Write
\begin{multline}
\Pr\left( \sum_{k=1}^\infty Y_k \ge \gamma x / \mu \right) =
\Pr\left( \left(\sum_{k=1}^\infty Y_k \ge \gamma x / \mu\right) \bigcap \left( \bigcap_{j=1}^{\gamma x/\mu} \left[Y_j=1\right]\right)\right) + \\ \Pr\left( \left(\sum_{k=1}^\infty Y_k \ge \gamma x / \mu\right) \bigcap \left( \bigcup_{j=1}^{\gamma x/\mu} \left[Y_j=0\right]\right)\right).
\end{multline}
We can bound the first term:
\begin{align*}
\Pr\left( \left(\sum_{k=1}^\infty Y_k \ge \gamma x / \mu\right) \bigcap \left( \bigcap_{j=1}^{\gamma x/\mu} \left[Y_j=1\right]\right)\right) &= \Pr\left(\bigcap_{k=1}^{\gamma x/\mu} Y_k=1\right) \\
&\le \Pr\left( Y_{\gamma x/\mu} = 1 \right) \\
&\le \exp\left( -\frac{\mu x(\gamma-1)^2 }{2\gamma}\right)\\
&\le \exp\left( -\frac{\mu x(\gamma -2)}{2} \right)
\end{align*}
The second term similarly:
\begin{align*}
\Pr\left( \left(\sum_{k=1}^\infty Y_k \ge \gamma x / \mu\right) \bigcap \left( \bigcup_{j=1}^{\gamma x/\mu} \left[Y_j=0\right]\right)\right) &\le \Pr\left( \bigcup_{k=\frac{\gamma x}{\mu}+1}^\infty \left[Y_k=1\right] \right) \\
&\le \sum_{k=\frac{\gamma x}{\mu}+1}^\infty \Pr\left(Y_k=1\right) \\
&\le \sum_{k=\frac{\gamma x}{\mu}+1}^\infty \exp\left( -\frac{(k \mu - x)^2}{2k}\right) \\
&\le \exp\left(-\frac{\mu x (\gamma-2)}{2} \right)\qedhere.
\end{align*}
\end{proof}
The last fact we need about asymmetric simple random walks is a bound on the probability of going backwards.  If $p>q$ then we expect the walk to go right in the majority of steps.  The probability of going left a distance $a$ is exponentially small in $a$.  This is a well known result, often stated as part of the gambler's ruin problem:
\begin{lemma}[See e.g.~\cite{GrimmettWelsh86}]
\label{lem:NotBackwards}
Consider an asymmetric simple random walk that starts at $a>0$ and has an absorbing barrier at the origin.  The probability that the walk eventually absorbs at the origin is $1$ if $p \le q$ and $\left( q/p \right)^a$ otherwise.
\end{lemma}
This result is for infinitely many steps.  If we only consider finitely many steps, the probability of absorption must be at most this.

\subsubsection{Waiting Time}

From above we saw that the probability of moving is at least $2x/5n$ when at position $x$.  The length of time spent waiting at each step is therefore stochastically dominated by a geometric distribution with parameter $2x/5n$.  The following concentration result will be used to bound the waiting time (in our case $\beta=2/5$):
\begin{lemma}
\label{lem:WaitingConc}
Let the waiting time at each site be $W(x) \sim Geo\left(\beta x/n\right)$, the total waiting time $W = \sum_{x=1}^t W(x)$ and $t' = \frac{n \ln t}{\beta}$.  Then
\begin{equation*}
\Pr(W \ge C t') \le 2 t^{(1-C)/2}.
\end{equation*}
\end{lemma}
\begin{proof}
By Markov's inequality for $\lambda > 0$,
\begin{equation*}
\Pr(W \ge C t') \le \frac{\mathbb{E} e^{\lambda W}}{e^{\lambda C t'}}.
\end{equation*}
The $W(x)$ are independent so
\begin{equation*}
\mathbb{E}e^{\lambda W} = \prod_{x=1}^t \mathbb{E} e^{\lambda W(x)}.
\end{equation*}
Summing the geometric series we find
\begin{equation*}
\mathbb{E} e^{\lambda W(x)} = \frac{\frac{\beta x}{n}}{e^{-\lambda}-1+\frac{\beta x}{n}}
\end{equation*}
provided $e^{\lambda} < \frac{1}{1-\frac{\beta x}{n}}$ for all $1 \le x \le t$.  Therefore $e^\lambda$ is of the form $\frac{1}{1-\frac{\alpha \beta}{n}}$ where $0 < \alpha < 1$.  With this,
\begin{equation*}
\mathbb{E} e^{\lambda W(x)} = \frac{x}{x-\alpha}
\end{equation*}
and
\begin{equation*}
\mathbb{E} e^{\lambda W} = \frac{t! \Gamma(1-\alpha)}{\Gamma(t+1-\alpha)}.
\end{equation*}
We are free to choose $\alpha$ within its range to optimise the bound.  However, for simplicity, we will choose $\alpha = 1/2$.  From \lemref{GammaMGF},
\begin{equation*}
\mathbb{E} e^{\lambda W} \le 2 \sqrt{t}.
\end{equation*}
The result follows, using the inequality $1-x \le e^{-x}$.
\end{proof}

\subsubsection{Phase 1}

Here we prove that phase 1 completes successfully with high probability.  The bias here is large so the walk moves right every time with high probability:
\begin{lemma}
\label{lem:Phase1MovesRight}
The probability that the accelerated chain moves right at each step, starting from $x=1$ for $t$ steps, is at least
\begin{equation*}
1 - t^2/n.
\end{equation*}
\end{lemma}
\begin{proof}
The probability of moving right at each step is
\bas
\prod_{x=1}^t \frac{3(n-x)}{3n-2x-1} &= \frac{(n-2)(n-3) \ldots (n-t)}{(n-5/3)(n-7/3) \ldots (n-(2t+1)/3)} \\
&\ge (1-2/n)(1-3/n) \ldots (1-t/n) \\
&\ge (1-t/n)^t \ge 1-t^2/n\qedhere
\eas
\end{proof}
Let $t=n^\delta$.  Provided $\delta < 1/2$ this probability is close to one.  Therefore, with high probability, the walk moves to $n^\delta$ in $n^\delta$ steps.  Using \lemref{WaitingConc} the waiting time can be bounded:
\begin{lemma}
\label{lem:Phase1WaitingTime}
Let $W^{(1)}$ be the waiting time during phase 1.  Let $H$ be the event that the walk moves right at each step.  Then
\begin{equation}
\Pr\left(W^{(1)} \ge C t' | H\right) \le 2 n^{\delta(1-C)/2}
\end{equation}
where $t' = \frac{5 \delta n \ln n}{2}$.
\end{lemma}
\begin{proof}
This follows directly from \lemref{WaitingConc}, since each site is hit exactly once.
\end{proof}

We now combine these two lemmas to prove that phase 1 completes successfully with high probability:
\begin{proof}[Proof of \lemref{Phase1Completes}]
In \lemref{Phase1MovesRight}, we show that in $n^\delta$ accelerated steps, the walk moves right at each step with probability $\ge 1 - n^{2\delta-1}$.  Call this event $H$.  Then $\Pr(H) \ge 1 - n^{2\delta - 1}$.  \lemref{Phase1WaitingTime} shows that the waiting time $W^{(1)}$ is bounded with high probability (choosing $C=3$):
\begin{equation*}
\Pr(W^{(1)} \le 15 n \delta \ln n/2 | H) \ge 1-2n^{-\delta}.
\end{equation*}
Then we can bound the probability of phase 1 completing successfully:
\begin{align*}
\Pr(\text{Phase 1 completes successfully}) &\ge \Pr(H \cap W^{(1)} \le 15 n \delta \ln n/2 ) \\
&= \Pr(H) \Pr(W^{(1)} \le 15 n \delta \ln n/2 | H) \\
&\ge (1-n^{2\delta - 1}) (1 - 2n^{-\delta}) \\
&\ge 1 - n^{2\delta-1} - 2 n^{-\delta}.\qedhere
\end{align*}
\end{proof}

\subsubsection{Phase 2}

Phase 2 starts at $n^\delta/2$ and finishes when the walk has reached $\theta n$ for some constant $0<\theta<3/4$.  We show that, with high probability, this also takes time $O(n \log n)$.  The probability of moving right during this phase is at least $p = \frac{3(1-\theta)}{3-2\theta}$.  We first define some constants that we will derive bounds in terms of.  Let $\gamma$ be a constant $>2$.  Let $\mu = p - (1-p)$ and $\tilde{\mu} = \mu/\gamma$.  Finally let $s = \tilde{\mu} t$ for some $t$ (which will be the number of accelerated steps).  Then, with high probability, the walk will have passed $s$ after $t$ steps:
\begin{lemma}
\label{lem:Phase2GetsToRightPlace}
Let $X_t$ be the position of the walk at accelerated step $t$, where $X_0 = n^\delta$.  Then
\begin{equation*}
\Pr(X_t \le s) \le \exp(-\mu^2 t (1-1/\gamma)^2/2).
\end{equation*}
\end{lemma}
\begin{proof}
Let $X_t' = X_t - n^\delta$.  Then from \lemref{BiasedWalkPositionWithBoundedBias},
\begin{equation*}
\Pr(X_t' \le \mu t - \eta) \le \exp\left(-\frac{\eta^2}{2 t}\right).
\end{equation*}
Now let $\eta = \mu t - s$ and use
\begin{align*}
\Pr(X_t \le s) &= \Pr(X_t' \le s - n^\delta) \\
&\le \Pr(X_t' \le s)
\end{align*}
to complete the proof.
\end{proof}
We now prove a bound on the waiting time:
\begin{lemma}
\label{lem:Phase2WaitingTime}
Let $W^{(2)}$ be the waiting time in phase 2.  Then, assuming the walk does not go back beyond $n^\delta/2$,
\begin{equation}
\Pr\left(W^{(2)} \ge \frac{15 n \ln s}{\mu} \right) \le (4/s)^{3/2\mu} + 
\frac{2 \exp \left(\frac{-\mu n^\delta}{4}\right)}{1-\exp \left(\frac{-\mu}{2} \right)}.
\end{equation}
\end{lemma}
\begin{proof}
Let $W_k \sim Geo\left(\frac{2 X_k}{5n}\right)$ where $X_k$ is the position of the walk at accelerated step $k$ ($X_0 = n^\delta$).  We want to bound (w.h.p.) the waiting time $W^{(2)} = \sum_{k=1}^t W_k$ of $t$ steps of the accelerated walk.  

Define the event $H$ to be
\begin{equation}
H = \left\{ \bigcap_{x \ge n^\delta/2} \left[ \sum_{k=1}^{\infty} \mathbb{I}(X_k \le x) \le x/\tilde{\mu} \right] \right\}.
\end{equation}
If $H$ occurs, no sites have been hit too often and the walk has not gone back further than $n^\delta/2$.  It is important that we also use the restriction that $X_k \ge n^\delta/2$ because the waiting time grows the longer the walk moves back.  However, it is very unlikely that the walk will go backwards (even to $n^\delta/2$).

We now define some more notation to bound the waiting time.  Let $\mathbf{X} = \\ (X_1, X_2, \ldots, X_t)$ be a tuple of positions and let $N_x(\mathbf{X})$ be the number of times that $x$ appears in $\mathbf{X}$ and let $\mathbf{N}(\mathbf{X}) = (N_1(\mathbf{X}), N_2(\mathbf{X}), \ldots, N_n(\mathbf{X}))$.  Then we have $\sum_x N_x(\mathbf{X}) = t$.

As we said above, the waiting time at $x=a$ stochastically dominates the waiting time at $x=b$ for $b \ge a$.  In other words,
\begin{equation}
\label{eq:WaitingTimeStocDom}
W_k \stocdoml W_{k'} \text{ if } X_k \le X_{k'}
\end{equation}
where $X \stocdoml Y$ means that $X$ stochastically dominates $Y$.  Now write the waiting time for all steps
\begin{align}
W^{(2)}(\mathbf{X}) &= \sum_{k=1}^t W_k \nonumber \\
&= \sum_x \sum_{h=1}^{N_x(\mathbf{X})} W_h(x)
\label{eq:WaitingTimeX}
\end{align}
where $W_h(x) \sim Geo\left(\frac{2 x}{5n}\right)$.

If event $H$ occurs, we can put some bounds on $N_x$.  We find that, for all $x \ge n^\delta/2$,
\begin{equation}
\sum_{y = n^\delta/2}^x N_y(\mathbf{X}) \le x/\tilde{\mu}
\end{equation}
and $N_x(\mathbf{X}) = 0$ for $x < n^\delta/2$.  Now let $\mathbf{X}_m$ be such that $N_{n^\delta/2}(\mathbf{X}_m) = \frac{n^\delta}{2 \tilde{\mu}}$ and $N_x(\mathbf{X}_m) = 1/\tilde{\mu}$ for $x > n^\delta/2$.  Then
\begin{equation}
\sum_{y = n^\delta/2}^x N_y(\mathbf{X}_m) = x/\tilde{\mu}.
\end{equation}
Now we introduce the relation $\majr$:
\begin{definition}
Let $\mathbf{x}$ and $\mathbf{y}$ be $n$-tuples.  Then $x \majr y$ if
\be
\sum_{i = 1}^k x_i \le \sum_{i=1}^k y_i
\ee
for all $1 \le k \le n$ with equality for $k=n$.
\end{definition}
Note that this is like majorisation, except the elements of the tuples are not sorted.
Using this, we find that $\mathbf{N}(\mathbf{X}) \majr \mathbf{N}(\mathbf{X}_m)$ (Using $\sum_y N_y(\mathbf{X}) = \sum_y N_y(\mathbf{X'}) = t$ for all $\mathbf{X}, \mathbf{X'}$.)

If we combine Equations \ref{eq:WaitingTimeStocDom} and \ref{eq:WaitingTimeX} we find that $W^{(2)}(\mathbf{X}) \stocdoml W^{(2)}(\mathbf{X}')$ if $\mathbf{N}(\mathbf{X}) \majl \mathbf{N}(\mathbf{X'})$.  Roughly speaking, this is simply saying that the waiting time is larger if the earlier sites are hit more often.  But since for all $\mathbf{X}$ that satisfy $H$, $\mathbf{X} \majr \mathbf{X}_m$, we have $W^{(2)}(\mathbf{X}) \stocdomr W^{(2)}(\mathbf{X}_m)$ provided $H$ occurs.  We will simplify further by noting that $\mathbf{X_m} \majr \mathbf{X}_0$ where $N_x(\mathbf{X}_0) = 1/\tilde{\mu}$ for $1 \le x \le \tilde{\mu} t = s$ and zero elsewhere.  Therefore
\begin{equation*}
\Pr\left(W^{(2)}(\mathbf{X}) \ge \frac{5C n \ln s}{2\tilde{\mu}} \bigg| H\right) \le \Pr\left(W^{(2)}(\mathbf{X_0}) \ge \frac{5C n \ln s}{2\tilde{\mu}} \right).
\end{equation*}
We can bound this by applying \lemref{WaitingConc}.  Let $W_h = \sum_{x=1}^{s} W_h(x)$.  From \lemref{WaitingConc},
\begin{equation}
\Pr(W_h \ge C t') \le 2 s^{\frac{1-C}{2}}
\end{equation}
where $t' = \frac{5n \ln s}{2}$.  However, we want a bound on $\Pr\left(\sum_{h=1}^{1/\tilde{\mu}} W_h \ge C t' / \tilde{\mu}\right)$.  The same reasoning as in \lemref{WaitingConc} bounds this as
\begin{equation}
\Pr\left(\sum_{h=1}^{1/\tilde{\mu}} W_h \ge C t' / \tilde{\mu}\right) \le \left(2 s^{\frac{1-C}{2}}\right)^{1/\tilde{\mu}}.
\end{equation}
Therefore 
\begin{equation}
\Pr\left(W^{(2)}(\mathbf{X_0}) \ge \frac{5C n \ln s}{2\tilde{\mu}} \right) \le 2^{1/\tilde{\mu}} s^{\frac{(1-C)/2}{\tilde{\mu}}}.
\end{equation}

To complete the proof, we just need to find $\Pr(H^c)$.  We can bound it using the union bound and \lemref{NumberOfHits}:

\begin{align*}
\Pr(H^c) &= \Pr\left(\bigcup_{x=n^\delta/2}^{n} \left[\sum_{k=1}^\infty \mathbb{I}(X_k \le x) > x/\tilde{\mu}\right] \right) \\
&\le \sum_{x=n^\delta/2}^n \Pr\left(\sum_{k=1}^\infty \mathbb{I}(X_k \le x) \ge x/\tilde{\mu}\right) \\
&\le \sum_{x=n^\delta/2}^n 2 \exp\left(\frac{-\mu x(\gamma -2)}{2}\right) \\
&\le \sum_{x=n^\delta/2}^\infty 2 \exp\left(\frac{-\mu x(\gamma -2)}{2}\right) \\
&= \frac{2 \exp\left(\frac{-\mu n^\delta(\gamma-2)}{4}\right)}{1-\exp\left(\frac{-\mu (\gamma-2)}{2}\right)} \\
\end{align*}
Now, for any events $A$ and $B$
\begin{align*}
\Pr(A) &= \Pr(A \cap B) + \Pr(A \cap B^c) \\
&= \Pr(A | B) \Pr(B) + \Pr(A \cap B^c) \\
&\le \Pr(A | B) + \Pr(B^c)
\end{align*}
and set $C=2$ and $\gamma=3$ to obtain the result.
\end{proof}

We now combine these two lemmas to prove that phase 2 completes successfully with high probability:
\begin{proof}[Proof of \lemref{Phase2Completes}]
Phase 2 can fail if:
\begin{itemize}
\item{The walk does not reach $\theta n$.  The probability of this is bounded by \lemref{Phase2GetsToRightPlace}:
\begin{equation*}
\Pr(X_t \le \theta n) \le \exp\left(-\frac{2}{3} \mu \theta n\right).
\end{equation*}
This follows from setting $t = \frac{3 \theta n}{\mu}$ and $\gamma = 3$.}
\item{The waiting time is too long.  This probability is bounded by \lemref{Phase2WaitingTime}:
\begin{equation*}
\Pr\left(W^{(2)} \ge \frac{15n \ln(\theta n)}{\mu}\right) \le \left(\frac{4}{\theta n}\right)^\frac{3}{2\mu} + \frac{2 \exp\left(\frac{-\mu n^\delta}{4}\right)}{1-\exp(-\mu/2)} + (q/p)^{n^\delta/2}.
\end{equation*}}
\item{The walk gets back to $n^\delta/2$.  This is bounded by \lemref{NotBackwards}:
\begin{equation*}
\Pr\left(\text{Walk gets to $n^\delta/2$}\right) \le \left(q/p \right)^{n^\delta/2}.
\end{equation*}}
\end{itemize}
So, using the union bound we can bound the overall probability of failure:
\begin{equation*}
\Pr(\text{Phase 2 fails}) \le \exp\left(-\frac{2}{3} \mu \theta n\right) + \left(\frac{4}{\theta n}\right)^\frac{3}{2\mu} + \frac{2 \exp\left(\frac{-\mu n^\delta}{4}\right)}{1-\exp(-\mu/2)} + \left(q/p \right)^{n^\delta/2}.\qedhere
\end{equation*}
\end{proof}

\subsubsection{Phase 3}

This phase starts at $\theta n$.  We show that this mixes quickly using log-Sobolev arguments.  

\begin{lemma}
\label{lem:MixingPhase3}
The zero chain on the restricted state space $x \in [m, n]$ where $m = \theta n$ for $\theta > 0$ has mixing time $O\left(n \log \frac{n}{\eps}\right)$.
\end{lemma}
\begin{proof}
We restrict the Markov chain to only run from $m$ by adjusting the holding probability at $m$, $P(m, m)$.  Construct the chain $P'$ with transition matrix
\begin{equation}
P'(x, y)=
\begin{cases}
0 & x < m \, \text{or} \, y < m \\
1 - P(m, m+1) & x = y = m \\
P(x, y) & \text{otherwise}
\end{cases}
\end{equation}
where $P$ is the transition matrix of the full zero chain.  This chain then has stationary distribution
\begin{equation}
\pi'(x)=
\begin{cases}
\pi(x)/(1-w) & m \le x \le n \\
0 & \text{otherwise}
\end{cases}
\end{equation}
where $w = \sum_{x=1}^{m-1} \pi(x)$.  To see this, first note that the distribution is normalised.  We want to show that 
\begin{equation}
\label{eq:Phase3StatDistribDef}
\sum_{x = m}^n P'(x, y) \pi'(x) = \pi'(y).
\end{equation}
When $y = m$ we are required to prove that $P'(m, m) \pi'(m) + P'(m+1, m) \pi'(m+1) = \pi'(m)$.  This follows from the reversibility of the unrestricted zero chain, using $P'(m, m) = 1-P(m, m+1)$.  For $y>m$, \eq{Phase3StatDistribDef} is satisfied simply because $\pi(x)$ is the stationary distribution of $P$ and related by a constant factor to $\pi'(x)$.

We can now prove this final mixing time result, making use of \lemref{ProductChain}.  Let $Q_i$ be the chain that uniformly mixes site $i$.  This converges in one step and has a log-Sobolev constant independent of $n$; call it $\rho_1$.  Let $Q$ be the chain that chooses a site at random and then uniformly mixes that site.  This is the product chain of the $Q_i$ so, by \lemref{ProductChain}, has gap $1/n$ and Sobolev constant $\rho_Q = \rho_1/n$.  We can construct the zero chain for this and find its Sobolev constant.

The Sobolev constant is defined (\defref{LogSobolev}) in terms of a minimisation over functions on the state space.  For the chain $Q$ we can write
\begin{equation*}
\rho_Q = \inf_{\phi} f(\phi).
\end{equation*}
If we restrict the infimum to be over functions $\phi$ with $\phi(x) = \phi(y)$ for $x$ and $y$ containing the same number of zeroes then we obtain the Sobolev constant for the zero-Q chain, $\rho_{Q_0}$, which is the chain which counts the number of zeroes in the full chain Q.  Since taking the infimum over less functions cannot give a smaller value,
\begin{equation*}
\rho_{Q_0} \ge \rho_Q \ge \rho_1/n.
\end{equation*}
We can now compare this chain to the zero-$P$ chain.  The stationary distributions are the same.  The transition matrix for the zero-$Q$ chain is
\begin{equation*}
Q_0(x,y) = 
\begin{cases}
\frac{n+2x}{4n} & y = x \\
\frac{x}{4n} & y = x-1 \\
\frac{3(n-x)}{4n} & y = x+1 \\
0	&	{\rm otherwise} \\
\end{cases}
\end{equation*}
Then construct $Q_0'$ by restricting the space to only run from $m$ in exactly the same way as $P'$ is constructed from $P$.  $Q_0'$ has the same stationary distribution as $P'$.  Now we can perform the comparison.  From \eq{ComparisonWalk}:
\begin{align*}
A &= \max_{a \ge m} \frac{Q_0'(a,a+1)}{P'(a,a+1)} \\
&= \max_{a \ge m} \frac{5(n-1)}{8a} \le \frac{5}{8\theta}.
\end{align*}
Therefore $\rho_{P'} \ge \frac{8 \theta \rho_1}{5n}$.  Exactly the same argument applies to show the gap is $\Omega(1/n)$ so the mixing time is (from \eq{SobolevMixingTime}) $O(n \log \frac{n}{\eps})$.
\end{proof}

Now we can prove that phase 3 completes successfully with high probability:
\begin{proof}[Proof of \lemref{Phase3Completes}]
In \lemref{MixingPhase3}, we show that after $O\left(n \log \frac{n}{\eps}\right)$ steps the chain mixes to distance $\eps$.  We just need to show that the walk goes back to $\theta n/2$ with small probability.  This follows from \lemref{NotBackwards}.
\end{proof}

\subsection{Moment Generating Function Calculations}

The following lemma is needed in the moment generating function calculations.
\begin{lemma}
\label{lem:GammaMGF}
For Integer $s > 0$,
\begin{equation}
\label{eq:GammaS}
\frac{\Gamma(s+1)\Gamma(1/2)}{\Gamma(s+1/2)} \le 2 \sqrt{s}
\end{equation}
\end{lemma}
\begin{proof}
From expanding the $\Gamma$ functions, \eq{GammaS} becomes
\begin{align*}
\frac{s! 2^s}{(2s-1)!!} &= \frac{2 \times 4 \times 6 \times \ldots \times 2(s-1) \times 2s}{1 \times 3 \times 5 \times \ldots \times (2s-3) \times (2s-1)} \\
&= \prod_{x=1}^s \frac{2x}{2x-1}
\end{align*}
We then proceed by induction.  $\prod_{x=1}^1 \frac{2x}{2x-1} = 2$ and by the inductive hypothesis
\begin{equation*}
\prod_{x=1}^{s+1} \frac{2x}{2x-1} \le \frac{2(s+1)}{2(s+1)-1} 2\sqrt{s}.
\end{equation*}
It is easy to show that $\frac{2(s+1)}{2(s+1)-1} \le \sqrt{\frac{s+1}{s}}$ and the result follows.
\end{proof}

\subsection{Mixing Times}

We find bounds for the mixing time above that are valid with high probability.  Below we turn these into full mixing time bounds.
\begin{lemma}
\label{lem:DistanceToStat}
If after $O(n \log n)$ steps the state $v$ of a random walk satisfies
\begin{equation*}
|| v - \pi || \le \delta
\end{equation*}
where $\pi$ is the stationary distribution and $\delta$ is $1/poly(n)$ then the number of steps required to be at most a distance $\eps$ from stationarity is
\begin{equation*}
O\left(n \log \frac{n}{\eps}\right).
\end{equation*}
\end{lemma}
\begin{proof}
Let $s$ be the slowest mixing initial state.  Then, after $t = O(n \log n)$ steps we have at worst the state
\begin{equation*}
(1-\delta)\pi+\delta s
\end{equation*}
and if we repeat $k t$ times $\delta$ becomes $\delta^k$.  So to get a distance $\eps$, $k = \left\lceil\frac{\log \eps}{\log \delta}\right\rceil$.

Now we evaluate the mixing time:
\begin{align*}
k t = O(n \log n) \left\lceil\frac{\log \eps}{\log \delta}\right\rceil &= O(n \log n) \left\lceil\frac{\log 1/\eps}{\log 1/\delta}\right\rceil \\
&= O(n \max(\log n, \log 1/\eps)) \\
&= O\left(n \log \frac{n}{\eps}\right).
\qedhere
\end{align*}
\end{proof}

\chapter{Quantum Tensor Product Expanders and an Efficient Unitary Design Construction}
\label{chap:TPE}

\section{Introduction}

In this chapter, we give an efficient construction of a unitary $k$-design on $n$
qubits for any
$k$ up to $O(n/\log(n))$.  We will do this by first finding an
efficient construction of a 
quantum \emph{$k$-copy tensor product expander} (k-TPE), which can
then be iterated to produce a $k$-design.  We will therefore need to
understand some of the theory of expanders before presenting our
construction.

Classical expander graphs have the property that a marker executing a
random walk on the graph will have a distribution close to the
stationary distribution after a small number of steps.  We consider a
generalisation of this, known as a $k$-tensor product expander (TPE)
and due to \cite{HastingsHarrow08}, to graphs that randomise $k$
different markers carrying out correlated random walks on the same
graph.  This is a stronger requirement than for a normal ($k=1$)
expander because the correlations between walkers (unless they start
at the same position) must be broken.  We then generalise quantum
expanders in the same way, so that the unitaries act on $k$ copies of
the system.  We give an efficient construction of a quantum $k$-TPE
which uses an efficient classical $k$-TPE as its main ingredient.  We
then give as a key application the first efficient construction of a
unitary $k$-design for any $k$.

While randomised constructions yield $k$-designs (by a modification of
Theorem 5 of \cite{TamperResistance}) and $k$-TPEs (when the dimension is
polynomially larger than $k$ \cite{HastingsHarrow08}) with near-optimal
parameters, these approaches are not efficient.  Previous efficient
constructions of $k$-designs were known only for $k=1,2$, and no
efficient constant-degree, constant-gap quantum $k$-TPEs were
previously known, except for the $k=1$ case corresponding to quantum
expanders \cite{QExpandersEntropyDifference,AmbainisSmith04,AramExpanders07,QuantumMargulisExpanders}.

In \secref{exp-def}, we will define quantum expanders and other key
terms.  Then in \secref{result} we will describe our main result which
will be proved in \secref{proof}.  In this chapter, we will use $N$ to denote the dimension rather than $d$ to be consistent with the rest of the quantum expander literature.

This chapter has been published previously as \cite{TPE} and is joint work with Aram Harrow.

\subsection{Quantum Expanders}\label{sec:exp-def}

We will only consider $D$-regular expander graphs here.  We can think
of a random walk on such a graph as selecting one of $D$ permutations
of the vertices randomly at each step.  We construct the permutations
as follows.  Label the vertices from $1$ to $N$.  Then label each edge
from $1$ to $D$ so that each edge label appears exactly once on the
incoming and outgoing edges of each vertex.  This gives a set of $D$
permutations. Choosing one of these permutations at random (for some
fixed probability distribution) then defines a random walk on the graph.

We now define a classical $k$-TPE:
\begin{definition}[\cite{HastingsHarrow08}]
Let $\nu$ be a probability distribution on $\cS_N$ with support on
$\leq D$ permutations.  Then $\nu$ is a classical $(N, D, \lambda, k)$-TPE if
\label{def:ClassicalTPE}
\be
\left\|\bbE_{\pi\sim\nu} \l[B(\pi)^{\ot k}\r] - \bbE_{\pi\sim\cS_N}
   \l[B(\pi)^{\ot k}\r]\right\|_\infty
 = \l\| \sum_{\pi\in\cS_N} \l(\nu(\pi) - \frac{1}{N!}\r)
B(\pi)^{\ot k}\r\|_\infty
 \le \lambda.
\ee
with $\lambda < 1$.
Here $\bbE_{\pi\sim\nu}$ means the expectation over $\pi$ drawn
according to
$\nu$
and $\bbE_{\pi\sim\cS_N}$ means the expectation over $\pi$ drawn
uniformly from $\cS_N$.
\end{definition}
Setting $k=1$ recovers the usual spectral definition of an expander.
Note that a $(N,D,\lambda,k)$-TPE is also
a $(N,D,\lambda,k')$-TPE for any $k'\leq k$.  The largest meaningful
value of $k$ is $k=N$, corresponding to the case when $\nu$ describes
a Cayley graph expander on $\cS_N$.

The degree of the map is $D=|\supp \nu|$ and the gap is $1-\lambda$.
 Ideally, the degree should be small and gap large.  To be useful,
 these should normally be independent of $N$ and possibly $k$.  We say
 that a TPE construction is efficient if it can be implemented in
 $\poly (\log N)$ steps.  There are known constructions of efficient
 classical TPEs.  The construction of Hoory and Brodsky
 \cite{HooryBrodsky04} provides an expander with $D = \poly (\log N)$
 and $\lambda = 1-1/\poly(k, \log N)$ with efficient running time.  An
 efficient TPE construction is also known, due to Kassabov
 \cite{Kassabov05}, which has constant degree and gap (independent of
 $N$ and $k$).   

Similarly, we define a quantum $k$-TPE:
\begin{definition}[\cite{HastingsHarrow08}]
Let $\nu$ be a distribution on $\cU(N)$, the group of $N\times N$ unitary matrices, with $D=|\supp\nu|$.  Then $\nu$ is a quantum $(N, D, \lambda, k)$-TPE if
\label{def:QuantumTPE}
\be
\left\|\bbE_{U\sim \nu} \l[U^{\ot k,k}\r] - \bbE_{U \sim \cU(N)}
\l[U^{\ot k,k}\r] \right\|_\infty \le \lambda
\ee
with $\lambda < 1$.  Here $\bbE_{U\sim \cU(N)}$ means the expectation over $U$ drawn from the Haar measure.
\end{definition}
Again, normally we want $D$ and $\lambda$ to be constants and setting $k=1$ recovers the usual definition of a quantum expander.  Note that an equivalent statement of the above definition is that, for all $\rho$,
\be
\l\| \bbE_{U\sim\nu} \l[U^{\ot k} \rho (U^\dagger)^{\ot k}\r] - 
\bbE_{U\sim\cU(N)} \l[U^{\ot k} \rho (U^\dagger)^{\ot k}\r] 
\right\|_2 \le \lambda \norm{\rho}_2
\ee
%We will often refer to the map $\frac{1}{D} \sum_{U \in \mathcal{U}} U^{\ot k} \rho (U^\dagger)^{\ot k}$, rather than the set $\mathcal{U}$, as the expander.

A natural application of this is to make an efficient unitary $k$-design.  The definition we use here is the same as for a $k$-TPE, except with closeness in the 1-norm rather than the $\infty$-norm.  This is given in \defref{ApproxUnitaryDesignkk} (TRACE).

We can make an $\eps$-approximate unitary $k$-design from a quantum $k$-TPE with $O(k \log N)$ overhead:
\begin{theorem}
\label{thm:DesignFromTPE}
If $\mathcal{U}$ is a quantum $(N, D, \lambda, k)$-TPE then iterating the map $m=\frac{1}{\log 1/\lambda} \log \frac{N^{2k}}{\eps}$ times gives an $\eps$-approximate unitary $k$-design according to \defref{ApproxUnitaryDesignkk} (TRACE) with $D^m$ unitaries.
\end{theorem}
\begin{proof}
Iterating the TPE $m$ times gives
\bes
\left\|\bbE_{U\sim \nu} [U^{\ot k,k}] - \bbE_{U \sim \cU(N)} 
[U^{\ot k,k}] \right\|_\infty \le \lambda^m
\ees
This implies that
\bes
\left\|\bbE_{U\sim \nu} [U^{\ot k,k}] - \bbE_{U \sim \cU(N)} 
[U^{\ot k,k}] \right\|_1 \le N^{2k} \lambda^m
\ees
We take $m$ such that $N^{2k} \lambda^m=\eps$ to give the result.
\end{proof}

\begin{corollary}\label{cor:k-designs}
A construction of an efficient quantum $(N, D, \lambda, k)$-TPE yields
an efficient approximate unitary $k$-design, provided $\lambda = 1 -
1/\poly (\log N)$.  Further, if $D$ and $\lambda$ are constants, the
number of unitaries in the design is $N^{(O(k))}$. 
\end{corollary}

Our approach to construct an efficient quantum $k$-TPE will be to take an efficient classical $2k$-TPE and mix it with a quantum Fourier transform.  The degree is thus only larger than the degree of the classical expander by one.  Since the quantum Fourier transform on $\bbC^N$ requires $\poly (\log N)$ time, it follows that if the classical expander is efficient then the quantum expander is as well.  The main technical difficulty is to show for suitable values of $k$ that the gap of the quantum TPE is not too much worse than the gap of the classical TPE.

A similar approach to ours was first used in \cite{HastingsHarrow08} to construct a quantum expander (i.e.~a 1-TPE) by mixing a classical 2-TPE with a phase.  However, regardless of the set of phases chosen, this approach will not yield quantum $k$-TPEs from classical $2k$-TPEs for any $k\geq 2$.

\subsection{Main Result} \label{sec:result}

Let $\omega=e^{2\pi i/N}$ and define the $N$-dimensional Fourier
transform to be
\be
\cF=\frac{1}{\sqrt{N}}\sum_{m=1}^N\sum_{n=1}^N
\omega^{mn}\ket{m}\bra{n}.
\ee
Define $\delta_\cF$ to be the
distribution on $\cU(N)$ consisting of a point mass on $\cF$.  Our
main result in this chapter is that mixing $\delta_\cF$ with a classical $2k$-TPE
yields a quantum $k$-TPE for appropriately chosen $k$ and $N$.

\begin{theorem}
\label{thm:mainresult}
Let
$\nu_C$ be a classical $(N, D, 1-\eps_C, 2k)$-TPE, and for $0<p<1$, define
$\nu_Q = p \nu_C + (1-p) \delta_\cF$.  Suppose that
\be
\eps_A := 1- 2(2k)^{4k}/\sqrt{N} > 0.
\ee
Then $\nu_Q$ is a quantum $(N, D+1, 1-\eps_Q, k)$-TPE where 
\be\eps_Q \geq \frac{\eps_A}{12}\min(p\eps_C,1-p) > 0
\label{eq:eps_Q-bound}\ee
The bound in \eq{eps_Q-bound} is
 optimised when $p=1/(1+\eps_C)$, in which case we have 
\be \eps_Q \geq \frac{\eps_A\eps_C}{24}.\ee
\end{theorem}

This means that any constant-degree, constant-gap classical
$2k$-TPE gives a quantum $k$-TPE with constant degree and gap. If the
the classical TPE is efficient then the quantum TPE is as
well.  Using \corref{k-designs}, we obtain approximate unitary
$k$-designs with polynomial-size circuits.

Unfortunately the construction does not work for all dimensions; we
require that $N = \Omega((2k)^{8k})$, so that $\eps_A$ is lower-bounded by a
positive constant.  However, in applications normally $k$ is
fixed.  An interesting open problem is to find a construction that
works for all dimensions, in particular a $k=\infty$ expander.  (Most
work on  $k=\infty$ TPEs so far has focused on the $N=2$
case \cite{BG06}.) 
We suspect our construction may work for $k$ as large as $cN$ for a
small constant $c$.  On the other hand, if $2k> N$ then the gap in
our construction drops to zero.

\section{Proof of Theorem 4.1.5} %\thmref{mainresult}}
\label{sec:proof}
\subsection{Proof overview}
First, we introduce some notation.
Define ${\cE}_{\cS_N}^{2k} = \bbE_{\pi\sim\cS_N}[B(\pi)^{\ot 2k}]$ and 
${\cE}_{\cU(N)}^{k} = \bbE_{U\sim\cU(N)}[U^{\ot k,k}]$.  These are
both projectors onto spaces which we label $V_{\cS_N}$ and
$V_{\cU(N)}$ respectively.  Since $V_{\cU(N)}\subset V_{\cS_N}$, it
follows that ${\cE}_{\cS_N}^{2k} - {\cE}_{\cU(N)}^{k} $ is a
projector onto the space $V_0:= V_{\cS_N}\cap V_{\cU(N)}^\perp$.
We also define ${\cE}_{\nu_C}^{2k} =
\bbE_{\pi\sim \nu_C} [B(\pi)^{\ot 2k}]$ and ${\cE}_{\nu_Q}^k =
\bbE_{U\sim \nu_Q}[U^{\ot k,k}]$.

The idea of our proof is to consider ${\cE}_{\nu_C}^{2k}$ a proxy
for ${\cE}_{\cS_N}^{2k}$; if $\lambda_C$ is small enough then this
is a reasonable approximation.  Then we can restrict our attention to
vectors in $V_0$, which we would like to
show all shrink substantially under the action of our expander.  This
in turn can be reduced to showing that $\cF^{\ot k,k}$ maps any vector in
$V_0$ to a vector that has $\Omega(1)$
amplitude in $V_{S_N}^\perp$.  This last step is the most technically
involved step of the chapter, and involves careful examination of the
different vectors making up $V_{\cS_N}$.

Thus, our proof reduces to two key Lemmas.  The first allows us to
substitute  ${\cE}_{\nu_C}^{2k}$ 
for ${\cE}_{\cS_N}^{2k}$ while keeping the gap constant.
\begin{lemma}[\cite{HastingsHarrow08} Lemma 1]
\label{lem:ProjUnitaryGap}
Let $\Pi$ be a projector and let $X$ and $Y$ be operators such that
$\|X\|_\infty\leq 1$, $\|Y\|_\infty\leq 1$, $\Pi X = X \Pi = \Pi$,
$\|(I-\Pi)X(I-\Pi)\|_\infty \leq 1 - \eps_C$ and 
$\|\Pi Y \Pi\|_\infty \leq 1 - \eps_A$.  Assume $0<\eps_C, \eps_A<1$.
Then for any $0<p<1$, $\|p X + (1-p)Y\|_\infty < 1$.  Specifically,
\be \| p X + (1-p)Y \|_\infty \leq 
 1-\frac{\eps_A}{12}\min(p\eps_C,1-p).
\label{eq:intermediate-norm}\ee
%Setting $p = 1 / (1 + \eps_C)$, we obtain
%\be \| p X + (1-p)Y \| \leq 1 - \frac{\eps_C\eps_A}{12(1+\eps_C)}
%\leq 1 - \frac{\eps_C\eps_A}{24}.
%\label{eq:mixed-norm}\ee
\end{lemma}

We will restrict to $V_{\cU(N)}^\perp$, or equivalently, subtract the
projector $\cE_{\cU(N)}^k$ from each operator.  Thus we have $X =
\cE_{\nu_C}^{2k} - \cE_{\cU(N)}^k$, $\Pi =
{\cE}_{\cS_N}^{2k} - {\cE}_{\cU(N)}^k$ and $Y=\cF^{\ot 
  k,k} - \cE_{\cU(N)}^k$.  
According to
\defref{ClassicalTPE}, we have the bound
 \be\| (I -\Pi) X (I-\Pi) \|_\infty =
 \| {\cE}_{\nu_C}^{2k} -  {\cE}_{\cS_N}^{2k}\|_\infty \le
 1-\eps_C.
\label{eq:c2k-TPE-bound}\ee
It will remain only to bound $\lambda_A := 1-\eps_A =
\l|\l|\l({\cE}_{\cS_N}^{2k} - 
  {\cE}_{\cU(N)}^k\r) \cF^{\ot k,k} \l({\cE}_{\cS_N}^{2k} -
  {\cE}_{\cU(N)}^k\r)\r|\r|_\infty$.
\begin{lemma}
\label{lem:Gap}
For $N \ge (2k)^2$,
\be
\label{eq:Gap}
\lambda_A = \l|\l|\l({\cE}_{\cS_N}^{2k} -
  {\cE}_{\cU(N)}^k\r) \cF^{\ot k,k} \l({\cE}_{\cS_N}^{2k} -
  {\cE}_{\cU(N)}^k\r)\r|\r|_\infty \leq 2(2k)^{4k}/\sqrt{N}.
\ee
\end{lemma}

Combining  \eq{c2k-TPE-bound}, \lemref{Gap} and
\lemref{ProjUnitaryGap} now completes the proof of
 \thmref{mainresult}.

\subsection{Action of a Classical \texorpdfstring{$2k$}{2k}-TPE}

We start by analysing the action of a classical $2k$-TPE.  (We
consider $2k$-TPEs rather than general $k$-TPEs since our quantum
expander construction only uses these.)  The fixed points are states which are
unchanged when acted on by $2k$ copies of any permutation matrix.
Since the same permutation is applied to all copies, any equal indices
will remain equal and any unequal indices will remain unequal.  This
allows us to identify the fixed points of the classical expander: they
are the sums over all states with the same equality and difference
constraints.  For example, for $k=1$ (corresponding to a 2-TPE), the
fixed points are $\sum_{n_1} \ket{n_1, n_1}$ and $\sum_{n_1 \ne
n_2}\ket{n_1, n_2}$ (all off-diagonal entries equal to 1).  In
general, there is a fixed point for each partition of the set $\{1, 2,
\ldots, 2k\}$ into at most $N$ non-empty parts.  If $N \ge 2k$, which
is the only case we consider, the $2k^{\text{th}}$ Bell number
$\beta_{2k}$ gives the number of such partitions (see
e.g.~\cite{EnumerativeCombinatorics}).

We now write down some more notation to further analyse this.  If
$\Pi$ is a partition of $\{1,\ldots,2k\}$, then we write $\Pi\vdash
2k$.  We will see that $\cE_{\cS_N}^{2k}$ projects onto a space
spanned by vectors labelled by partitions.  For a partition $\Pi$, say
that $(i, j) \in \Pi$ if and only if elements $i$ and $j$ are in the
same block. 
Now we can write down the fixed points of the classical expander.  Let
\be
I_\Pi = \{(n_1, \ldots, n_{2k}) : n_i = n_j \iff (i, j) \in \Pi \}.
\ee
This is a set of tuples where indices in the same block of $\Pi$ are
equal and indices in different blocks are not equal.  
The corresponding state is
\be
\ket{I_\Pi} = \frac{1}{\sqrt{|I_\Pi|}} \sum_{\bfn\in I_\Pi} \ket{\bfn}
\ee
where $\bfn = (n_1, \ldots, n_{2k})$.  Note that the $\{I_\Pi\}_{\Pi\vdash 2k}$ form a
partition $\{1,\ldots,N\}^{2k}$ and thus the
$\{\ket{I_\Pi}\}_{\Pi\vdash 2k}$ form an orthonormal basis for $V_{\cS_N}$.  This is because, when applying the same permutation to all indices, indices that are the same remain the same and indices that differ remain different.  This implies that \be{\cE}_{\cS_N}^{2k} = \sum_{\Pi\vdash 2k} \proj{I_\Pi}.\ee
To evaluate the normalisation, use
$| I_\Pi | = (N)_{| \Pi |}$ where $(N)_n$ is the falling factorial
$N(N-1) \ldots (N-n+1)$ and $|\Pi|$ is the number of
blocks in $\Pi$.  We will later find it useful to bound
$(N)_n$ with
\be \l(1-\frac{n^2}{2N}\r)N^n \leq (N)_n \leq N^n.\ee

We will also make use of the refinement partial order:
\begin{definition}
The refinement partial order $\le$ on partitions $\Pi, \Pi' \in \Par(2k, N)$ is given by
\be
\Pi \le \Pi'  \iff (i, j) \in \Pi \Rightarrow (i, j) \in \Pi'.
\ee
\end{definition}
For example, $\{\{1, 2\}, \{3\}, \{4\}\} \le \{\{1, 2, 4\}, \{3\}\}$.  Note that $\Pi \le \Pi'$ implies that $|\Pi| \ge |\Pi'|$.

\subsubsection{Turning Inequality Constraints into Equality Constraints.}

In the analysis, it will be easier to consider just equality constraints rather than both inequality and equality constraints as in $I_\Pi$.  Therefore we make analogous definitions:
\be
E_\Pi = \{(n_1, \ldots, n_{2k}) : n_i = n_j \forall (i, j) \in \Pi \}
\ee
and
\be
\ket{E_\Pi} = \frac{1}{\sqrt{|E_\Pi|}} \sum_{\bfn \in E_\Pi} \ket{\bfn}.
\ee
Then $| E_\Pi | = N^{|\Pi|}$.  For $E_\Pi$, indices in the same block are equal, as with $I_\Pi$, but indices in different blocks need not be different.

We will need relationships between $I_\Pi$ and $E_\Pi$.   First, observe that $E_\Pi$ can be written as the union of some $I_\Pi$ sets:
\be
E_\Pi = \bigcup_{\Pi' \ge \Pi} I_{\Pi'}.
\label{eq:E-I-set-rel}\ee
To see this, note that for $\bfn\in E_\Pi$, we have $n_i=n_j \forall (i,j)\in\Pi$, but we may also have an arbitrary number of additional equalities between $n_i$'s in different blocks.   The (unique) partition 
$\Pi'$ corresponding to these equalities has the property that $\Pi$ is a refinement of $\Pi'$; that is, $\Pi'\geq \Pi$.
Thus for any $\bfn\in E_\Pi$ there exists a unique $\Pi'\ge \Pi$ such
that $\bfn\in I_{\Pi'}$.  Conversely, whenever $\Pi'\geq \Pi$, we also
have $I_{\Pi'}\subseteq E_{\Pi'} \subseteq E_{\Pi}$ because each
inclusion is achieved only be relaxing constraints. 

Using \eq{E-I-set-rel}, we can obtain a useful identity involving sums
over partitions:
\be N^{|\Pi|} = |E_\Pi| = \sum_{\Pi'\geq \Pi} |I_{\Pi'}| =
\sum_{\Pi'\geq \Pi} N_{(|\Pi'|)}. \label{eq:stirling-rel}\ee
Additionally, since both sides in \eq{stirling-rel} are degree $|\Pi|$ polynomials and are equal on $\ge |\Pi|+1$ points (we can choose any $N$ in \eq{stirling-rel} with $N\geq 2k$), it
implies that $x^{|\Pi|} = \sum_{\Pi'\geq \Pi} x_{(\Pi')}$ as an identity
on formal polynomials in $x$.

The analogue of \eq{E-I-set-rel} for the states $\ket{E_\Pi}$ and $\ket{I_\Pi}$ is similar but has to account for normalisation factors.  Thus we have
\be \sqrt{|E_\Pi|} \ket{E_\Pi} =  \sum_{\Pi'\geq \Pi} \sqrt{|I_{\Pi'}|} \ket{I_{\Pi'}}.
\label{eq:E-I-state-rel}\ee

We would also like to invert this relation, and write $\ket{I_\Pi}$ as a sum over various $\ket{E_{\Pi'}}$.  Doing so will require introducing some more notation.  Define $\zeta(\Pi,\Pi')$ to be 1 if $\Pi \leq \Pi'$ and 0 if $\Pi \not\leq \Pi'$.  This can be thought of as a matrix that, with respect to the refinement ordering, has ones on the diagonal and is upper-triangular.  Thus it is also invertible.  Define $\mu(\Pi,\Pi')$ to be the matrix inverse of $\zeta$, meaning that for all $\Pi_1,\Pi_2$, we have 
$$\sum_{\Pi'\vdash 2k} \zeta(\Pi_1,\Pi') \mu(\Pi',\Pi_2) = 
\sum_{\Pi'\vdash 2k} \mu(\Pi_1,\Pi')\zeta(\Pi',\Pi_2)  = \delta_{\Pi_1,\Pi_2},$$
where $\delta_{\Pi_1,\Pi_2}=1$ if $\Pi_1=\Pi_2$ and $=0$ otherwise.
Thus, if we rewrite \eq{E-I-state-rel} as 
\be \sqrt{|E_\Pi|} \ket{E_\Pi} =  \sum_{\Pi'\vdash 2k} \zeta(\Pi,\Pi')
\sqrt{|I_{\Pi'}|} \ket{I_{\Pi'}},\ee
then we can use $\mu$ to express $\ket{I_\Pi}$ in terms of the $\ket{E_\Pi}$ as
\be \sqrt{|I_{\Pi}|} \ket{I_{\Pi}}=  \sum_{\Pi'\vdash 2k} \mu(\Pi,\Pi')
\sqrt{|E_{\Pi'}|} \ket{E_{\Pi'}}.
\label{eq:I-into-E}\ee

This approach is a generalisation of inclusion-exclusion known as
M\"obius inversion, and the function $\mu$ is called the M\"{o}bius
function (see Chapter 3 of \cite{EnumerativeCombinatorics} for more
background).  For the case of the refinement partial order, the
M\"obius function is known:
\begin{lemma}[\cite{RotaMobius}, Section 7]
\label{lem:MobiusInversionRefinement}
$$\mu(\Pi, \Pi') = (-1)^{|\Pi| - |\Pi'|} \prod_{i=1}^{| \Pi' |} (b_i-1)!$$ where $b_i$ is the number of blocks of $\Pi$ in the $i^{\text th}$ block of $\Pi'$.
\end{lemma}

We can use this to evaluate sums involving the M\"{o}bius function for the refinement order.
\begin{lemma}
\label{lem:ModMobiusSumx}
\be
\sum_{\Pi' \ge \Pi} |\mu(\Pi, \Pi')| \,  x^{|\Pi'|} = x^{(|\Pi|)}
\label{eq:MobSumClaim}\ee
where $x$ is arbitrary and $x^{(n)}$ is the rising factorial $x(x+1) \cdots (x+n-1)$.
\end{lemma}
\begin{proof}
Start with $|\mu(\Pi, \Pi')| = (-1)^{|\Pi| - |\Pi'|} \mu(\Pi, \Pi')$ to obtain
\bas
\sum_{\Pi' \ge \Pi} | \mu(\Pi, \Pi') | x^{|\Pi'|}
&=(-1)^{|\Pi|} \sum_{\Pi' \ge \Pi} \mu(\Pi, \Pi') (-x)^{|\Pi'|} 
\\  &= (-1)^{|\Pi|} \sum_{\Pi' \ge \Pi} \mu(\Pi, \Pi') 
\sum_{\Pi''\geq\Pi'} \zeta(\Pi',\Pi'') (-x)_{(|\Pi''|)}
\eas
using \eq{stirling-rel}.  Then use M\"obius inversion and $(-x)_{(n)} = (-1)^n x^{(n)}$ to prove the result.
\end{proof}

We will mostly be interested in the special case $x=1$:
\begin{corollary}
\label{cor:ModMobiusSum}
\be
\sum_{\Pi' \ge \Pi} | \mu(\Pi, \Pi') | = |\Pi|!
\ee
\end{corollary}

Using $|\mu(\Pi,\Pi')|\geq 1$ and the fact that $\Pi\geq
\{\{1\},\ldots,\{n\}\}$ for all $\Pi\vdash n$, we obtain a bound on
the total number of partitions:
\begin{corollary}
\label{cor:BellNumberBound}
The Bell numbers $\beta_n$ satisfy $\beta_n \le n!$.
\end{corollary}

\subsection{Fixed Points of a Quantum Expander}

We now turn to $V_{\cU(N)}$, the space fixed by the quantum expander.
As in \chapref{PRIntro}, the only operators on $(\bbC^N)^{\ot
  k}$ to commute with $U^{\ot k}$ for all $U$ are linear combinations of
subsystem permutations.  The equivalent statement for
$V_{\cU(N)}$ is that the only states invariant under all $U^{\ot k,k}$
are of the form
\be
\frac{1}{\sqrt{N^k}}
\sum_{n_1,\ldots,n_k\in[N]}\ket{n_1,\ldots,n_k,n_{\pi(1)},\ldots,n_{\pi(k)}},
\label{eq:VUN-basis}\ee
for some permutation $\pi\in\cS_k$.  Since $\cE_{\cU(N)}^k=\bbE[U^{\ot
    k,k}]$ projects onto the set of states that is invariant under all
$U^{\ot k,k}$, it follows that $V_{\cU(N)}$ is equal to the span of
    the states in \eq{VUN-basis}.

Now we relate these states to our previous notation.  
\begin{definition}
For $\pi\in\cS_k$, define the partition corresponding to $\pi$ by 
$$P(\pi) = \l\{ \{1,k+\pi(1)\}, \{2,k+\pi(2)\},\ldots,
\{k, k + \pi(k)\}\r\}.$$
\end{definition}
Then the state in \eq{VUN-basis} is simply $\ket{E_{P(\pi)}}$, and so 
\be V_{\cU(N)} = \Span\{\ket{E_{P(\pi)}} : \pi\in\cS_k\}.
\label{eq:VUN-part-basis}\ee
%(We remark in passing that since $\ket{E_{P(\pi)}}\in V_{\cS(N)}$,
%\eq{VUN-part-basis} allows to confirm that $V_{\cU(N)}\subset
%V_{\cS_N}$.)

Note that the classical expander has many
more fixed points than just the desired $\ket{E_{P(\pi)}}$.  The main task in
constructing a quantum expander from a classical one is to modify the
classical expander to decay the fixed points that should not be fixed
by the quantum expander.

\subsection{Fourier Transform in the Matrix Element Basis}
\label{sec:FT-mat-el}

Since we make use of the Fourier transform, we will need to know how it acts on a matrix element.  We find
\bes
\cF^{\ot k,k} \ket{\bfm} = \frac{1}{N^k} \sum_\bfn \omega^{\bfm.\bfn} \ket{\bfn}
\ees
where
\be
\bfm.\bfn = m_1 n_1 + \ldots + m_kn_k - m_{k+1}n_{k+1} - \ldots
- m_{2k} n_{2k}
\ee

We will also find it convenient to estimate the matrix elements
$\bra{E_{\Pi_1}}\cF^{\ot k,k}\ket{E_{\Pi_2}}$.  The properties we require are proven in the following lemmas.

\begin{lemma}
\label{lem:EqualityConstraints}
Choose any $\Pi_1,\Pi_2 \vdash 2k$.
Let $\bfm \in \Pi_1$ and $\bfn \in \Pi_2$.
Call the free indices of $\bfm$ $\tilde{m}_i$ for $1 \le i \le |\Pi_1|$.
Then let $\bfm.\bfn = \sum_{i = 1}^{|\Pi_1|} \sum_{j=1}^{2k} \tilde{m}_i
A_{i,j} n_j$ where $A_{i,j}$ is a $|\Pi_1|\times 2k$ matrix with entries
in $\{0, 1, -1\}$ which depends on $\Pi_1$ (but not $\Pi_2$).
Then
\be \bra{E_{\Pi_1}}\cF^{\ot k,k}\ket{E_{\Pi_2}} = 
N^{-k + \frac{|\Pi_1|-|\Pi_2|}{2}}
 \sum_{\bfn \in E_{\Pi_2}} 
\mathbb{I}\left(\sum_j A_{i,j} n_j \equiv 0 \bmod{N} \, \forall \, i\right)
\label{eq:EqualityConstraints}
\ee
where $\mathbb{I}$ is the indicator function.
\end{lemma}
\begin{proof}
Simply perform the $\bfm$ sum in
\be
\label{eq:EFE}
\bra{E_{\Pi_1}}\cF^{\ot k,k}\ket{E_{\Pi_2}} = 
N^{-\l(k + \frac{|\Pi_1|+|\Pi_2|}{2}\r)}
 \sum_{\bfm \in E_{\Pi_1}} \sum_{\bfn \in E_{\Pi_2}} 
\omega^{\bfm.\bfn}\qedhere
\ee
\end{proof}

\begin{lemma}
$\bra{E_{\Pi_1}}\cF^{\ot k,k}\ket{E_{\Pi_2}}$ is real and
  positive.
\end{lemma}
\begin{proof}
Since all entries in the sum in \eq{EqualityConstraints} are nonnegative and at least one
($\bfn=0$) is strictly positive, \lemref{EqualityConstraints} implies the
result.
\end{proof}

\begin{lemma}
\label{lem:E-relax}
If $\Pi_1'\leq \Pi_1$ and $\Pi_2'\leq \Pi_2$ then
\be \sqrt{|E_{\Pi_1}|\cdot|E_{\Pi_2}|}
\bra{E_{\Pi_1}}\cF^{\ot k,k}\ket{E_{\Pi_2}}
\leq \sqrt{|E_{\Pi_1'}|\cdot|E_{\Pi_2'}|}
\bra{E_{\Pi_1'}}\cF^{\ot k,k}\ket{E_{\Pi_2'}}
\label{eq:E-relax}\ee
\end{lemma}
\begin{proof}
We prove first the special case when
$\Pi_1'=\Pi_1$, but $\Pi_2'\leq \Pi_2$ is arbitrary.   Recall that
$\Pi_2'\leq \Pi_2$ implies that $E_{\Pi_2}\subseteq E_{\Pi_2'}$.  Now the
LHS of \eq{E-relax} equals
\bas N^{-k}\sum_{\bfm \in E_{\Pi_1}, \bfn \in E_{\Pi_2}} 
&\exp\left(\tpiN \bfm.\bfn\right) \\
&=  N^{|\Pi_1|-k}
 \sum_{\bfn \in E_{\Pi_2}} 
\mathbb{I}\left(\sum_j A_{i,j} n_j \equiv 0 \bmod{N} \, \forall \, i\right)
\\ & =
N^{|\Pi_1|-k}
 \sum_{\bfn \in E_{\Pi_2'}} \bbI\l(\bfn\in E_{\Pi_2}\r)
\mathbb{I}\left(\sum_j A_{i,j} n_j \equiv 0 \bmod{N} \, \forall \,
  i\right)
\\ & \leq 
N^{|\Pi_1|-k}
 \sum_{\bfn \in E_{\Pi_2'}}
\mathbb{I}\left(\sum_j A_{i,j} n_j \equiv 0 \bmod{N} \, \forall \,
  i\right)
\\ & = 
\sqrt{|E_{\Pi_1}|\,|E_{\Pi_2'}|}
\bra{E_{\Pi_1}}\cF^{\ot k,k}\ket{E_{\Pi_2'}},
\eas
as desired.
To prove \eq{E-relax} we repeat this argument, interchanging the
roles of $\Pi_1$ and $\Pi_2$ and use the fact that $\bra{E_{\Pi_1}}\cF^{\ot k,k}\ket{E_{\Pi_2}}$ is symmetric in $\Pi_1$ and $\Pi_2$.
\end{proof}
\begin{lemma}
\label{lem:E-mat-el-bound}
\item \be \bra{E_{\Pi_1}}\cF^{\ot k,k}\ket{E_{\Pi_2}} \leq 
N^{-\frac{1}{2}\l|2k - (|\Pi_1| + |\Pi_2|)\r|}
\label{eq:E-mat-el-bound}\ee
\end{lemma}
\begin{proof}
Here, there are two
cases to consider.  The simpler case is when $|\Pi_1|+|\Pi_2|\leq
2k$.  Here we simply apply the inequality
$$\sum_{\bfm \in E_{\Pi_1}, \bfn \in E_{\Pi_2}} 
\exp\left(\tpiN \bfm.\bfn\right) \leq |E_{\Pi_1}|\, |E_{\Pi_2}|
 = N^{|\Pi_1|+|\Pi_2|}$$
to \eq{EFE}, and conclude that 
$\bra{E_{\Pi_1}}\cF^{\ot k,k}\ket{E_{\Pi_2}} \leq
N^{\frac{|\Pi_1| + |\Pi_2|}{2} -k}$.

Next, we would like to prove that
\be \bra{E_{\Pi_1}}\cF^{\ot k,k}\ket{E_{\Pi_2}} \leq
N^{k-\frac{|\Pi_1| + |\Pi_2|}{2}}. 
\label{eq:desired-EFE-bound}\ee
Here we use \lemref{E-relax} with $\Pi_1'=\Pi_1$ and $\Pi_2'=\{ \{1\},
\{2\}, \ldots, \{2k\} \}$, the maximally refined partition. 
Note that $|E_{\Pi_2'}|=N^{2k}$ and $\cF^{\ot k,k}\ket{E_{\Pi_2'}} =
\ket{0}$.  Thus
\bas \bra{E_{\Pi_1}}\cF^{\ot k,k}\ket{E_{\Pi_2}}
\leq N^{k-\frac{|\Pi_2|}{2}}
\bra{E_{\Pi_1}}\cF^{\ot k,k}\ket{E_{\Pi_2'}}
=N^{k-\frac{|\Pi_2|}{2}}
\braket{E_{\Pi_1}}{0} = 
N^{k-\frac{|\Pi_1| + |\Pi_2|}{2}},
\eas
establishing \eq{desired-EFE-bound}.
\end{proof}

\begin{lemma}
\label{lem:E-not-from-perm}
If $\Pi_1=\Pi_2=P(\pi)$ then
$\bra{E_{\Pi_1}}\cF^{\ot k,k}\ket{E_{\Pi_2}}=1$.  If, for any $\Pi_1$, $\Pi_2$ with $|\Pi_1| + |\Pi_2| = 2k$, either condition
isn't met (i.e. either $\Pi_1\neq \Pi_2$ or there does not exist
$\pi\in\cS_k$ such that $P(\pi)=\Pi_1=\Pi_2$) then
\be \bra{E_{\Pi_1}}\cF^{\ot k,k}\ket{E_{\Pi_2}} \leq \frac{2k}{N}
\label{eq:E-not-from-perm}\ee
for $N > k$.
\end{lemma}
\begin{proof}
In \lemref{tildeA}, we introduce the $\Pi_1 \times \Pi_2$ matrix $\tilde{A}$ with the property that
\be
\bfm.\bfn = \sum_{i=1}^{|\Pi_1|}\sum_{j=1}^{|\Pi_2|}
\tilde{m}_i \tilde{A}_{i,j} \tilde{n}_j
\ee
for all $\bfm \in \Pi_1$ and $\bfn \in \Pi_2$ where $\tilde{m}_j$ and $\tilde{n}_j$ are the free indices of $\bfm$ and $\bfn$.  This is similar to the matrix $A$ introduced in \lemref{EqualityConstraints} except only the free indices of $\bfn$ are considered.

For $\Pi_1 = \Pi_2 = P(\pi)$, \lemref{tildeA} implies that $\tilde{A} = 0$, or equivalently $\bfm.\bfn = 0$ for all $\bfm, \bfn \in P(\pi)$.  Using $|\Pi_1|+|\Pi_2| = 2k$, $\bra{E_{\Pi_1}}\cF^{\ot k,k}\ket{E_{\Pi_2}} = 1$.

Otherwise we have  $(\Pi_1,\Pi_2)\not\in\{(P(\pi),P(\pi)):\pi\in\cS_k\}$ with $|\Pi_1| + |\Pi_2| = 2k$.  For all these, \lemref{tildeA} implies that $\tilde{A}$ is nonzero (for $N > k$, no entries in $\tilde{A}$ can be $>N$ or $<-N$ so $\tilde{A} \equiv 0 \bmod{N}$ is equivalent to $\tilde{A} = 0$).  Fix an $i$ for which the $i^{\text{th}}$ row of
$\tilde{A}$ is nonzero.  We wish to count the number of $(\tilde
n_1,\ldots,\tilde n_{|\Pi_2|})$ such that $\sum_j \tilde A_{i,j} \tilde
n_j \equiv 0 \bmod{N}$.  Assume that each $\tilde{A}_{i,j}$ divides $N$ and is
nonnegative; if
not, we can replace $\tilde{A}_{i,j}$ with $\text{GCD}(|\tilde{A}_{i,j}|,N)$ by a suitable
change of variable for $\tilde{n}_j$.  

Now choose an arbitrary $j$ such that $\tilde A_{i,j}\neq 0$.  For any
values of $\tilde n_1,\ldots,\tilde n_{j-1}, \\ \tilde n_{j+1},
\ldots,\tilde n_{|\Pi_2|}$, there are $|\tilde A_{i,j}| \leq  2k $
choices of $\tilde n_j$ such that $\sum_{j} \tilde{A}_{i,j} \tilde{n}_j \equiv 0
\bmod{N}$.  Thus, there are $\leq 2kN^{|\Pi_2|-1}$ choices of $\tilde n$
such that $\sum_{j} \tilde{A}_{i,j} \tilde{n}_j \equiv 0 \bmod{N}$.
Substituting this into \eq{EqualityConstraints} (which we can
trivially modify to apply for $\tilde{A}$ rather than just $A$), we
find that
$$\bra{E_{\Pi_1}}\cF^{\ot k,k}\ket{E_{\Pi_2}}
\leq \frac{2k}{N} N^{-k + \frac{|\Pi_1|+|\Pi_2|}{2}}
= \frac{2k}{N},$$
thus establishing \eq{E-not-from-perm}.
\end{proof}

\begin{lemma}
\label{lem:tildeA}
Let $\tilde{A}$ be the matrix such that
$\bfm.\bfn = \sum_{i=1}^{|\Pi_1|}\sum_{j=1}^{|\Pi_2|}
\tilde{m}_i \tilde{A}_{i,j} \tilde{n}_j$
for all $\bfm \in \Pi_1$ and $\bfn \in \Pi_2$ where $\tilde{m}_j$ and $\tilde{n}_j$ are the free indices of $\bfm$ and $\bfn$.
Then
$\tilde{A}=0$ if and only if $\Pi_1, \Pi_2 \geq P(\pi)$ for
some $\pi\in\cS_k$.
\end{lemma}
\begin{proof}
We first consider
$\Pi_1=\Pi_2=P(\pi)$ for the ``if'' direction.  Note that for any
$\bfm,\bfn\in E_{P(\pi)}$, we have
\be\bfm. \bfn = \sum_{j=1}^k m_jn_j - \sum_{j=1}^k m_{\pi(j)} n_{\pi(j)}
 = 0.
\label{eq:P-P-cancel}\ee
This implies that $\tilde{A}=0$.  Now, choose any $\Pi_1 \ge P(\pi)$ and $\Pi_2 \ge P(\pi)$.  Then for any $\bfm \in \Pi_1$ and $\bfn \in \Pi_2$, $\bfm, \bfn \in P(\pi)$.  This means \eq{P-P-cancel} holds for this case so $\tilde{A}=0$ also.

On the other hand, suppose that $\tilde{A}=0$.  We will argue that
this implies the existence of a permutation $\pi$ such that $\Pi_1, \Pi_2 \ge P(\pi)$, thus establishing the ``only
if'' direction.

Let $\Pi_{1,j}$ (resp. $\Pi_{2,j}$) denote the $j^{\text{th}}$ block
of $\Pi_1$ (resp. $\Pi_2$).  Then
$$\tilde A_{i,j} = \sum_{\substack{i'\in\Pi_{1,i}\\ j'\in \Pi_{2,j}}}
\Lambda_{i',j'},$$
where $\Lambda_{i',j'}$ is defined to be 
$$\Lambda_{i',j'} = 
\begin{cases} 1 & \text{if $i'=j'\in\{1,\ldots,k\}$}\\
-1 & \text{if $i'=j'\in\{k+1,\ldots,2k\}$}\\
0 & \text{if $i'\neq j'$}\end{cases}.
$$
If $\tilde{A}=0$ then for each $i,j$ we have
\be \l| \Pi_{1,i} \cap \Pi_{2,j} \cap \{1,\ldots,k\}\r| = 
\l| \Pi_{1,i} \cap \Pi_{2,j} \cap \{k+1,\ldots,2k\}\r|
\label{eq:meet-balanced}.\ee
Denote the {\em meet} of $\Pi_1$ and $\Pi_2$, $\Pi_1\land\Pi_2$ to be
the greatest lower bound of $\Pi_1$ and $\Pi_2$, or equivalently the
unique partition with the fewest blocks that satisfies $\Pi_1\land\Pi_2 \leq
\Pi_1$ and $\Pi_1\land\Pi_2 \leq \Pi_2$.  The blocks of 
$\Pi_1\land\Pi_2$ are simply all of the nonempty sets
$\Pi_{1,i}\cap\Pi_{2,j}$, for $i=1,\ldots,|\Pi_1|$ and
$j=1,\ldots,|\Pi_2|$.  Thus, \eq{meet-balanced} implies that each
block of $\Pi_1\land\Pi_2$ contains an equal number of indices from
$\{1,\ldots,k\}$ as it does from $\{k+1,\ldots,2k\}$.
This implies
the existence of a permutation $\pi\in\cS_k$ such that
$\{i,k+\pi(i)\}$ is contained in a single block of $\Pi_1\land\Pi_2$
for each $i=1,\ldots,k$.  Equivalently $\Pi_1\land\Pi_2\geq P(\pi)$,
implying that $\Pi_1\geq P(\pi)$ and $\Pi_2\geq P(\pi)$.
\end{proof}

\subsection{Proof of Lemma 4.2.2} % \lemref{Gap}}
\label{sec:ProofOfLemGap}

\begin{proof}

We would like to show that, for any unit vector $\ket{\psi}\in V_0$, 
$|\bra{\psi}\cF^{\ot k,k}\ket{\psi}|^2 \leq 2(2k)^{4k}/\sqrt{N}$.   Our strategy
will be to calculate the matrix elements of $\cF^{\ot k,k}$ in the
$\ket{I_\Pi}$ and $\ket{E_\pi}$ bases.  While the $\ket{I_\Pi}$ states
are orthonormal, we will see that the $\bra{E_{\Pi_1}}\cF^{\ot
  k,k}\ket{E_{\Pi_2}}$ matrix elements are easier to calculate.  We
then use  M\"obius functions to express $\ket{I_\Pi}$ in terms of
$\ket{E_\Pi}$. 

Consider the matrix $\cE_{\cS_N}^{2k} \cF^{\ot k,k}\cE_{\cS_N}^{2k}$.  It
has $k!$ unit eigenvalues, corresponding to the $k!$-dimensional space
$V_{\cU(N)}$.   Call the $k!+1^{\text{st}}$ largest eigenvalue $\lambda_A$.
We bound $\lambda_A$ with 
\ba k! + \lambda_A^2  
&\leq \tr \l(\cE_{\cS_N}^{2k} \cF^{\ot k,k}\cE_{\cS_N}^{2k}\r)^2 \nonumber
\\ & = \sum_{\Pi_1,\Pi_2\vdash 2k} 
\l|\bra{I_{\Pi_1}}\cF^{\ot k,k}\ket{I_{\Pi_2}}\r|^2
\label{eq:I-sum}.\ea

We divide the terms in \eq{I-sum} into four types.
\begin{subequations}\label{eq:I-sum-pieces}
\begin{enumerate}\item The leading-order contribution comes from the $k!$ terms
of the form $\Pi_1=\Pi_2=P(\pi)$ for $\pi\in\cS_k$.  We bound them
with the trivial upper bound
\be |\bra{I_{\Pi_1}}\cF^{\ot k,k}\ket{I_{\Pi_2}}|^2\leq 1
\ee (which turns
out to be nearly tight).
We will then show that the remaining terms are all $k^{O(k)}/N$.
\item If $|\Pi_1|+|\Pi_2| < 2k$ then
\ba \l|\bra{I_{\Pi_1}}\cF^{\ot k,k}\ket{I_{\Pi_2}}\r|^2
& = \frac{1}{|I_{\Pi_1}|\cdot |I_{\Pi_2}| N^{2k}} 
\l| \sum_{\substack{\bfm\in\Pi_1\\ \bfn\in\Pi_2}}
e^{\frac{2\pi i \bfm.\bfn}{N}} \r|^2 
\nn & \leq \frac{|I_{\Pi_1}|\cdot |I_{\Pi_2}|}{N^{2k}} 
\nn & \leq N^{|\Pi_1| + |\Pi_2| - 2k} \leq \frac{1}{N},
\ea
where in the last line we have used the fact that $|I_\Pi| \leq
|E_\Pi| = N^{|\Pi|}$.
\item If $|\Pi_1|+|\Pi_2| > 2k$ then we will show that
\be \l|\bra{I_{\Pi_1}}\cF^{\ot k,k}\ket{I_{\Pi_2}}\r|^2 
\leq \frac{4\cdot(2k!)^2}{N}
\label{eq:too-big}\ee
\item If $|\Pi_1|+|\Pi_2|=2k$ but either $\Pi_1\neq \Pi_2$ or there is
  no $\pi\in\cS_k$ satisfying $P(\pi)=\Pi_1=\Pi_2$, then we will show that
\be \l|\bra{I_{\Pi_1}}\cF^{\ot k,k}\ket{I_{\Pi_2}}\r|^2 
\leq \frac{((2k)!+2k)^2}{N^2} \leq \frac{4\cdot(2k!)^2}{N}
\label{eq:not-from-perm}\ee
\end{enumerate}
\end{subequations}
To establish these last two claims, we will find it useful to express
$\ket{I_\Pi}$ in terms of the various $\ket{E_\Pi}$ states.  

Lemmas \ref{lem:E-mat-el-bound} and \ref{lem:E-not-from-perm} can now be used together with the M\"obius
function to bound $|\bra{I_{\Pi_1}}\cF^{\ot k,k}\ket{I_{\Pi_2}}|^2$.
First, suppose $|\Pi_1|+|\Pi_2|>2k$.  Then
\ba
 \l|\bra{I_{\Pi_1}}\cF^{\ot k,k}\ket{I_{\Pi_2}}\r| & =
\l|\sum_{\substack{\Pi_1'\geq \Pi_1\\\Pi_2'\geq \Pi_2}}
\sqrt{\frac{|E_{\Pi_1'}|\, |E_{\Pi_2'}|}
{|I_{\Pi_1}|\, |I_{\Pi_2}|}}
 \mu(\Pi_1,\Pi_1')\mu(\Pi_2,\Pi_2') 
\bra{E_{\Pi_1}}\cF^{\ot k,k}\ket{E_{\Pi_2}}\r| \nonumber
\\ & \leq
\sum_{\substack{\Pi_1'\geq \Pi_1\\\Pi_2'\geq \Pi_2}}
\sqrt{\frac{|E_{\Pi_1'}|\, |E_{\Pi_2'}|}
{|I_{\Pi_1}|\, |I_{\Pi_2}|}}
\l| \mu(\Pi_1,\Pi_1')\mu(\Pi_2,\Pi_2') \r|
\bra{E_{\Pi_1'}}\cF^{\ot k,k}\ket{E_{\Pi_2'}} \nonumber
\\ & \leq
\frac{N^k}{\sqrt{|I_{\Pi_1}|\, |I_{\Pi_2}|}}
\sum_{\substack{\Pi_1'\geq \Pi_1\\\Pi_2'\geq \Pi_2}}
\l|\mu(\Pi_1',\Pi_1)\mu(\Pi_2',\Pi_2)\r| \nonumber
\ea
by \lemref{E-mat-el-bound}.  Then using by \corref{ModMobiusSum} we find
\ba
 \l|\bra{I_{\Pi_1}}\cF^{\ot k,k}\ket{I_{\Pi_2}}\r| & = \frac{N^k |\Pi_1|!\, |\Pi_2|!}{\sqrt{(N)_{|\Pi_1|}(N)_{|\Pi_2|}}}
\label{eq:penultimate-g2k}
\\ & \leq \frac{2\cdot(2k)!}{\sqrt{N}} \nonumber
\ea
In the last step, we have assumed that $4k^2<N$, so that
$(N)_{\ell}\geq N^\ell/2$ for any $\ell\leq 2k$.  We have also made
use of the fact that (still assuming $4k^2<N$) \eq{penultimate-g2k} is maximised
when $|\Pi_1| +|\Pi_2|=2k+1$, and in particular, when one of
$|\Pi_1|$, $|\Pi_2|$ is equal to $2k$ and the other is equal to 1.

A similar analysis applies to the pairs $\Pi_1,\Pi_2$ with
$|\Pi_1|+|\Pi_2|=2k$, but with
$(\Pi_1,\Pi_2)\not\in\{(P(\pi),P(\pi)):\pi\in\cS_k\}$.  In this case, 
\ba
\bra{I_{\Pi_1}}\cF^{\ot k,k}\ket{I_{\Pi_2}} 
&= \sqrt{\frac{|E_{\Pi_1}|\, |E_{\Pi_2}|}
{|I_{\Pi_1}|\, |I_{\Pi_2}|}}
\bra{E_{\Pi_1}}\cF^{\ot k,k}\ket{E_{\Pi_2}}  \,+ \nonumber \\
& 
\sum_{\Pi_1' \ge \Pi_1, \Pi'_2 \ge \Pi_2 \atop
(\Pi_1',\Pi_2')\neq (\Pi_1,\Pi_2)} 
\sqrt{\frac{|E_{\Pi_1'}|\, |E_{\Pi_2'}|}
{|I_{\Pi_1}|\, |I_{\Pi_2}|}}
\mu(\Pi_1, \Pi_1') \mu(\Pi_2, \Pi_2')
\bra{E_{\Pi_1'}}\cF^{\ot k,k}\ket{E_{\Pi_2'}}
\label{eq:strictly-greater}\ea
We now use Lemmas \ref{lem:E-not-from-perm} and \ref{lem:E-mat-el-bound} to bound each of the two terms.  For the
first term, we use \eq{E-not-from-perm} to upper bound it with
$2k/N$.  For each choice of $\Pi_1'$ and $\Pi_2'$ in the second sum,
we have $|\Pi_1'|+|\Pi_2'|\leq 2k-1$.  Thus we can upper bound
the absolute value of the second term in \eq{strictly-greater} with
\bas \frac{1}{\sqrt{|I_{\Pi_1}|\, |I_{\Pi_2}|}}
\sum_{\Pi_1' \ge \Pi_1, \Pi'_2 \ge \Pi_2 \atop
(\Pi_1',\Pi_2')\neq (\Pi_1,\Pi_2)} 
|\mu(\Pi_1, \Pi_1') \mu(\Pi_2, \Pi_2')|
N^{|\Pi_1'|+|\Pi_2'|-k}
 &\leq \frac{2\cdot |\Pi_1|!\cdot |\Pi_2|!}{N} \\ &\leq \frac{(2k)!}{N}.
\eas
We combine the two terms and square to establish \eq{not-from-perm}.

We now put together the components from \eq{I-sum-pieces} to
upper bound \eq{I-sum}, and find that
$$k! + \lambda_A^2  \leq k! + 
\beta_{2k}^2 \frac{4\cdot(2k!)^2}{N},$$
implying that $\lambda_A \leq 2\beta_{2k}(2k!)/\sqrt{N} \leq 2(2k)^{4k}/\sqrt{N}$.
  This concludes the proof of \lemref{Gap}.
\end{proof}

\section{Conclusions}

We have shown how efficient quantum tensor product expanders can be
constructed from efficient classical tensor product expanders.  This
immediately yields an efficient construction of unitary
$k$-designs for any $k$.  Unfortunately our results do not work for
all dimensions; we require the dimension $N$ to be $\Omega((2k)^{8k})$.
While tighter analysis of our construction could likely improve this,
our construction does not work for $N < 2k$.  Constructions of
expanders for all dimensions remains an open problem.

\chapter{Applications of Designs}
\label{chap:DesignApplications}

In this chapter we first survey known applications of designs from a wide variety of areas.  Then we present new results applying designs to derandomise some large deviation bounds.

\section{Review of Applications}

As we have already discussed, random unitaries and random states have many applications.  For some of these applications a design is sufficient since only the first few moments of the distribution are required to be equal to those of the Haar measure.

\subsection{Quantum Cryptography}
\label{sec:ApplicationsQCrypto}

The first applications we discuss are to quantum cryptography.  In classical cryptography, the one-time pad is the most basic operation that perfectly encrypts a message using a key that is the same length as the message.  In quantum cryptography the analogue is a quantum operation $\cE$ such that for all input states $\rho$, $\cE(\rho) = \rho_0$ with the requirement that given a secret key Bob can decode Alice's message perfectly.  Since all states are encoded to $\rho_0$ Eve cannot learn anything about the message without knowing the key.  In this section all logs will be taken to base $2$.

If $\rho_0$ is the identity, then the map $\cE$ is a unitary 1-design.  Therefore using $2n$ bits of key to label the Pauli operators provides a quantum one-time pad.  In fact, in \cite{AMTW00} it is shown that $2n$ bits of key are also necessary for a quantum one-time pad.  Therefore this unitary 1-design is of optimal size.

It is interesting to note that $2n$ bits of key are required rather than just $n$ for the classical one-time pad.  This is related to the fact that quantum states allow superdense coding \cite{SuperdenseCoding}, which allows two classical bits to be sent per qubit.  Although it is not possible to use a shorter key for exact encryption, it would be desirable to shorten the key if we can tolerate Eve learning a small amount of information about the message.  In \cite{AmbainisSmith04}, they consider closeness in the 1-norm, and set $\rho_0 = I$.  They define a map $\cE$ as an $\eps$-approximate quantum encryption scheme if for all $\rho$
\be
\label{eq:ApproxQ1TimePad1Norm}
|| \cE(\rho) - I/d ||_1 \le \eps.
\ee
Thus we see an $\eps$-approximate unitary 1-design according to (for example) \defref{ApproxUnitaryDesignDiamond} (DIAMOND) suffices.  In \cite{AmbainisSmith04}, they present an efficient construction that satisfies \eq{ApproxQ1TimePad1Norm} with $n + 2\log n + 2\log\l(\frac{1}{\eps}\r) + O(1)$ bits of key.  While this does not immediately provide an $\eps$-approximate 1-design according to any of our definitions with only $(1+o(1))n$ bits of key, \eq{ApproxQ1TimePad1Norm} is a valid definition of an $\eps$-approximate unitary 1-design.  This key length was further improved by Dickinson and Nayak \cite{DickinsonNayak06} to $n + 2 \log \l( \frac{1}{\eps} \r) + O(1)$ and their construction is efficient.

A stronger definition for $\eps$-approximate encryption was given in \cite{RandomizingQuantumStates04}.  They define a map $\cE$ to be an $\eps$-approximate quantum encryption scheme if for all $\rho$
\be
\label{eq:ApproxQ1TimePadInftyNorm}
|| \cE(\rho) - I/d ||_\infty \le \eps/d.
\ee
This implies the 1-norm bound in \eq{ApproxQ1TimePad1Norm} but a dimension factor is lost when converting the other way.  This could be used as yet another approximate 1-design definition.  However, there are no known efficient constructions of such $\infty$-norm randomising maps.  In \cite{RandomizingQuantumStates04} they provide an inefficient randomised construction with key length $n + \log n + 2 \log \l( \frac{1}{\eps} \r) + O(1)$.  Their method is to show that with non-zero probability random unitaries suffice.

This result was improved by Aubrun in \cite{AubrunRandomizingChannels} to reduce the key length to $n + 2 \log \l( \frac{1}{\eps} \r) + O(1)$.  The method is the same as \cite{RandomizingQuantumStates04} except the analysis is tighter.  Aubrun also makes a step towards finding an efficient construction by showing that the unitaries can be Pauli matrices, which can be implemented efficiently, although the sampling is still inefficient.

Besides cryptographic applications, it is shown in \cite{RandomizingQuantumStates04} that $\infty$-norm randomising maps can be used to hide correlations from local operations and classical communication (LOCC) and have applications to data hiding (see later) and locking of classical correlations \cite{Locking04}, whereby classical correlations can be hidden but unlocked by a very short key.

The last cryptography example we give is that of non-malleable encryption given in \cite{TamperResistance}.  Here the authors not only consider hiding information from Eve but they also require that she cannot change the message.  Of course, Eve could always replace the message with some fixed state or do nothing, so according to \cite{TamperResistance}, an encryption scheme is non-malleable if these (or a convex combination) are the only operations Eve can perform on the encoded data.  The main result of this paper is that a unitary 2-design is necessary and sufficient.  They then show, as do Gross et al.~\cite{GAE07}, that a 2-design requires at least $(d^2-1)^2+1$ unitaries i.e.~the key must be at least $4n - o(1)$ bits long.  Even for approximate encryption (which can be seen as an approximate 2-design) the key length is essentially the same.

\subsection{Measurement}

In some cases a random measurement is a good choice but cannot be performed efficiently.  One example of such a result is:
\begin{theorem}[Sen, \cite{Sen05}]
Let $\rho_1$ and $\rho_2$ be any mixed states with $r(\rho_1) + r(\rho_2) \le \sqrt{d}/C$ for a sufficiently large constant $C$.  Here, $r(\rho)$ is the rank of the state $\rho$.  Then
\be
\label{eq:RandomMeasurementBound}
\bbE_M || M(\rho_1) - M(\rho_2) ||_1 = \Omega(|| \rho_1 - \rho_2 ||_2)
\ee
where $M$ is an orthonormal basis picked from the Haar measure.  Here, $M(\rho)$ is the probability distribution of outcomes according to the POVM $M$.
\end{theorem}
Since a large 1-norm distance between probability distributions means the distributions are easily distinguishable, this result places a lower bound on the distinguishability of the states $\rho_1$ and $\rho_2$ in terms of their 2-norm distance.

In \cite{AmbainisEmerson07}, Ambainis and Emerson show that a POVM made from a state 4-design achieves the bound in \eq{RandomMeasurementBound}.  In fact, an $\eps$-approximate state 4-design suffices, provided that $\eps = O(|| \rho_1 - \rho_2 ||_2^4)$.  To ensure the POVM is suitably normalised, we insist here that the approximate 4-design is also an exact 1-design rather than an $\eps$-approximate 1-design, which is all that \defref{ApproxStateDesign} ensures.

In \cite{IblisdirRoland06}, Iblisdir and Roland consider a slightly different measurement problem for which a random measurement achieves the best outcome.  The setting is that Alice chooses a random pure state (the authors only consider the case that Alice's system is 2-dimensional i.e.~a single qubit) from the Haar measure and makes $k$ copies of it.  Bob then has to find a state with high overlap with the given state.  The POVM that achieves the optimum is \cite{MassarPopsecu95}
\be
(k+1)\ket{\psi} \bra{\psi}^{\ot k} d\psi.
\ee
From \lemref{SymmetricStateAverage}, the average of this is the projector onto the symmetric subspace of $k$ qubits.  While this is not the identity, no other outcomes are possible because the input state is symmetric.  The states in the POVM can be replaced by a state $k$-design and in \cite{IblisdirRoland06} the authors present a construction of a state $k$-design for all $k$, although only for one qubit.

\subsection{Average Gate Fidelity}

When implementing a quantum operation, we would like to know how far the actual operation is from the desired.  One way of measuring this is the average gate fidelity \cite{Nielsen02}:
\be
\bar{F}(\cE, U) = \int d\psi \bra{\psi} U^\dagger \cE(\ket{\psi} \bra{\psi}) U \ket{\psi}
\ee
where $\cE$ is the operation implemented and $U$ is the desired unitary.  We see immediately, following \cite{DCEL06}, that the integrand is a balanced polynomial of degree 2 so the Haar measure on states can be replaced by a state 2-design.  We can even use an approximate design if the average fidelity only needs to be known approximately.  By repeatedly sampling from the design we can obtain an estimate of the average to $1/\poly (\log d)$ accuracy efficiently whereas naively sampling random states will not be efficient.

\subsection{Data Hiding}

Data hiding was introduced by Terhal, DiVincenzo and Leung \cite{TDL01,DLT02} as a fundamentally quantum concept.  The setting is that Alice and Bob share a quantum state which contains secret bits.  However, the state is chosen so that if they can only communicate using LOCC then they cannot learn this secret bit.  To encode one secret bit, the ``hider'' constructs one of two orthogonal mixed states $\rho_0^{(m)}$ and $\rho_1^{(m)}$ and hands half to Alice and the other half to Bob.  $\rho_0^{(m)}$ is the state with $m$ random Bell pairs chosen subject to the constraint that the number of singlets is even.  $\rho_1^{(m)}$ is the same state except with an odd number of singlets.  The parameter $m$ controls the degree of security.

The way that designs help here is in the construction of these states using minimal resources.  The authors show that $\rho_0$ can be obtained from twirling any initial pure state of the form $\ket{\psi}\bra{\psi} \ot \ket{\psi}\bra{\psi}$\footnote{By unitary invariance of the Haar measure, the choice of $\ket{\psi}$ does not affect the resultant state.}:
\be
\rho_0^{(m)} = \int_{\cU(2^m)} dU (U \ot U) \ket{\psi}\bra{\psi} \ot \ket{\psi}\bra{\psi} (U \ot U)^\dagger.
\ee
$\rho_1^{(m)}$ can be created from $\rho_0^{(m-1)}$.

The authors consider replacing the Haar integral with a sum over a unitary 2-design.  If errors can be tolerated then an approximate 2-design can be used and the state can be prepared efficiently.

\subsection{Decoupling and Evolution of Black Holes}

For various tasks in quantum Shannon theory, it is desirable to decouple a system from the environment.  In \cite{HHYW07} and \cite{ADHW06}, it is shown that for most random unitaries applied to the system the resulting overall state is close to a product state.

The setting is that there is a system $S$ with two parts $S_1$ and $S_2$.  The environment is $E$.  Let the initial state be $\psi_{SE}$ and let
\be
\sigma_{S_2 E}(U) = \tr_{S_1} \l[ ( U \ot I_E ) \psi_{S E} ( U^\dagger \ot I_E ) \r].
\ee
Then we have
\begin{theorem}[\cite{ADHW06}, Theorem 4.2]
\label{thm:Decoupling}
\begin{multline}
\int_{\cU(S)} \l| \l| \sigma_{S_2 E}(U) - \sigma_{S_2}(U) \ot \sigma_R(U) \r| \r|_1^2 dU \le \\ \frac{d_S d_E}{d_{S_1}^2} \l( \tr \l[ \l( \psi_{SE} \r)^2 \r] + \tr \l[ \l( \psi_S \r)^2 \r] \tr \l[ \l( \psi_E \r)^2 \r] \r)
\end{multline}
where $\sigma_{S_2}(U) = \tr_{S_1 E} \sigma_{S E}(U)$, etc..
\end{theorem}

The proof uses the 2-norm squared, which is a polynomial of degree 2 in the matrix elements of the random unitary.  Therefore the same result holds when $U$ is selected from a unitary 2-design instead and, as above, an approximate design can be used to allow an efficient implementation.  This allows the encoding circuits in \cite{ADHW06} to be made efficient although unfortunately the decoding circuits are still inefficient.

Decoupling has also been used in the study of the evolution of black holes.  While many aspects of quantum gravity are not understood, some attempts have been made to understand how black holes leak information.  Two examples are by Hayden and Preskill \cite{HaydenPreskill07} and Sekino and Susskind \cite{SekinoSusskind08}.  We concentrate on the approach in \cite{HaydenPreskill07} here.  The idea is that Alice wishes to destroy some quantum information by throwing it into a black hole.  However, Bob has been watching it and storing the Hawking radiation emitted.  The question they ask is how long does Bob have to wait before he can recover Alice's information.

Imagine that Alice's information is maximally entangled with a system $N$ held by Charlie.  Should Bob acquire a state from the emitted radiation that is maximally entangled with $N$ then we say he has successfully recovered Alice's information.  Decoupling is used because, if what remains of the black hole after some evaporation is uncorrelated with $N$, then the emitted radiation must be maximally entangled with $N$ and Bob has succeeded.  We therefore require that the evolution of the black hole produces a decoupling unitary.  If the evolution is random then, using \thmref{Decoupling}, this will likely happen, provided enough radiation has been emitted.  In fact, if Bob holds a system that is maximally entangled with the black hole's internal state before Alice throws in her message, then he can recover her state with fidelity $1-2^{-c}$ by reading in only the $k+c$ qubits emitted after Alice deposits her information, where $k$ is the number of qubits in Alice's message.

This model is not physically realistic because most unitaries cannot be implemented efficiently so the black hole would take far too long to apply the decoupling unitary.  However, as we said above, only a 2-design is required.  In fact, in \cite{HaydenPreskill07} they consider the case that the evolution of a black hole is a local random quantum circuit.  This is similar to the random circuits discussed in \chapref{RandomCircuits} except they assume that the unitaries are only applied to nearest-neighbour qubits.  Should the random circuit converge to a 2-design quick enough (as Hayden and Preskill conjecture) then the evolution will be sufficiently fast for Bob to find Alice's state.  While our results do not prove this they could readily be extended to cover the local case considered here.

\subsection{Applications for Larger \texorpdfstring{$k$}{k}}

So far we have only used $k$-designs for $k \le 4$.  However, the higher $k$ is the more similar a $k$-design is to a random unitary.  In the next section we consider replacing random unitaries with $k$-designs in large deviation bounds, thus finding applications for larger $k$.

\section{Derandomising Large Deviation Bounds}
\label{sec:LargeDeviations}

The remainder of this chapter has been published previously as \cite{LargeDeviationskDesigns}.

There are many results in quantum information theory that show generic properties of states or unitaries (e.g.~\cite{AspectsOfGenericEntanglement, RandomizingQuantumStates04}).  Often, these results say that, with high probability, a random state or unitary has some property, for example high entropy.  However, as we have seen above, neither random unitaries nor random states can be implemented efficiently.  This limits the usefulness of such results since no physical systems will behave truly randomly.  To make such results more physically relevant, it would be desirable to show that these properties are generic properties of unitaries from some natural distribution that can be implemented efficiently.  Only then could we conclude that we would expect to see such properties in natural systems.

In many cases, the generic properties of unitaries are desirable but randomised constructions given by the large deviation bounds are inefficient.  We would like to come up with distributions which can be implemented efficiently that have similar generic properties.  One example where the best known construction is an inefficient randomised one is the $\infty$-norm randomising map (see \secref{ApplicationsQCrypto}).  Another example is locking of classical correlations \cite{Locking04, RandomizingQuantumStates04}, which is a quantum phenomenon whereby a small amount of communication can greatly enhance the classical correlation between two parties.  To prove the randomised constructions, the authors show that, with some non-zero probability, random unitaries have the required property.  However, there are no known efficient constructions of unitaries with these properties.  If, on the other hand, we could show that unitaries drawn randomly from a set that can be implemented efficiently have the property with non-zero probability, we could move an important step closer to finding efficient constructions.  (It would not actually provide an efficient construction unless we could find an efficient sampling method.)  In fact, for the case of $\infty$-norm randomisation, this was done by Aubrun in \cite{AubrunRandomizingChannels}.

In this section we continue the theme of replacing the Haar measure with a $k$-design.  The reason for using $k$-designs is two-fold.  Firstly, because the first $k$ moments are the same we would expect similar (although weaker) measure concentration results.  Secondly, for $k = \poly(n)$ (when the design is on $n$ qubits), we might expect to be able to implement the $k$-design efficiently (i.e.~in $\poly(n)$ time).  Indeed, for $k = O(n/\log n)$, we can use the construction from \chapref{TPE}, provided we allow for approximate designs.  However, in the applications we consider here we can always make the approximation good enough to make the error negligible.

Not only can $k$-designs be constructed efficiently, they may even be the product of generic dynamics.  In \chapref{RandomCircuits}, we show that random quantum circuits quickly converge to a 2-design for a quite general model of such circuits.  We also conjecture in \chapref{RandomCircuits} that random circuits give $k$-designs for $k>2$ and $k=\poly(n)$ in polynomial time.  If a physical system can be accurately modelled by a random circuit then, assuming this conjecture, the naturally occurring states will be $k$-designs rather than fully random states.

We now summarise some related results in this area.  Smith and Leung \cite{SmithLeung06} and Dahlsten and Plenio \cite{DahlstenPlenio05} found large deviation bounds for stabiliser states.  They showed that, in certain regimes, stabiliser states are very likely to have large entanglement.  Stabiliser states are state 2-designs so our results can be seen as a generalisation of this to $k$-designs for $k>2$ and to other problems.  There are also some recent classical results related to the present work.  Alon and Nussboim \cite{AlonNussboim08} consider replacing full randomness with $k$-wise independence, a classical analogue of $k$-designs, in random graph theory.  They show that $k$-wise independent random graphs with $k = \poly(\log N)$ ($N$ is the number of vertices) have similar generic properties to fully random graphs.

In the remainder of this chapter, unless otherwise stated, we will use the definition of an $\eps$-approximate unitary design given in terms of monomials, as in \defref{ApproxUnitaryDesignMonomials}.  Using the tensor product expander construction of \chapref{TPE} together with \lemref{ApproxUnitaryDesignEquiv} gives an efficient construction for $k = O(\log d/\log \log d)$ for this definition.

\subsection{Introductory Problem: Entanglement of a 2-design}
\label{sec:EntanglementOfA2design}

We now illustrate our main idea by showing a large deviation bound for the entanglement of a 2-design, but in a different way to \cite{SmithLeung06, DahlstenPlenio05}.

It has been known for a long time that random states are highly entangled across any bipartition \cite{PagesConjecture, PagesConjectureProof94, PagesConjectureProof95}.  Further, in \cite{AspectsOfGenericEntanglement}, it is shown that random unitaries generate almost maximally entangled states with high probability.  However, generating random states is inefficient so it is an interesting question to ask if random efficiently obtainable states are highly entangled.

Let the system be $\cH = \cH_S \ot \cH_E$, where we label the two systems $S$ and $E$.  Let the dimensions be $d_S$ and $d_E$ and $d = d_S d_E$.  Let the overall initial state be any fixed pure state $\rho_0$.  Then consider applying a random unitary $U$ to $SE$ to get the state $\psi = U \rho_0 U^\dagger$.  Then the von Neumann entropy $S(\psi_S) = - \tr \psi_S \log \psi_S$ of the reduced state $\psi_S = \tr_E \psi$ is close to $\log_2 d_S$ (the maximal) with high probability:
\begin{theorem}[\cite{AspectsOfGenericEntanglement} Theorem 3.3]
\label{thm:EntropyTailBoundFullRandomness}
Let $d_E \ge d_S \ge 3$.  Then for unitaries chosen from the Haar measure
\be
\Pr(S(\psi_S) \le \log_2 d_S -\alpha -\beta) \le \exp \l(- \frac{ (d-1)C \alpha^2}{(\log_2 d_S)^2} \r)
\ee
where $C = \frac{1}{8 \pi^2}$ and $\beta = \frac{1}{\ln 2} \frac{d_S}{d_E}$.
\end{theorem}
Now, consider choosing the unitary from a 2-design instead.  Later on (\lemref{ExpectedPurity}), we show that $\bbE \tr \psi_S^2 = \frac{d_S + d_E}{d+1} =: \mu$.  Since purity is a polynomial of degree $2$, it does not matter if we take the expectation over the Haar measure or the 2-design.  We now apply Markov's inequality:
\bas
\Pr\l(\tr \psi_S^2 \ge \mu\gamma\r) &\le \frac{\bbE \tr \psi_S^2}{\mu\gamma} \\
&= \frac{1}{\gamma}.
\eas
Using the bound $S(\psi_S) \ge -\log_2 \tr \psi_S^2$ and some manipulations (the details are in \secref{Entropy}), this can be written as
\be
\Pr(S(\psi_S) \le \log_2 d_S -\alpha -\beta) \le 2^{-\alpha}
\ee
where $\beta$ is as in \thmref{EntropyTailBoundFullRandomness}.  This bound is much weaker than the bound in \thmref{EntropyTailBoundFullRandomness} and, in particular, does not show stronger concentration as $d$ increases.  Later in the chapter, we will show that choosing unitaries from a $k$-design with larger $k$ will give a much stronger bound that does give sharp concentration results for large $d$.

\subsection{Main Results}

We will now state our main results.

Our most general result is:
\begin{theorem}
\label{thm:ConcentrationPolynomial}
Let $f$ be a polynomial of degree $K$.  Let $f(U) = \sum_i \alpha_i M_i(U)$ where $M_i(U)$ are monomials and let $\alpha(f) = \sum_i | \alpha_i |$.  Suppose that $f$ has probability concentration
\begin{equation}
\label{eq:concentration}
\Pr_{U \sim \cU(d)}(|f - \mu| \ge \delta) \le C e^{-a \delta^2}
\end{equation}
and let $\nu$ be an $\eps$-approximate unitary $k$-design.
Then
\begin{equation}
\Pr_{U \sim \nu}(|f - \mu| \ge \delta) \le \frac{1}{\delta^{2m}} \left( C \left(\frac{m}{a}\right)^m + \frac{\eps}{d^k} \left(\alpha + | \mu | \right)^{2m} \right)
\end{equation}
for integer $m$ with $2mK \le k$.
\end{theorem}
We therefore take a bound for Haar random unitaries of the form \eq{concentration} and turn it into a bound for $k$-designs.  Often, we will use Levy's Lemma (\lemref{Levy}) to give the initial concentration bound in \eq{concentration}.  In this case, $a = \Theta(d)$ (provided the Lipschitz constant (see later) is constant).

We then apply this to entropy, as a generalisation of \secref{EntanglementOfA2design}.  We go via the 2-norm since the entropy function is not a polynomial.  We find
\begin{theorem}
\label{thm:EntropyBound}
Let $\nu$ be a $4^{-n^2}$-approximate unitary $\frac{n}{10 \log_2 n}$-design on dimension $2^n$ with $n \ge 19$.  Let $d_S d_E = 2^n$ and $2 \le d_S \le 2^{n/10}$ and $\alpha \ge 2$.  Then
\be
\Pr_{U \sim \nu}(S(\psi_S) \le \log_2 d_S - \alpha - \beta) \le 8 \exp_2 \l(-\frac{n}{80\log_2 n} \l( \frac{n}{5} + \alpha \r) \r)
\ee
where $\beta = \frac{1}{\ln 2} \frac{d_S}{d_E}$ and $\exp_2$ is the exponential function base 2.
\end{theorem}
We choose a $k$-design for $k=\frac{n}{10 \log_2 n}$ since this is (up to constants) the largest $k$ for which we have an efficient unitary $k$-design construction (using the construction of \chapref{TPE}).

We then move on to apply our results to ideas in statistical mechanics from Popescu et al.~\cite{ThermalisationPSW}.  In this paper, the authors show that, for almost all pure states of the universe, any subsystem is very close to the canonical state, which is the state obtained by assuming a uniform distribution over all allowed states of the universe (defined in \eqn{CanonicalState}).  This could be achieved if the dynamics of the universe produced a random unitary,  but this would take exponential time in the size of the universe.  We show that the random unitary can be replaced by a $k$-design, showing that the canonical state can be reached in polynomial time:
\begin{theorem}
Let $\Omega_S$ be the canonical state of the system (defined in \eq{CanonicalState}) and $\rho_S$ be the state after choosing a unitary from an $\eps$-approximate $k$-design.  Let $d_R$ be the dimension of the universe's Hilbert space subject to the arbitrary constraint $R$ (normally this will be a total energy constraint).  Then for $\eps \le \frac{3}{2}\l( \frac{4 d_S^3}{d_R} \r)^{k/8}$, $k \le \frac{4 d_S^2}{9 \pi^3}$
\be
\Pr_{U \sim \nu}( || \rho_S - \Omega_S ||_1 \ge \delta) \le 6 \l(\frac{4 d_S^3}{d_R \delta^2} \r)^{k/8}.
\ee
\end{theorem}

Finally, we use results from \cite{MostStatesUselessMBQCGFE} to show that most states in an $O(1)$-approxi-mate state $n^2$-design on $n$ qubits are useless for measurement-based quantum computing, in the sense that any computation using such states could be simulated efficiently on a classical computer.  We do this, following \cite{MostStatesUselessMBQCGFE}, by showing that the states are so entangled that the measurement outcomes are essentially random.

\subsection{Optimality of Results}

An important question is how close our results are to optimal, in terms of their scaling with dimension $d$.  In \thmref{ConcentrationPolynomial}, we will normally have $a=\Theta(d)$ so for $m$ constant, we obtain polynomial bounds, rather than the exponential bounds for full randomness.  This is to be expected:
\begin{theorem}
\label{thm:Optimal}
Let $\nu$ be an $\eps$-approximate unitary $k$-design.  Suppose also that it is discrete i.e.~contains a finite number $S$ of unitaries.  Let $f(U)$ be any function on matrix elements of $U$ and $\mu$ be any constant.  Then either $f(U) = \mu$ for all $U$ in $\nu$ or for some $\delta > 0$
\be
\Pr_{U \sim \nu}(|f-\mu| \ge \delta) \ge p_{min}
\ee
where $p_{min}$ is the probability of choosing the least probable unitary from $\nu$.  If the probability is uniform, $p_{min} = 1/S$.
\end{theorem}
\begin{proof}
There exists at least one $U$ such that $|f(U)-\mu| \ge \delta$ for some $\delta > 0$; the probability of selecting one such $U$ is at least $p_{min}$.
\end{proof}
\begin{corollary}
Our results are polynomially related to the optimal (i.e.~the optimal bounds can be obtained by raising ours to a constant power).
\end{corollary}
\begin{proof}
Our results apply for any design, so must obey the bound in \thmref{Optimal} for all designs.  The unitary design construction we use (from \chapref{TPE} using \lemref{ApproxUnitaryDesignEquiv}) has $p_{min}=d^{-O(k)}$ hence the bounds cannot scale better than this.
\end{proof}
We can also almost recover the tail bound for full randomness in \thmref{ConcentrationPolynomial}.  Suppose for simplicity that we have an exact design (i.e.~$\eps=0$), so that
\bes
\Pr_{U \sim \nu}(|f - \mu| \ge \delta) \le C \left(\frac{m}{a\delta^2}\right)^m.
\ees
The optimal $m$ is $a \delta^2/e$, which gives
\bes
\Pr_{U \sim \nu}(|f - \mu| \ge \delta) \le C e^{-a \delta^2/e}.
\ees
So our result allows us to interpolate from Markov's inequality, which gives weak bounds, all the way to full Haar randomness and is within a polynomial correction of optimal for the full range.

The remainder of the chapter is organised as follows.  In \secref{MainTechnique} we present our main technique for finding large deviation bounds for $k$-designs.  We then apply this to entropy in \secref{Entropy}, to ideas in statistical mechanics in \secref{StatMech} and to using $k$-designs for measurement-based quantum computing in \secref{MBQC}.  We then conclude in \secref{ConclusionLargeDeviations}.

\section{Main Technique}
\label{sec:MainTechnique}

The main idea in this chapter can be summarised in three steps.  Let $f: \cU(d) \rightarrow \mathbb{C}$ be a balanced polynomial of degree $K$ in the matrix elements of a unitary $U$.  Then to get a concentration bound on $f$ when $U$ is chosen from a $k$-design:
\begin{enumerate}
\item{Find some measure concentration result for $|f(U) - \mu|$ when the unitaries are chosen uniformly at random from the Haar measure.  Normally $\mu$ will be the expectation of $f$.}
\item{Use this to bound the moments $\bbE |f(U) - \mu|^{2m}$ for some integer $m \le \frac{k}{2K}$.}
\item{Then use Markov's inequality and the fact that for a (approximate) $k$-design the moments are (almost) the same as for uniform randomness.  We then optimise the bound for $m$, which will often involve setting $m$ close to the maximum, $\l\lfloor\frac{k}{2K}\r\rfloor$.}
\end{enumerate}
We will now work through each of these steps and finish with a proof of \thmref{ConcentrationPolynomial}.

\subsection{Step 1: Concentration for uniform randomness}
For the first step, we will often start with Levy's Lemma.  This states, roughly speaking, that slowly varying functions in high dimensions are approximately constant.  We quantify `slowly varying' by the Lipschitz constant:
\begin{definition}
The Lipschitz constant $\eta$ (with respect to the Euclidean norm) for a function $f$ is
\begin{equation}
\eta = \sup_{U_1, U_2} \frac{| f(U_1) - f(U_2) |}{|| U_1 - U_2 ||_2}.
\end{equation}
\end{definition}
Then we have Levy's lemma:
\begin{lemma}[Levy, see e.g.~\cite{Ledoux}]
\label{lem:Levy}
Let $f$ be an $\eta$-Lipschitz function on $U(d)$ with mean $\bbE f$.  Then
\begin{equation}
\Pr(|f - \bbE f| \ge \delta) \le 4 \exp\left(-\frac{C_1 d \delta^2}{\eta^2}\right)
\end{equation}
where $C_1$ can be taken to be $\frac{2}{9 \pi^3}$.
\end{lemma}

\subsection{Step 2: A bound on the moments}
Levy's Lemma says that $f$ is close to its mean.  This means that $\bbE |f - \bbE f|^m$ should be small.  We will bound the moments for slightly more general concentration results:
\begin{lemma}
\label{lem:MomentBoundFromLargeDeviationBound}
Let $X$ be any random variable with probability concentration
\begin{equation}
\label{eq:LargeDeviationAssumption}
\Pr(|X - \mu| \ge \delta) \le C e^{-a \delta^2}.
\end{equation}
(Normally $\mu$ will be the expectation of $X$, although the bound does not assume this.)  Then
\begin{equation}
\bbE |X - \mu|^m \le C \Gamma(m/2+1) a^{-m/2} \le C \left( \frac{m}{2a} \right)^{m/2}
\end{equation}
for any $m > 0$.
\end{lemma}
\begin{proof}
This proof is based on the proof of an analogous result by Bellare and Rompel \cite{BellareRompel}, Lemma A.1.

Note that, for any random variable $Y \ge 0$, 
\begin{equation}
\bbE Y = \int_0^\infty \Pr(Y \ge y) dy.
\end{equation}
Therefore
\bas
\bbE |X-\mu|^m &= \int_0^\infty \Pr(|X-\mu|^m \ge x) dx \\
&= \int_0^\infty \Pr(|X-\mu| \ge x^{1/m}) dx \\
&\le C \int_0^{\infty} \exp(-a x^{2/m}) dx
\eas
where in the last line we used the assumed large deviation bound \eq{LargeDeviationAssumption}.  To evaluate this integral, use the change of variables $y = a x^{2/m}$ to get
\bas
\bbE |X-\mu|^m &\le \frac{Cm}{2} a^{-m/2} \int_0^\infty e^{-y} y^{m/2-1} dy\\
&= C a^{-m/2} \Gamma(m/2+1) \\
&\le C \left( \frac{m}{2a} \right)^{m/2}.\qedhere
\eas
\end{proof}

\subsection{Step 3: A concentration bound for a \texorpdfstring{$k$}{k}-design}

Now we show how to obtain a measure concentration result for polynomials when the unitaries are selected from an approximate $k$-design.  We first show that the moments of $|f-\mu|$ for $f$ a polynomial are close to the Haar measure moments:
\begin{lemma}
\label{lem:ApproxkdesignMoments}
Let $f$ be a balanced polynomial of degree $K$ and $\mu$ be any constant.  Let $f = \sum_{i=1}^t \alpha_i M_i$ where each $M_i$ is a monomial.  Let $\alpha(f) = \sum_i |\alpha_i|$.  Then for $m$ an integer with $2mK \le k$ and $\nu$ an $\eps$-approximate $k$-design,
\be
\bbE_{U \sim \nu} |f - \mu|^{2m} \le \bbE_{U \sim \cU(d)} | f-\mu |^{2m} + \frac{\eps}{d^k} \left(\alpha + | \mu | \right)^{2m}.
\ee
\end{lemma}
\begin{proof}
For simplicity, we assume that $f$ and $\mu$ are real.  Our proof easily generalises to the complex case.

Firstly we calculate $|\bbE_{U \sim \nu} f^i - \bbE_{U \sim \cU(d)} f^i|$ using the multinomial theorem:
\bas
|\bbE_{U \sim \nu} f^i - &\bbE_{U \sim \cU(d)} f^i| \\
&=\l| \sum_{k_1 + \ldots + k_t = i} {i \choose k_1, \ldots, k_t} \alpha_1^{k_1} \ldots \alpha_t^{k_t} \l( \bbE_{U \sim \nu} - \bbE_{U \sim \cU(d)} \r) M_1^{k_1} \ldots M_t^{k_t} \r| \\
&\le \sum_{k_1 + \ldots + k_t = i} {i \choose k_1, \ldots, k_t} |\alpha_1|^{k_1} \ldots |\alpha_t|^{k_t} \l| \l( \bbE_{U \sim \nu} - \bbE_{U \sim \cU(d)} \r) M_1^{k_1} \ldots M_t^{k_t} \r| \\
&\le \frac{\eps}{d^k} \sum_{k_1 + \ldots + k_t = i} {i \choose k_1, \ldots, k_t} |\alpha_1|^{k_1} \ldots |\alpha_t|^{k_t} \\
&= \frac{\eps}{d^k} \alpha^i.
\eas
We now calculate $\bbE_{U \sim \nu} |f - \mu|^{2m}$:
\bas
\l| \bbE_{U \sim \nu} | f - \mu |^{2m} - \bbE_{U \sim \cU(d)} | f - \mu |^{2m} \r|
&= \l| \bbE_{U \sim \nu} (f - \mu)^{2m} - \bbE_{U \sim \cU(d)} (f - \mu)^{2m} \r| \\
&= \l|\sum_{i = 0}^{2m} {2m \choose i} (\bbE_{U \sim \nu} f^i - \bbE_{U \sim \cU(d)} f^i) (-\mu)^{2m-i} \r|\\
&\le \sum_{i = 0}^{2m} {2m \choose i} |\bbE_{U \sim \nu} f^i - \bbE_{U \sim \cU(d)} f^i| |\mu|^{2m-i}\\
&\le \frac{\eps}{d^k} \sum_{i = 0}^{2m} {2m \choose i} \alpha^i |\mu|^{2m-i}\\
&= \frac{\eps}{d^k} \left(\alpha + | \mu | \right)^{2m}.\qedhere
\eas
\end{proof}

Now we can simply apply Markov's inequality to prove \thmref{ConcentrationPolynomial}.
\begin{proof}[Proof of \thmref{ConcentrationPolynomial}]
Apply Markov's inequality and Lemmas \ref{lem:MomentBoundFromLargeDeviationBound} and \ref{lem:ApproxkdesignMoments}:
\bas
\Pr_{U \sim \nu}(|f - \mu| \ge \delta) &= \Pr_{U \sim \nu}(|f - \mu|^{2m} \ge \delta^{2m}) \\
&\le \frac{\bbE_{U \sim \nu} | f - \mu |^{2m}}{\delta^{2m}} \\
&\le \frac{1}{\delta^{2m}} \left( C \left(\frac{m}{a}\right)^m + \frac{\eps}{d^k} \left(\alpha + | \mu | \right)^{2m}\right).\qedhere
\eas 
\end{proof}

We finish this section with two remarks.  Firstly, provided $\alpha(f)$ (the sum of the absolute value of all the coefficients) is at most polynomially large in $d$, we can choose $\eps$ to be polynomially small to cancel this at no change to the asymptotic efficiency.  Secondly, when applying the theorem we will optimise the choice of $m$ (and normally choose $k = 2mK$).  Often $a=\Theta(d)$ and the optimal choice of $m$ is often $\Theta(d)$ as well.  However, we will not take $m$ so large because we can only implement an efficient $k$-design for $k = O(\log d/\log \log d)$.

\section{Application 1: Entropy of a \texorpdfstring{$k$}{k}-design}
\label{sec:Entropy}

We now apply the above to show that most unitaries in a $k$-design generate large amounts of entropy across any bipartition, provided the dimensions are sufficiently far apart.  This means that, for any initial state, for most choices of a unitary from a $k$-design applied to the state, the resulting state will be highly entangled.  We go via the purity of the reduced density matrix, since the entropy function is not a polynomial.

We will call the two systems $S$ (the `system') and $E$ (the `environment') and calculate the purity of the reduced state.  That the purity, $\tr\left[\left( \tr_E U \rho U^\dagger \right)^2\right]$, is a balanced polynomial of degree 2 is easily seen by noting that the trace is linear and the reduced state is squared.  However, we should check that there are not too many terms or terms with large coefficients.  To do this, we should calculate $\alpha$ to apply \thmref{ConcentrationPolynomial}.

There is a general method for calculating $\alpha(f)$ which we will use.  Write $f(U) = \sum_i \alpha_i M_i(U)$ for monomials $M_i$.  To evaluate $\alpha(f) = \sum_i | \alpha_i |$, calculate $f(A)$ where $A$ is the matrix with all entries equal to $1$ (so that $M_i(A) = 1$) and replace $\alpha_i$ with $| \alpha_i|$. Using this here we find
\bas
\alpha &= d^2\l(\sum_{i j} | \rho_{ij} | \r)^2 \\
&\le d^4 \sum_{ij} | \rho_{ij} |^2 \\
&= d^4 || \rho ||_2^2 \\
&\le d^4.
\eas
We now calculate the expected purity:
\begin{lemma}
\label{lem:ExpectedPurity}
The expected purity of the reduced state is $\frac{d_S + d_E}{d+1}$, where $d_S$ is the dimension of subsystem $S$ and $d_E = d/d_S$ is the dimension of subsystem $E$.
\end{lemma}
\begin{proof}
We have
\be
\bbE_{U \sim \cU(d)} || \psi_S ||_2^2 = \bbE_{U \sim \cU(d)} \l[ \tr \cF_{S_1 S_2} (\tr_E U \rho U^\dagger \ot \tr_E U \rho U^\dagger) \r]
\ee
where $\cF_{S_1 S_2}$ is swap acting between systems $S_1$ and $S_2$.  By linearity of the trace, we can commute the $\bbE_{U \sim \cU(d)}$ through and use $\bbE_{U \sim \cU(d)} \l[ U \rho U^\dagger \ot U \rho U^\dagger \r] = \frac{I_{12} + \cF_{12}}{d(d+1)}$ to find
\bas
\bbE_{U \sim \cU(d)} || \psi_S ||_2^2 &= \tr \l[\frac{\cF_{S_1 S_2}}{d(d+1)} (d_E^2 I_{S_1 S_2} + d_E \cF_{S_1 S_2}) \r] \\
&= \frac{1}{d(d+1)} (d_E^2 d_S + d_E d_S^2) \\
&= \frac{d_S + d_E}{d+1}\qedhere
\eas
\end{proof}

Working out the higher moments in this way is difficult (although has been done in \cite{GiraudPurityMoments}) so we use Levy's Lemma and \lemref{MomentBoundFromLargeDeviationBound}.  To use Levy's Lemma, all we have to do is find the Lipschitz constant for the purity:
\begin{lemma}
The Lipschitz constant for purity is $\le 2$.
\end{lemma}
\begin{proof}
\bas
\eta &= \sup_{\psi, \phi} \frac{\left| ||\psi_S||_2^2 - ||\phi_S||_2^2 \right|}{||\psi - \phi||_2} \\
&= \sup_{\psi, \phi} \frac{\left| ||\psi_S||_2 - ||\phi_S||_2 \right| (||\psi_S||_2 + ||\phi_S||_2)}{||\psi - \phi||_2}
\eas
Now we use $\left| ||S||_2 - ||T||_2 \right| \le ||S-T||_2$ to find
\bes
\eta \le \sup_{\psi, \phi} (||\psi_S||_2 + ||\phi_S||_2) \le 2
\ees
using the fact that the purity is upper bounded by 1.
\end{proof}

\begin{lemma}
\label{lem:MessyEntropyBound}
For $\mu = \frac{d_S + d_E}{d+1}$ and $m$ an integer with $m \le k/4$ and $\nu$ an $\eps$-approximate $k$-design,
\be
\label{eq:MessyEntropyBound}
\Pr_{U \sim \nu}(S(\psi_S) \le -\log_2 \mu - \alpha) \le \frac{1}{(\mu (2^\alpha-1))^{2m}} \left(4 \left(\frac{4 m}{C_1 d}\right)^m + \frac{\eps}{d^k}(d^4 + \mu)^{2m} \right).
\ee
\end{lemma}
\begin{proof}
We use the fact that von Neumann entropy is lower bounded by the Renyi 2-entropy i.e.~$-\log_2 || \psi_S ||_2^2$:
\begin{equation}
S(\psi_S) \ge S_2(\psi_S) = - \log_2 || \psi_S ||_2^2.
\end{equation}
Then
\bas
\Pr_{U \sim \nu}(S(\psi_S) \le - \log_2 (1+\delta) \mu) &\le \Pr_{U \sim \nu}(S_2(\psi_S) \le - \log_2 (1+\delta) \mu) \\
&= \Pr_{U \sim \nu}(|| \psi_S ||_2^2 \ge (1+\delta)\mu) \\
&\le \Pr_{U \sim \nu}(\left| || \psi_S ||_2^2 - \mu \right|\ge \delta\mu) \\
&\le \frac{1}{(\mu \delta)^{2m}} \left(4 \left(\frac{4 m}{C_1 d}\right)^m + \frac{\eps}{d^k}(d^4 + \mu)^{2m} \right)
\eas
using \thmref{ConcentrationPolynomial} in the last line.
\end{proof}
We have written this in a more convenient form in \thmref{EntropyBound} which is proved in \secref{EntropyBoundProof}.  This is to be compared with the Haar random version \thmref{EntropyTailBoundFullRandomness}.  As expected, we have $n = \log_2 d$ appearing in the exponent rather than $d$.  Note also that our bound does not work well for $d_S \approx d_E$.  In fact, in this case, we do not get a bound that improves with dimension.  In order to achieve such a bound in this regime a different technique will be necessary.

\section{Application 2: \texorpdfstring{$k$}{k}-designs and Statistical Mechanics}
\label{sec:StatMech}

We can also apply these ideas to partially derandomise some of the arguments on the foundations of statistical mechanics in \cite{ThermalisationPSW}.  In this paper, the authors develop the idea that the uncertainty in statistical mechanics comes from entanglement rather than the traditional assumption of the principle of equal a priori probabilities.  They consider the universe being in a pure quantum state and that the uncertainty in the state of a subsystem comes from the entanglement between this system and the rest of the universe.

The setting is that there is an arbitrary global linear constraint $R$.  Often this will be a total energy constraint although this is not assumed.  Let the Hilbert space of states satisfying $R$ be $\cH_R$.  Then let the system and environment Hilbert spaces be $\cH_S$ and $\cH_E$ respectively.  Then
\be
\cH_R \subseteq \cH_S \ot \cH_E.
\ee
Let the dimensions be $d_R$, $d_S$ and $d_E$ and let $\cE_R = \frac{I_R}{d_R}$.  Note that $d_R \le d_S d_E$, unlike in the above where we took $d = d_S d_E$.  Normally we will have $d_S \ll d_R$.  The principle of equal a priori probabilities says that the state of the universe is $\cE_R$ which implies the subsystem state is the canonical state, given by
\be
\label{eq:CanonicalState}
\Omega_S = \tr_E(\cE_R).
\ee
The main result of \cite{ThermalisationPSW} (the `principle of \emph{apparently} equal a priori probabilities') is that, for almost all pure states of the universe, the subsystem state is almost exactly the canonical state.  
\begin{theorem}[Theorem 1 of \cite{ThermalisationPSW}]
\label{thm:ThermalisationPSW}
For a randomly chosen state $\ket{\phi} \in \mathcal{H}_R \subseteq \mathcal{H}_S \ot \mathcal{H}_E$ and arbitrary $\eps > 0$, the distance between the reduced density matrix of the system $\rho_S = \tr_E(\ket{\phi} \bra{\phi})$ and the canonical state $\Omega_S$ (\eq{CanonicalState}) is given probabilistically by
\begin{equation}
\Pr_{U \sim \cU(d)}\left( || \rho_S - \Omega_S ||_1 \ge \eps + \sqrt{\frac{d_S}{d_E^{\text{\rm eff}}}}\right) \le 2 \exp\left(-C_2 d_R \eps^2\right)
\end{equation}
where $C_2 = 1/(18 \pi^3)$ and $d_E^{\text{\rm eff}} = \frac{1}{\tr \Omega_E^2} \ge \frac{d_R}{d_S}$.
\end{theorem}

This result gives compelling evidence to replace the principle of equal a priori probabilities with the principle of apparently equal a priori probabilities, but it does not address the problem of how the system reaches this state.  It will take an extremely (exponentially) long time for the universe to reach a random pure state, in contrast to the observed fact that thermalisation occurs quickly.  Here, we show that for almost all unitaries in a $k$-design applied to the universe, the subsystem state is close to the canonical state.  Since these unitaries can be implemented and sampled efficiently, this means that equilibrium could be reached quickly to match observations.

We are now ready to show that a $k$-design gives a small $|| \rho_S - \Omega_S ||_1$.  First, we have to modify \lemref{MomentBoundFromLargeDeviationBound} slightly:
\begin{lemma}
\label{lem:MomentBoundFromLargeDeviationBound2}
Let $X$ be any non-negative random variable with probability concentration
\begin{equation}
\label{eq:LargeDeviationAssumption2}
\Pr(X \ge \delta + \eta) \le C e^{-a \delta^2}
\end{equation}
where $\eta \ge 0$.  Then
\begin{equation}
\bbE X^m \le C \left( \frac{2m}{a} \right)^{m/2} + (2\eta)^m
\end{equation}
for any $m > 0$.
\end{lemma}
\comment{\begin{proof}
For $\delta \ge \eta$, $\Pr(X \ge 2\delta) \le C e^{-a \delta^2}$.  Alternatively, if $\delta \ge 2\eta$, $\Pr(X \ge \delta) \le C e^{-a \delta^2/4}$.  Now
\begin{align}
\Expect X^m &= \int_0^\infty \Pr(X^m \ge x) dx \\
&= \int_0^\infty \Pr(X \ge x^{1/m}) dx \\
&= \int_{(2\eta)^m}^\infty \Pr(X \ge x^{1/m}) dx + \int_0^{(2\eta)^m} \Pr(X \ge x^{1/m}) dx \\
&\le \int_{(2\eta)^m}^\infty C \exp(-a x^{2/m}/4) dx + (2\eta)^m \\
&\le \int_0^\infty C \exp(-a x^{2/m}/4) dx + (2\eta)^m \\
&\le C \left( \frac{2m}{a} \right)^{m/2} + (2\eta)^m
\end{align}
where the last line follows from evaluating the integral, as in the proof of \lemref{MomentBoundFromLargeDeviationBound}.
\end{proof}}
The proof is very similar to the proof of \lemref{MomentBoundFromLargeDeviationBound}.

Now we state and prove the main result in this section:
\begin{theorem}
Let $\nu$ be an $\eps$-approximate unitary $k$-design.  Then
\be
\label{eq:ThermalisationMessy}
\Pr_{U \sim \nu}( || \rho_S - \Omega_S ||_1 \ge \delta) \le \l(\frac{d_S}{\delta^2}\r)^{k/8} \l( 2 \l(\frac{k}{2 C_2 d_R}\r)^{k/8} + \l( \frac{4 d_S^2}{d_R} \r)^{k/8} + \frac{\eps}{d_R^k}(d_R^2+1)^{k/2}\r).
\ee
In particular, with $\eps = \frac{3}{2}\l( \frac{4 d_S^3}{d_R} \r)^{k/8}$, $k \le 8C_2 d_S^2$,
\be
\label{eq:ThermalisationSimplified}
\Pr_{U \sim \nu}( || \rho_S - \Omega_S ||_1 \ge \delta) \le 6 \l(\frac{4 d_S^3}{d_R \delta^2} \r)^{k/8}.
\ee
\end{theorem}
Again, we need $d_S$ to be polynomially smaller than $d_R$ to obtain non-trivial bounds.

\begin{proof}
We go via the 2-norm and use Lemmas \ref{lem:MomentBoundFromLargeDeviationBound2} and \ref{lem:ApproxkdesignMoments}.

We have from \thmref{ThermalisationPSW} that
\begin{equation}
\Pr_{U \sim \cU(d)}(||\rho_S - \Omega_S||_1 \ge \delta + \eta) \le 2 e^{-C_2 d_R \delta^2}
\end{equation}
where $\eta = \sqrt{\frac{d_S}{d_E^\text{eff}}} \le \frac{d_S}{\sqrt{d_R}}$.  Since $||\rho_S - \Omega_S||_2 \le ||\rho_S - \Omega_S||_1$,
\begin{equation}
\Pr_{U \sim \cU(d)}(||\rho_S - \Omega_S||_2 \ge \delta + \eta) \le 2 e^{-C_2 d_R \delta^2}.
\end{equation}
We now apply \lemref{MomentBoundFromLargeDeviationBound2} to get
\begin{equation}
\bbE_{U \sim \cU(d)} ||\rho_S - \Omega_S||_2^{2m} \le 2\left(\frac{4m}{C_2 d_R}\right)^{m} + (2 \eta)^{2m}.
\end{equation}
So for $m \le k/4$, using Markov's inequality and \lemref{ApproxkdesignMoments} (with $\mu = 0$) on the polynomial $||\rho_S - \Omega_S||_2^{2}$ :
\be
\Pr_{U \sim \nu} (||\rho_S - \Omega_S||_2 \ge \delta) \le \frac{1}{\delta^{2m}} \left( 2\left(\frac{4m}{C_2 d_R}\right)^{m} + (2 \eta)^{2m} + \frac{\eps}{d_R^k}(d_R^2+1)^{4m} \r).
\ee
Here, we used an estimate of $\alpha$, the sum of the moduli of the coefficients:
\be
\alpha \le (d_R^2+1)^2
\ee
which we obtain via a similar calculation to that in \secref{Entropy}.

Now we go back to the 1-norm, using $||\rho_S - \Omega_S||_1 \le \sqrt{d_S} ||\rho_S - \Omega_S||_2$ to get
\begin{align}
\Pr_{U \sim \nu} (||\rho_S - \Omega_S||_1 \ge \delta) &\le \Pr_{U \sim \nu} (||\rho_S - \Omega_S||_2 \ge \delta/\sqrt{d_S}) \\
&\le \l(\frac{d_S}{\delta^2}\r)^{m} \l( 2 \l(\frac{4m}{C_2 d_R}\r)^{m} + \l( 2 \eta \r)^{2m} + \frac{\eps}{d_R^k}(d_R^2+1)^{4m} \r).
\end{align}
To obtain the result in \eq{ThermalisationMessy}, we just use $\eta \le \frac{d_S}{\sqrt{d_R}}$ and set $m=k/8$.

To prove the simplified version, first use, as in \secref{Entropy}, that $(d_R^2+1)^{4m} \le 2 d_R^{8m}$ for $m \le d_R^2/8$.  This is implied by $k \le 8C_2 d_S^2$.  We then set $m = k/8$ to find
\be
\Pr_{U \sim \nu}( || \rho_S - \Omega_S ||_1 \ge \delta) \le 2 \l(\frac{k d_S}{2 C_2 d_R \delta^2}\r)^{k/8} + \l( \frac{4 d_S^3}{d_R \delta^2} \r)^{k/8} + 2\frac{\eps}{\delta^{k/4}}.
\ee
Then, using $k \le 8 C_2 d_S^2$, with $\eps \le \frac{3}{2}\l( \frac{4 d_S^3}{d_R} \r)^{k/8}$, we obtain the simplified result \eq{ThermalisationSimplified}.
\end{proof}

\section{Application 3: Using \texorpdfstring{$k$}{k}-designs for Measurement-Based Quantum Computing}
\label{sec:MBQC}

Here we apply our ideas to partially derandomise some results of Gross, Flammia and Eisert in \cite{MostStatesUselessMBQCGFE} and Bremner, Mora and Winter in \cite{MostStatesUselessMBQCBMW}.  The main result in these two papers is that most states do not offer any advantage over classical computation when used in the measurement-based quantum computing (MBQC) model.  In MBQC, a classical computer is given access to a large quantum state on which it can do single qubit measurements.  Some states allow for universal quantum computation whereas others do not add any extra power to the classical computer.  These results are concerned with the question of characterising which states do and do not work.  Showing that random states do not give any speed up shows that useful states for MBQC are not generic and so must be carefully constructed.

While the results in these two papers are similar, we will concentrate on the methods from \cite{MostStatesUselessMBQCGFE} since their methods are simpler to apply here.  They prove their result by showing that most states are very entangled in the geometric measure (see \defref{GeometricMeasure}).  They then use this to show that the measurement outcomes of even the best possible measurement scheme are almost completely random.  In fact, the state could be thrown away and the measurement outcomes replaced with random numbers to solve the computational problem just as efficiently.  This shows that you can classically simulate any quantum computation that uses these highly entangled states.  The measure of entanglement they use is the geometric measure:
\begin{definition}
\label{def:GeometricMeasure}
The geometric measure of entanglement of a state $\ket{\Psi}$ is \cite{GeometricEntanglementShimony, BarnumLinden01}
\be
E_g(\ket{\Psi}) = -\log_2 \sup_{\alpha \in \cP} | \braket{\alpha}{\Psi} |^2.
\ee
where $\cP$ is the set of all product states.
\end{definition}
They show that any MBQC using a state $\ket{\Psi}$ with $E_g(\ket{\Psi}) = n - O(\log_2 n)$ can be efficiently simulated classically.  They then show that
\begin{theorem}[\cite{MostStatesUselessMBQCGFE}, Theorem 2]
\label{thm:HaarRandomLargeGeometricEntropy}
For $n \ge 11$,
\be
\Pr_{\ket{\psi} \sim \cS(d)} (E_g(\ket{\Psi}) \le n - 2 \log_2 n - 3) \le e^{-n^2}.
\ee
\end{theorem}
This shows that most states are useless.  We partially derandomise this result to show that most states in an $\eps$-approximate ($\eps$ can be taken as a constant) state $n^2$-design have high geometric measure of entanglement and thus are useless in the same way.

We could apply our technique and use \thmref{ConcentrationPolynomial} but in this case, it is simpler to directly bound the probability using Markov's inequality.
\begin{lemma}
\label{lem:RandomStateOverlap}
\begin{equation}
\Pr_{\ket{\Psi} \sim \nu}(| \braket{\Phi}{\Psi} |^2 \ge \delta) \le (1+\eps)\frac{m!}{(d\delta)^m} \le (1+\eps)\left(\frac{m}{d\delta}\right)^m
\end{equation}
where $\ket{\Psi}$ is chosen from an $\eps$-approximate state $k$-design $\nu$, $m \le k$ and a positive integer and $\ket{\Phi}$ is any fixed state.
\end{lemma}
\begin{proof}
We prove this bound directly using Markov's inequality:
\bas
\Pr_{\ket{\Psi} \sim \nu}(| \braket{\Phi}{\Psi} |^2 \ge \delta) &= \Pr_{\ket{\Psi} \sim \nu}(| \braket{\Phi}{\Psi} |^{2m} \ge \delta^m) \\
&\le \frac{\bbE_{\ket{\Psi} \sim \nu}  | \braket{\Phi}{\Psi} |^{2m}}{\delta^m} \\
&= \frac{\bbE_{\ket{\Psi} \sim \nu}  \bra{\Phi}^{\ot m} \ket{\Psi}^{\ot m} \bra{\Psi}^{\ot m} \ket{\Phi}^{\ot m}}{\delta^m} \\
&= \frac{\bra{\Phi}^{\ot m} \bbE_{\ket{\Psi} \sim \nu} \l[ \ket{\Psi}^{\ot m} \bra{\Psi}^{\ot m} \r] \ket{\Phi}^{\ot m}}{\delta^m} \\
&\le \frac{\bra{\Phi}^{\ot m} (1+\eps)\frac{\Pi^\text{sym}_m}{{m+d-1 \choose d-1}} \ket{\Phi}^{\ot m}}{\delta^m} \\
&= \frac{1+\eps}{{m+d-1 \choose d-1} \delta^m} \\
&\le \frac{(1+\eps)m!}{(d \delta)^m} \le (1+\eps)\left(\frac{m}{d \delta} \right)^m.\qedhere
\eas
\end{proof}
We now prove the main result in this section:
\begin{theorem}
\label{thm:LargeGeometricEntanglement}
For $\ket{\Psi}$ randomly drawn from an $\eps$-approximate state $k$-design with $d = 2^n$
\begin{equation}
\Pr_{\ket{\Psi} \sim \nu}(E_g(\ket{\Psi}) \le n - \delta) \le (1+\eps) \exp_2 (k \log_2 2k + 4n \log_2 10n - k\delta + 4n(n-\delta)).
\end{equation}
In particular, for $k = n^2$, $\delta = 3 \log_2 n + 5$ and $\eps = 1$,
\begin{equation}
\Pr_{\ket{\Psi} \sim \nu}(E_g(\ket{\Psi}) \le n - 3 \log_2 n - 5) \le 2 \cdot n^{-n^2}.
\end{equation}
\end{theorem}
We note that this bound is almost the same as in \thmref{HaarRandomLargeGeometricEntropy}.  It only works for slightly larger deviations from $n$, which is why we obtain a slightly better probability bound.  Note also that we can obtain an exponential bound in $n$ (not $d=2^n$) because the design is exponentially large in $n$.
\begin{proof}
This proof closely mirrors the proof of Theorem 2 in \cite{MostStatesUselessMBQCGFE}.  We use the idea of a $\gamma$-net.  $\cN_{\gamma, n}$ is a $\gamma$-net on product states if
\be
\sup_{\ket{\alpha} \in \cP} \inf_{\ket{\tilde{\alpha}} \in \cN_{\delta, n}} \big|\big| \ket{\alpha} - \ket{\tilde{\alpha}} \big|\big|_2 \le \gamma/2.
\ee
In \cite{MostStatesUselessMBQCGFE}, it is shown that such a net exists with $| \mathcal{N}_{\gamma, n}| \le (5n/\gamma)^{4n}$.  We then proceed by showing that most states in the state design have small overlap with every state in the net using the union bound and \lemref{RandomStateOverlap}.  Finally, since every state is close to one in the net, we can show that most states in the design have small overlap with every product state.

We now formalise the above.  Using \lemref{RandomStateOverlap} and the union bound,
\begin{equation}
\label{eq:NetStateOverlap}
\Pr_{\ket{\Psi} \sim \nu} \left( \sup_{\ket{\tilde{\alpha}} \in \mathcal{N}_{\gamma, n}} | \braket{\tilde{\alpha}}{\Psi} |^2 \ge \delta'/2 \right) \le | \mathcal{N}_{\gamma,n} | (1+\eps) \left(\frac{2k}{d\delta'}\right)^k \le \left(\frac{5n}{\gamma}\right)^{4n} (1+\eps) \left(\frac{2k}{2^n \delta'}\right)^k.
\end{equation}
Now, we need to bound
\bas
\Pr_{\ket{\Psi} \sim \nu}(E_g(\ket{\Psi}) \le n - \delta) &= \Pr_{\ket{\Psi} \sim \nu}\left(-\log_2 \sup_{\ket{\alpha} \in \mathcal{P}} | \braket{\alpha}{\Psi} |^2 \le n - \delta\right) \\
&= \Pr_{\ket{\Psi} \sim \nu}\left(\sup_{\ket{\alpha} \in \mathcal{P}} | \braket{\alpha}{\Psi} |^2 \ge 2^{-(n - \delta)}\right).
\eas
We now claim that
\be
\label{eq:OverlapInNet}
\sup_{\ket{\alpha} \in \cP} | \braket{\alpha}{\Psi} |^2 \ge \delta' \Rightarrow \sup_{\ket{\tilde{\alpha}} \in \cN_{\delta'/2, n}} | \braket{\tilde{\alpha}}{\Psi} |^2 \ge \delta'/2.
\ee
To prove this claim, let $\ket{\alpha}$ be the state that achieves the supremum on the left hand side, and let $\ket{\tilde{\alpha}}$ be the state closest to it in the $\delta'/2$-net.  It is shown in \cite{MostStatesUselessMBQCGFE} that this implies for any $\ket{\Psi}$
\be
\left| | \braket{\alpha}{\Psi} | ^2 - | \braket{\tilde{\alpha}}{\Psi} | ^2 \right| \le \delta'/2.
\ee
Therefore
\bas
| \braket{\tilde{\alpha}}{\Psi} |^2 &\ge  | \braket{\alpha}{\Psi} |^2 - \delta'/2 \\
&\ge \delta'/2.
\eas
This implies that the supremum over all states in the net must be at least $\delta'/2$ to prove the claim.

We can now finish the proof.  Set $\delta' = 2^{-(n - \delta)}$ in \eq{OverlapInNet} and use \eq{NetStateOverlap} with $\gamma = \delta'/2$ to find
\bas
\Pr_{\ket{\Psi} \sim \nu}&\left(\sup_{\ket{\alpha} \in \mathcal{P}} | \braket{\alpha}{\Psi} |^2 \ge 2^{-(n - \delta)}\right) \\
&\le \Pr_{\ket{\Psi} \sim \nu}\left(\sup_{\ket{\tilde{\alpha}} \in \mathcal{N}_{2^{-(n-\delta)-1}, n}} | \braket{\tilde{\alpha}}{\Psi} |^2 \ge 2^{-(n - \delta)-1}\right) \\
&\le (1+\eps) \exp_2 (k \log_2 2k + 4n \log_2 10n - k\delta + 4n(n-\delta)).\qedhere
\eas
\end{proof}
Combining this with the arguments of \cite{MostStatesUselessMBQCGFE} shows that most states in a state $n^2$-design on $n$ qubits are useless for MBQC.  This shows that even many efficiently preparable states are useless.

\section{Conclusions}
\label{sec:ConclusionLargeDeviations}

We have seen how to turn large deviation bounds for Haar-random unitaries into bounds for $k$-designs.  The main technique was applied to show that unitaries from $k$-designs generate large amounts of entanglement.  Then we showed that, if the dynamics of the universe produced a $k$-design, the entanglement generated would be sufficient to reproduce the principle of equal a priori probabilities.  Finally we showed that most states in sufficiently large state designs are useless for measurement-based quantum computing, in the sense that computation using them can be efficiently simulated classically.

However, there are other bounds for which our technique does not work.  Since we cannot obtain exponential bounds for polynomially sized designs, our technique cannot directly derandomise some bounds.  Some results, for example showing that the $\infty$-norm of the reduced state of a random pure state is close to $1/d_S$ \cite{SuperdenseCodingHHL}, are proven by using an $\eps$-net of states and the union bound.  Since the $\eps$-net is exponentially large, exponentially small bounds are required.  We do not know how to apply our idea to results of this kind and still have $k=\poly(\log d)$.  (Note that we could cope with the $\eps$-net in \secref{MBQC} since it was just a net on product states which is considerably smaller.)

It is also possible that our ideas could be used to completely derandomise some constructions (e.g.~locking \cite{RandomizingQuantumStates04, Locking04}).  If we could show that unitaries drawn from a $k$-design work with non-zero probability, and come up with an efficient sampling method, then we could obtain efficient randomised constructions.

\section{Proof of Theorem 5.2.3} %\thmref{EntropyBound}}
\label{sec:EntropyBoundProof}

Here we prove the more convenient form of \lemref{MessyEntropyBound} stated as \thmref{EntropyBound}.

\begin{proof}[Proof of \thmref{EntropyBound}]

Firstly, we will write the left hand side of \eq{MessyEntropyBound} in a more useful way.  Using $\ln(1+x) \le x$, we find
\bes
-\log_2 \mu \ge \log_2 d_S - \beta
\ees
where $\beta = \frac{1}{\ln 2} \frac{d_S}{d_E}$, following the notation in \cite{AspectsOfGenericEntanglement}.  This means
\bas
\Pr_{U \sim \nu}(S(\psi_S) \le \log_2 d_S - \alpha - \beta) &\le \Pr_{U \sim \nu}(S(\psi_S) \le -\log_2 \mu - \alpha) \\
&\le \frac{1}{(\mu (2^\alpha-1))^{2m}} \left(4 \left(\frac{4 m}{C_1 d}\right)^m + \frac{\eps}{d^k}(d^4 + \mu)^{2m} \right).
\eas
We now simplify the right hand side.  Let $\delta = 2^\alpha - 1$.  For $d_S \ge 2$, we have $\mu \ge 1/d_S$.  We shall also assume that $m = k/8$.  This gives us (using $\mu \le 1$)
\be
\Pr_{U \sim \nu}(S(\psi_S) \le \log_2 d_S - \alpha - \beta) \le \left(\frac{d_S}{\delta}\right)^{k/4} \left(4 \left(\frac{k}{2C_1 d}\right)^{k/8} + \eps\left(1 + \frac{1}{d^4}\right)^{k/4} \right).
\ee
Now, one can easily show (e.g.~by induction on $n$) that
\be
(1+\delta)^n \le 2
\ee
for $2n\delta \le 1$.  We use this for $n = k/4$ and $\delta = 1/d^4$.  The condition is then $k \le 2 d^4$, which we shall assume (we will set $k = \log d/\log \log d$ later).  We now obtain
\be
\Pr_{U \sim \nu}(S(\psi_S) \le \log_2 d_S - \alpha - \beta) \le \left(\frac{d_S}{\delta}\right)^{k/4} \left(4 \left(\frac{k}{2C_1 d}\right)^{k/8} + 2\eps \right).
\ee
We will now take $\eps = 2 \left(\frac{k}{2C_1 d}\right)^{k/8}$, so that the two terms are the same.  $\log 1/\eps$ is $\poly (\log d)$ so this remains efficient.  Now
\be
\Pr_{U \sim \nu}(S(\psi_S) \le \log_2 d_S - \alpha - \beta) \le 8 \left(\frac{d_S^2 k}{2 C_1 d \delta^2}\right)^{k/8}.
\ee
Assuming that $\delta^2 > \frac{k d_S^2}{2 C_1 d}$, we should take $k$ as large as possible up to $\frac{2 C_1 \delta^2 d}{e d_S^2}$, when the right hand side is maximised.  We then find the result after further simplification.
\end{proof}

\part{Quantum Learning}
\label{part:Learning}

\chapter{Learning and Testing Algorithms for the Clifford Group}
\label{chap:LearningCliffords}

\section{Introduction}

A central problem in quantum computing is to determine an unknown quantum state from measurements of multiple copies of the state.  This process is known as quantum state tomography (see \cite{NielsenChuang} and references therein).  By making enough measurements, the probability distributions of the outcomes can be estimated from which the state can be inferred.  A related problem is that of quantum process tomography, where an unknown quantum evolution is determined by applying it to certain known input states.  There are several methods for doing this, including what are known as Standard Quantum Process Tomography \cite{ChuangNielsen97, PoyatosCiracZoller97} and Ancilla Assisted Process Tomography \cite{DArianoPresti01, Leung03}.  These methods work by using state tomography on the output states for certain input states.

However, all these procedures share one important downside: the number of measurements required increases exponentially with the number of qubits.  This already presents problems even with systems achievable with today's technology, for which complete tomographical measurements can take hours (e.g.~\cite{HaffnerTomography}) making tomography of larger systems unfeasible.  Unfortunately this exponential cost is necessary to determine a completely unknown state or process, since there are exponentially many parameters to measure.  To make tomography feasible for larger systems, we need to find a restriction that requires fewer measurements, ideally polynomially many.

One way to improve the measurement, or query, complexity is to assume some prior knowledge of the process.  For example, suppose the process was known to be one of a small number of unitaries, then the task is just to decide which.  This is the approach we take here.  As a simple example, consider being given a black box implementing an unknown Pauli matrix.  By applying this to half a maximally entangled state, the Pauli can be identified with one query.  This is essentially superdense coding \cite{SuperdenseCoding} and is explained in \secref{LearningPaulis}.  Indeed, if the black box performed a tensor product of arbitrary Paulis on $n$ qubits then it too can be identified with just one query.

We extend this to work for elements of the Clifford group (the normaliser of the Pauli group; see \defref{CliffordGroup}) and show that any member of the Clifford group can be learnt with $O(n)$ queries, which we show is optimal.  The Clifford group is an important subgroup of the unitary group that has found uses in quantum error correction and fault tolerance \cite{QECGeometry,ShorFaultTolerance,GottesmanFaultTolerance}.

Then generalising further, we show that elements of the Gottesman-Chuang hierarchy \cite{GottesmanChuang} (see \defref{GottesmanChuangHierarchy}), also known as the $\cC_k$ hierarchy, can also be learnt efficiently.  As the level $k$ increases, the set $\cC_k$ includes more and more unitaries so this implies ever larger sets can be learnt, although the number of queries scales exponentially with $k$.  Our methods also work if the unitary is known to be close to a Clifford (or any element of $\cC_k$ for some known $k$) rather than exactly a Clifford.

We also give a Clifford testing algorithm, which determines whether an unknown unitary is close to a Clifford or far from every Clifford.  This is an extension of the Pauli testing algorithm given in \cite{QBF}.  Indeed, our results are closely related to results in \cite{QBF} and we use some of the algorithms presented there as ingredients.  Our results can also be compared with \cite{AaronsonLearnability}, which contains methods to approximately learn quantum states.  Another related result is that of Aaronson and Gottesman \cite{AaronsonGottesmanStabilisers}, which provides a method of learning stabiliser states with linearly many copies.

We only consider query complexity although, at least for the Clifford group results, our methods are computationally efficient too.

The rest of the chapter is organised as follows.  In \secref{definitions}, we define the Pauli and Clifford groups and the Gottesman-Chuang hierarchy.  In \secref{ExactLearning} we present our algorithm for exact learning of Clifford and $\cC_k$ elements.  In \secref{ApproxLearning} we show how to find the closest element of $\cC_k$ to an unknown unitary.  In \secref{CliffordTesting} we present our Clifford testing algorithm and then conclude in \secref{ConclusionLearning}.

This chapter has been published previously as \cite{LearningCliffords}.

\section{The Pauli and Clifford Groups and the Gottesman-Chuang Hierarchy}
\label{sec:definitions}

Firstly, we define the Pauli group.  Call the set of all Pauli matrices on $n$ qubits $\hat{\cP}$.  We then have $|\hat{\cP}| = 4^n$.  We write matrices in the Pauli basis using the normalisation $\rho = \sum_p \gamma(p) \sigma_p$.  To make $\hat{\cP}$ into a group, the Pauli group $\cP$, we must include each matrix in $\hat{\cP}$ with phases $\{\pm 1, \pm i\}$.

We can now define the Clifford group:
\begin{definition}[The Clifford group]
\label{def:CliffordGroup}
The Clifford group is the normaliser of the Pauli group i.e.
\bes
\cC = \{ U \in \cU(2^n) : U \cP U^\dagger \subseteq \cP \}.
\ees
\end{definition}

Then the Gottesman-Chuang hierarchy is a generalisation:
\begin{definition}[The Gottesman-Chuang hierarchy \cite{GottesmanChuang}]
\label{def:GottesmanChuangHierarchy}
Let $C_1$ be the Pauli group $\cP$.  Then level $C_k$ of the hierarchy is defined recursively:
\bes
\cC_k = \{ U \in \cU(2^n) : U \cP U^\dagger \subseteq \cC_{k-1} \}.
\ees
\end{definition}
By definition, $\cC_2$ is the Clifford group $\cC$.  For $k>2$, $\cC_k$ is no longer a group but contains a universal gate set, whereas $\cC_1$ and $\cC_2$ are not universal.

\section{Learning Gottesman-Chuang Operations}
\label{sec:ExactLearning}

Before we give our algorithm for learning Gottesman-Chuang operations, we present a simple method for learning Pauli operations, which we use as the main ingredient.

\subsection{Learning Pauli Operations}
\label{sec:LearningPaulis}

This is due to \cite{QBF} and is in fact identical to the superdense-coding protocol \cite{SuperdenseCoding}.
\begin{theorem}[\cite{QBF}, Proposition 20]
\label{thm:PauliLearning}
Pauli operations can be identified with one query and in time $O(n)$.
\end{theorem}
\begin{proof}
Apply the operator $\sigma_p$ to half of the maximally entangled state \bes \ket{\psi} = 2^{-n/2} \sum_i \ket{i i}.\ees  For different choices of $\sigma_p$, the resulting states are orthogonal so can be perfectly distinguished:
\bas
\bra{\psi} \l(\sigma_p \ot I\r) \l(\sigma_q \ot I\r) \ket{\psi} &= 2^{-n} \sum_{ij} \bra{ii} \sigma_p \sigma_q \ot I \ket{jj} \\
&= 2^{-n} \sum_{ij} \bra{i} \sigma_p \sigma_q \ket{j} \braket{i}{j} \\
&= 2^{-n} \sum_{i} \bra{i} \sigma_p \sigma_q \ket{i} \\
&= 2^{-n} \tr \sigma_p \sigma_q \\
&= \delta_{pq}.
\eas
The time complexity $O(n)$ comes from the preparation and measurement operations.
\end{proof}

\subsection{Learning Clifford Operations}

We can now present our algorithm for learning Clifford operations to illustrate our main idea for learning unitaries in the Gottesman-Chuang hierarchy.  We will use the fact that knowing how a unitary acts by conjugation on all elements of $\hat{\cP}$ identifies it uniquely (up to phase):
\begin{lemma}
\label{lem:UConjPauli}
Knowing $U \sigma_p U^\dagger$ for all $\sigma_p \in \hat{\cP}$ uniquely determines $U$, up to global phase.
\end{lemma}
\begin{proof}
The Pauli matrices form a basis for all $2^n \times 2^n$ matrices so knowing the action of $U$ on the Paulis is enough to determine the action of $U$ on any matrix up to phase.  The phase cannot be determined because action by conjugation does not reveal the phase.
\end{proof}
Now let $G = \{ \sigma_{x_i}, \sigma_{z_i} \}_{i=1}^n$ where $\sigma_{x_i}$ ($\sigma_{z_i}$) is the matrix with $\sigma_x$ ($\sigma_z$) acting on qubit $i$ and trivially elsewhere.  We think of this as a set of generators for $\hat{\cP}$ since each element of $\hat{\cP}$ can be written as a product of elements of $G$, up to phase.  Using this, knowledge of how $U$ acts on elements of $G$ is sufficient to determine the action on all of $\hat{\cP}$:
\begin{lemma}
\label{lem:UConjGenerator}
$U \sigma_p U^\dagger$ for any $\sigma_p \in \hat{\cP}$ can be calculated from knowledge of $U \sigma_g U^\dagger$ for each $\sigma_g \in G$.
\end{lemma}
\begin{proof}
Let $\sigma_p = \alpha \sigma_{g_1} \ldots \sigma_{g_m}$ for $\sigma_{g_i} \in G$ where $\alpha$ is a phase.  Then
\bes
U \sigma_p U^\dagger = \alpha U \sigma_{g_1} \ldots \sigma_{g_m} U^\dagger = \alpha U \sigma_{g_1} U^\dagger \ldots U \sigma_{g_m} U^\dagger.\qedhere
\ees
\end{proof}
With these definitions and observations, we can now present the Clifford learning algorithm.
\begin{theorem}
\label{thm:CliffordLearning}
Given oracle access to an unknown Clifford operation $C$ and its conjugate $C^\dagger$, $C$ can be determined exactly (up to global phase) with $2n+1$ queries to $C$ and $2n$ to $C^\dagger$.  The algorithm runs in time $O(n^2)$.
\end{theorem}
\begin{proof}
From the definition of the Clifford group, $C \sigma_p C^\dagger \in \cP$ for all $\sigma_p \in \hat{\cP}$.  Note that $C \sigma_p C^\dagger$ is not necessarily a Pauli operator in $\hat{\cP}$ because there is a phase of $\pm 1$ (complex phases are not allowed because $C \sigma_p C^\dagger$ is Hermitian).  Determining which Pauli operator and phase for every $\sigma_p$ would be sufficient to learn $C$ using \lemref{UConjPauli}.  But from \lemref{UConjGenerator}, we only need to know $C \sigma_g C^\dagger$ for each $\sigma_g \in G$.

Let $C \sigma_{x_i} C^\dagger = \alpha_i \sigma_{a_i}$ and $C \sigma_{z_i} C^\dagger = \beta_i \sigma_{b_i}$, where $\alpha_i, \beta_i = \pm 1$.  Knowing just $\sigma_{a_i}$ and $\sigma_{b_i}$ is enough to specify $C$ up to a Pauli correction factor $\sigma_q$ which gives the phases $\alpha_i$ and $\beta_i$.  Choosing $\sigma_q$ that anticommutes with $\sigma_{x_i}$ flips the sign of $\alpha_i$ and similarly for $\sigma_{z_i}$.  We now present the algorithm:
\begin{enumerate}
\item{Apply $C \sigma_{x_i} C^\dagger$ and $C \sigma_{z_i} C^\dagger$ for each $i$ and use \thmref{PauliLearning} to determine $\sigma_{a_i}$ and $\sigma_{b_i}$.  This uses $2n$ queries to both $C$ and $C^\dagger$.}
\item{Let $C'$ be such that $C' \sigma_{x_i} C'^\dagger = \sigma_{a_i}$ and $C' \sigma_{z_i} C'^\dagger = \sigma_{b_i}$ i.e.~the phases are all $+1$.  Then, choosing a phase for $C'$, we can write $C = C' \sigma$ where
\be
\sigma = \prod_{i : \alpha_i = -1} \sigma_{z_i} \prod_{i : \beta_i = -1} \sigma_{x_i}.
\ee
Then implement $C'^\dagger C$ to determine $\sigma$ using \thmref{PauliLearning}.  This uses one query to $C$.  We can now calculate the phases $\alpha_i$ and $\beta_i$.}
\end{enumerate}
To work out the time complexity, note that in step 1 the $O(n)$ time Pauli learning algorithm is called $2n$ times.  Then for step 2, the Clifford $C'$ can be implemented in $O(n^2)$ time using for example Theorem 10.6 of \cite{NielsenChuang}.
\end{proof}
We now show that this algorithm is optimal, in terms of number of queries, up to constant factors:
\begin{lemma}
\label{lem:CliffordLearningOptimal}
Any method of learning a Clifford gate requires at least $n$ queries.
\end{lemma}
\begin{proof}
Each application of the gate $C$ can give at most $2n$ bits of mutual information about $C$.  This follows from the optimality of superdense coding \cite{SuperdenseCoding}.  The Clifford group (modulo global phase) is of size \cite{Calderbank98quantumerror} $2^{n^2+2n+3} \prod_{j=1}^n (4^j-1) \ge 2^{2n^2+n+3}$.  To identify an element with $m$ queries, we therefore need
\be
2^{2nm} \ge 2^{2n^2+n+3}
\ee
which implies $m \ge n$.
\end{proof}
It is unfortunate that access to $C^\dagger$ is also required, but we do not know a method with optimal query complexity that works without $C^\dagger$.  There are however methods that use $O(n^2)$ queries that do not use $C^\dagger$.  The result of \cite{HowManyCopiesStateDesc} can be used to show that $O(n^2)$ queries to $C$ are sufficient, by distinguishing the states $C \ot I \ket{\psi}$ for different Cliffords $C$ and where $\ket{\psi}$ is the maximally entangled state.  We can use \lemref{UniqueCloseC} to show that these states are far apart in the distance measure used in \cite{HowManyCopiesStateDesc}, allowing us to apply their result.

\subsection{Learning Gottesman-Chuang Operations}

\thmref{CliffordLearning} can easily be generalised to learning any operation from the $\cC_k$ hierarchy:
\begin{theorem}
\label{thm:CkLearning}
Given oracle access to an unknown operation $C \in \cC_k$ and its conjugate $C^\dagger$, $C$ can be determined exactly (up to phase) with $\frac{(2n)^k - 1}{2n-1}$ queries to $C$ and $(2n)^{k-1}$ to $C^\dagger$.
\end{theorem}
\begin{proof}
The proof is by induction.  The base case is for the Paulis and is proven in \thmref{PauliLearning}.  Then, to learn $C \in \cC_{k+1}$, we assume we have a learning algorithm for members of $\cC_{k}$.  Apply $C \sigma_g C^\dagger$ for each $\sigma_g \in G$.  These operations are elements of $\cC_{k}$ so use the learning algorithm for $\cC_{k}$ to determine these up to phase.  Then use the last step of \thmref{CliffordLearning} to determine the phases.

We now determine the number of queries to $C$ and $C^\dagger$.  Let $T(k)$ be the number of queries to $C$ and $T'(k)$ the number of queries to $C^\dagger$.  We have the recurrences
\ba
T(k+1) &= 2n T(k) + 1 \nonumber \\
T(1) &= 1
\ea
and
\ba
T'(k+1) &= 2n T'(k) \nonumber \\
T'(2) &= 2n
\ea
which have solutions $T(k) = \frac{(2n)^k - 1}{2n-1}$ and $T'(k) = (2n)^{k-1}$ (with $T'(1) = 0$).
\end{proof}

\section{Learning Unitaries Close to \texorpdfstring{$\cC_k$}{C\_k} Elements}
\label{sec:ApproxLearning}

Here we suppose that we are given a unitary that is known to be close to an element of $\cC_k$ for some given $k$.  We present a method for finding this element.  But first we must define our distance measure.

We would like our distance measure to not distinguish between unitaries that differ by just an unobservable global phase.  We define a `distance' $D$ below with this property.  However, firstly define the distance $D^+$ to be a normalised 2-norm distance:
\begin{definition}
For $U_1$ and $U_2$ $d \times d$ matrices,
\bes
D^+(U_1, U_2) := \frac{1}{\sqrt{2d}} || U_1 - U_2 ||_2.
\ees
where $||A||_2 = \sqrt{\tr A^\dagger A}$.
\end{definition}
We have chosen the normalisation so that $0 \le D^+(U_1, U_2) \le 1$.  We now define our phase invariant `distance':
\begin{definition}
\label{def:DistanceMeasure}
For $U_1$ and $U_2$ $d \times d$ matrices,
\bes
D(U_1, U_2) := \frac{1}{\sqrt{2d^2}} || U_1 \ot U_1^* - U_2 \ot U_2^* ||_2
\ees
\end{definition}
This is not a true distance since $D(U_1, U_2) = 0$ does not imply $U_1 = U_2$, but that $U_1$ and $U_2$ are the same up to a phase so the difference is unobservable.  From the 2-norm definition, we can show:
\begin{lemma}
\label{lem:DistFromInnerProd}
\be
D^+(U_1, U_2) = \sqrt{1 - \Re \frac{\tr U_1 U_2^\dagger}{d}}
\ee
and
\be
D(U_1, U_2) = \sqrt{1 - \l| \frac{\tr U_1 U_2^\dagger}{d} \r|^2}.
\ee
\end{lemma}
From this we can easily see that $0 \le D(U_1, U_2) \le 1$ with equality if and only if $U_1$ and $U_2$ are orthogonal.  Further note that by the unitary invariance of the 2-norm, both $D$ and $D^+$ are unitarily invariant and from the triangle inequality for the 2-norm they both obey the triangle inequality.

Our approximate learning method will find the unique closest element of $\cC_k$ to $U$.  In order to guarantee uniqueness, the distance must be upper bounded:
\begin{lemma}
\label{lem:UniqueCloseC}
If $D(U, C) < \frac{1}{2^{k-1/2}}$ for some $C \in \cC_k$ then $C$ is unique up to phase.
\end{lemma}
The proof is in \secref{ProofLemUniqueCloseC}.
\begin{theorem}
\label{thm:ApproxLearning}
Given oracle access to $U$ and $U^\dagger$ and $k$ such that $D(U, C) \le \eps$ for some $C \in \cC_k$ with
\be
\eps' := \sqrt{2(1-(2^{k-1} \eps)^2)} - 1 > 0
\ee
then $C$ can be determined with probability at least $1-\delta$ with \bes O\l(\frac{1}{\eps'^2} (2n)^{k-1} \log \frac{(2n+1)^{k-1}}{\delta}\r) \ees queries.
\end{theorem}
\begin{proof}
By \lemref{UniqueCloseC}, $C$ is unique up to phase.  We now prove the Theorem by induction.

For $k=1$, use Proposition 21 of \cite{QBF} to learn the closest Pauli operator.  This works by repeating the Pauli learning method \thmref{PauliLearning} and taking the majority vote.  This uses $O\l(\frac{1}{\eps'^2} \log \frac{1}{\delta}\r)$ queries to succeed with probability at least $1-\delta$.

Now for the inductive step.  Assume we have a learning algorithm for level $k$.  Then for $C \in \cC_{k+1}$, let $C \sigma_{g_i} C^\dagger = C_{g_i}$ for $\sigma_{g_i} \in G$.  By \lemref{CloseHaveClosePauli}, we have $D(U \sigma_{g_i} U^\dagger, C_{g_i}) \le 2\eps$.  Use the learning algorithm for level $k$ to determine $C_{g_i}$ up to phase for all $i$.  Then to find the phases we use the same method as before: implement any $C'$ with $C' \sigma_{g_i} C'^\dagger = \pm C_{g_i}$ for any (known) choice of phase.  Then $C'= C \sigma_q$ for some Pauli operator $\sigma_q$.  We can determine $\sigma_q$ by implementing $C'^\dagger U$ and using the $k=1$ learning algorithm since
\ba
D(C'^\dagger U, \sigma_q) &= D(U, C' \sigma_q) \nonumber\\
&= D(U, C) \le \eps.
\ea
Now we calculate the success probabilities and number of queries.  There are $2n+1$ calls to the algorithm at lower levels, which all succeed with probability at least $1-\delta$.  So at this level the success probability is at least $1-(2n+1)\delta$.  So to succeed with probability at least $1-\delta$ we must replace $\delta$ with $\delta/(2n+1)$.  Then the overall number of queries is
\be
2n \cdot O\l(\frac{1}{\eps'^2} (2n)^{k-1} \log \frac{(2n+1)^{k}}{\delta}\r) + 1 = O\l(\frac{1}{\eps'^2} (2n)^{k} \log \frac{(2n+1)^{k}}{\delta}\r).\qedhere
\ee
\end{proof}
We remark that there is only $O(k \log n)$ overhead (for constant $\eps'$ and $\delta$) over the exact learning algorithm of \thmref{CkLearning}.

\section{Clifford Testing}
\label{sec:CliffordTesting}

Here we present an efficient algorithm to determine whether an unknown unitary operation is close to a Clifford or far from every Clifford.  Whereas the previous results allow us to find the Clifford operator close to the given black box unitary, in this section we are concerned with determining how far the given unitary is from any Clifford.  We do not measure this directly, but provide an algorithm of low query complexity that decides if the given unitary is close to a Clifford or far from all.  This type of algorithm is known in computer science as a property testing algorithm and has many applications, including the theory of probabilistically checkable proofs \cite{PCPTheorem}.  The result in this section could be extended to work for any level of the Gottesman-Chuang hierarchy although for simplicity we only present the version for Cliffords.

The key ingredient to our method will be a way of estimating the Pauli coefficients:
\begin{lemma}[Lemma 23 of \cite{QBF}]
\label{lem:MeasurePauliCoeff}
For any $p \in \{I, x, y, z\}^n$ and unitary $U$, the Pauli coefficients $|\gamma(p)| = \frac{1}{2^n} \l| \tr U \sigma_p \r|$ can be estimated to within $\pm \eta$ with probability $1-\delta$ using $O\l(\frac{1}{\eta^2} \log \frac{1}{\delta}\r)$ queries.
\end{lemma}
This is a generalisation of \thmref{PauliLearning} and the method is similar.  Instead of there being only one possible outcome, now the probability of obtaining the outcome corresponding to $\sigma_p$ is estimated.  This probability is equal to $|\gamma(p)|^2$.
\begin{theorem}
\label{thm:CliffordTesting}
Given oracle access to $U$ and $U^\dagger$ with the promise that for $0 < \eps < 1$ either
\begin{enumerate}
\item[a)]{CLOSE: there exists $C \in \cC$ such that $D(U, C) \le \frac{\eps}{\sqrt{32}n}$ or}
\item[b)]{FAR: for all $C \in \cC$, $D(U,C) > \eps$ and there exists $C \in \cC$ such that $D(U, C) \le 1/3$}
\end{enumerate}
holds then there is a $O\l(\frac{n^3}{\eps^2} \log \frac{n}{\delta}\r)$ algorithm that determines which with probability at least $1-\delta$.
\end{theorem}
\begin{proof}
In both cases, we have that $D(U, C) < 1/3$ for some $C$, which ensures that $C$ is unique (using \lemref{UniqueCloseC}, since $\frac{1}{3} < \frac{1}{2 \sqrt{2}}$) and can be found using \thmref{ApproxLearning} with $O\l(n \log \frac{n}{\delta} \r)$ queries.  Then the algorithm is:
\begin{enumerate}
\item{For each $\sigma_g \in G$, measure the Pauli coefficient of $C \sigma_g C^\dagger$ in $U \sigma_g U^\dagger$ (i.e.~measure $\l| \tr U \sigma_g U^\dagger C \sigma_g C^\dagger \r|/2^n$) to precision $\frac{\eps^2}{16 n^2}$ using \lemref{MeasurePauliCoeff}.}
\item{If all the coefficients are found to have modulus at least $1 - \frac{3 \eps^2}{16 n^2}$ then output \emph{CLOSE} else output \emph{FAR}.}
\end{enumerate}
This works because, for the two possibilities \emph{CLOSE} and \emph{FAR}:
\begin{enumerate}
\item[a)]{Using \lemref{CloseHaveClosePauli}, $D(U,C) \le \frac{\eps}{\sqrt{32} n}$ implies that for all $\sigma_p \in \hat{\cP},$
\be D(U \sigma_p U^\dagger, C \sigma_p C^\dagger) \le \frac{2\eps}{\sqrt{32}n}. \ee
Since we will only apply $U \sigma_g U^\dagger$ for $\sigma_g \in G$ we restrict this to only the generators to find that for all $\sigma_g \in G,$
\be D(U \sigma_g U^\dagger, C \sigma_g C^\dagger) \le \frac{2\eps}{\sqrt{32}n} \ee
giving
\be \l|\frac{\tr U \sigma_g U^\dagger C \sigma_g C^\dagger}{2^n} \r|^2 \ge 1 - \frac{\eps^2}{8 n^2}
\ee
for every generator $\sigma_g$.  We need a bound on the non-squared coefficients, which follows directly:
\be \l|\frac{\tr U \sigma_g U^\dagger C \sigma_g C^\dagger}{2^n} \r| \ge 1 - \frac{\eps^2}{8 n^2}.
\ee
Therefore when measuring the coefficients to precision $\frac{\eps^2}{16 n^2}$, all results will give at least $1 - \frac{3\eps^2}{16 n^2}$.}
\item[b)]{
Using the contrapositive of \lemref{ClosePauliAreClose}, $D(U, C) > \eps$ implies that there exists $\sigma_p \in \hat{\cP}$ such that
\be D^+(U \sigma_p U^\dagger, C \sigma_p C^\dagger) > \eps. \ee
Using the contrapositive of \lemref{CloseGenAreClosePauli} this in turn implies there exists $\sigma_g \in G$ such that
\be D^+(U \sigma_g U^\dagger, C \sigma_g C^\dagger) > \frac{\eps}{2n}, \ee
which means that for at least one $\sigma_g \in G$, $U \sigma_g U^\dagger$ will have a small overlap with $C \sigma_g C^\dagger$ i.e.~there exists $\sigma_g \in G$ such that
\be
\l| \frac{\tr U \sigma_g U^\dagger C \sigma_g C^\dagger}{2^n} \r| < 1- \frac{\eps^2}{4n^2}.
\ee
The $C$ returned by the application of \thmref{ApproxLearning} is such that $\tr U \sigma_g U^\dagger C \sigma_g C^\dagger$ is positive, which justifies inserting the absolute value signs above when using $D^+$ rather than $D$.  This implies that at least one coefficient will be found to be less than $1- \frac{3\eps^2}{16 n^2}$ when measuring to precision $\frac{\eps^2}{16n^2}$.\qedhere
}
\end{enumerate}
\end{proof}

\section{Conclusions and Further Work}
\label{sec:ConclusionLearning}

We have shown how to exactly identify an unknown Clifford operator in $O(n)$ queries, which we show is optimal.  This is then extended to cover elements of the $\cC_k$ hierarchy and for unitaries that are only known to be close to $\cC_k$ operations.  The key to the Clifford learning algorithm is to apply $C \sigma_p C^\dagger$ and then find the resulting Pauli operator.

A way of extending this idea could be to learn unitaries from larger sets.  Suppose $\cV$ is a set of unitaries with the property that for every $V \in \cV$, $V \sigma_p V^\dagger$ is a linear combination of a constant number of Paulis.  Then $V$ can be learnt in the same way as above, using the quantum Goldreich-Levin algorithm of \cite{QBF}, which can efficiently find which Paulis have large overlap with an input unitary.  However, we have not been able to find interesting sets $\cV$ other than the Clifford group with this property.

We also presented a Clifford testing algorithm, which determines whether a given black-box unitary is close to a Clifford or far from every Clifford.  This can be seen as a quantum generalisation of quadratic testing, just as Pauli testing can be seen as a quantum generalisation of linearity testing.  Property testing of this form is used to prove the PCP theorem \cite{PCPTheorem} so these quantum testing results could potentially be useful in proving a quantum PCP theorem.  It would also be interesting to strengthen the testing method in \thmref{CliffordTesting} to remove the $O(1/n)$ difference between the close and far conditions.

Finally, it would be interesting to see if it is possible to remove the requirement to have access to $U^\dagger$.  However, using both $U$ and $U^\dagger$ is the key to our method so we do not know if a method without $U^\dagger$ is possible with low query complexity.

\section{Proof of Lemma 6.4.4} %\lemref{UniqueCloseC}}
\label{sec:ProofLemUniqueCloseC}

\begin{proof}[Proof of \lemref{UniqueCloseC}]
The proof is by induction.  The base case is for $k=1$ when we have the Pauli group.  Without loss of generality, assume $C$ is a Pauli operator with no phase.  Let $C = \sigma_p$.

Expand $U$ in the Pauli basis:
\be
U = \sum_q \gamma(q) \sigma_q.
\ee
Since $U$ is unitary, we have $\sum_q | \gamma(q) |^2 = 1$.  By \lemref{DistFromInnerProd},
\be
D(U, \sigma_p)^2 = 1 - \l| \frac{\tr \sigma_p U}{2^n} \r|^2
\ee
which implies
\be
|\gamma(p)|^2 \ge 1- \eps^2.
\ee

Now, suppose for contradiction that there exists $\sigma_{p_1} \ne \sigma_{p_2}$ with $D(U, \sigma_{p_1}) \le \eps$ and $D(U, \sigma_{p_2}) \le \eps$.  Then by the above, $|\gamma(p_1)|^2, |\gamma(p_2)|^2 \ge 1- \eps^2$.  But there is also the constraint $|\gamma(p_1)|^2 + |\gamma(p_2)|^2 \le 1$ which combined give
\be
\eps \ge \frac{1}{\sqrt{2}}
\ee
which is false by assumption.  This implies $\sigma_{p_1} = \sigma_{p_2}$, which proves the base case.

To prove the inductive step, again assume for contradiction that there exist $C_1, C_2 \in \cC_{k+1}$ with $C_1 \ne C_2$ and $D(U, C_1) \le \eps$ and $D(U, C_2) \le \eps$.  Then there exists $\sigma_g \in G$ with
\be
C_1 \sigma_g C_1^\dagger =: C_{1g} \ne C_{2g} := C_2 \sigma_g C_2^\dagger.
\ee
Here, $C_{1g}, C_{2g} \in \cC_k$.

Using \lemref{CloseHaveClosePauli}, $D(U \sigma_g U^\dagger, C_{1g}) \le 2\eps$ and $D(U \sigma_g U^\dagger, C_{2g}) \le 2\eps$.

Now there are two cases.  Firstly, suppose we can choose $\sigma_g$ such that $C_1 \sigma_g C_1^\dagger \ne \pm C_2 \sigma_g C_2^\dagger$.  Then $C_{1g}$ and $C_{2g}$ are not equivalent up to phase so, using the inductive hypothesis, we must have
\be
2 \eps \ge \frac{1}{2^{k-1/2}}
\ee
or
\be
\eps \ge \frac{1}{2^{(k+1)-1/2}}
\ee
which is again false by assumption.

For the other case, $C_1 \sigma_g C_1^\dagger = \pm C_2 \sigma_g C_2^\dagger$ for all $\sigma_g \in G$.  This implies that $C_2 = C_1 \sigma_q$ for some Pauli $\sigma_q \ne I$.  Then we have 
\ba
D(U, C_1) &\le \eps\nonumber\\
D(U, C_1 \sigma_q) &\le \eps
\ea
which by unitary invariance gives
\ba
D(C_1^\dagger U, I) &\le \eps\nonumber\\
D(C_1^\dagger U, \sigma_q) &\le \eps.
\ea
But we proved that this is impossible in this range of $\eps$ in the $k=1$ proof above.
\end{proof}

\section{Miscellaneous Lemmas}

Here we prove some miscellaneous lemmas used earlier in the chapter.

The first lemma says that for two close operators $U_1$ and $U_2$, $U_1 \sigma_p U_1^\dagger$ is close to $U_2 \sigma_p U_2^\dagger$ for all Paulis $\sigma_p$:
\begin{lemma}
\label{lem:CloseHaveClosePauli}
If $D(U_1, U_2) \le \delta$ then for all $\sigma_p \in \hat{\cP}$,
\bes
D(U_1 \sigma_p U_1^\dagger, U_2 \sigma_p U_2^\dagger) \le 2 \delta.
\ees
\end{lemma}
\begin{proof}
Let $U_1 = V U_2$ and $U_{2p} = U_2 \sigma_p U_2^\dagger$.  Then we simply apply the triangle inequality for $D$ and unitary invariance:
\bas
D(U_1 \sigma_p U_1^\dagger, U_2 \sigma_p U_2^\dagger) &= D(V U_{2p} V^\dagger, U_{2p}) \\
&= D(V U_{2p}, U_{2p} V) \\
&\le D(V U_{2p}, U_{2p}) + D(U_{2p}, U_{2p} V) \\
&= D(V, I) + D(I, V) \\
&= 2D(U_1, U_2).\qedhere
\eas
\end{proof}
The next lemma is a converse to this:
\begin{lemma}
\label{lem:ClosePauliAreClose}
If for all $\sigma_p \in \hat{\cP}$
\be D^+(U_1 \sigma_p U_1^\dagger, U_2 \sigma_p U_2^\dagger) \le \delta \ee
then
\be
D(U_1, U_2) \le \delta.
\ee
\end{lemma}
\begin{proof}
If $D^+(U_1 \sigma_p U_1^\dagger, U_2 \sigma_p U_2^\dagger) \le \delta$ then $\frac{1}{2^n} \Re \tr U_1 \sigma_p U_1^\dagger U_2 \sigma_p U_2^\dagger \ge 1-\delta^2$.  Since this is true for all $\sigma_p$, we can take the average of this over the whole of $\hat{\cP}$ and use the fact that for any $d \times d$ matrix $A$ $\frac{1}{4^n}\sum_{\sigma_p \in \hat{\cP}} \sigma_p A \sigma_p = \frac{I}{2^n} \tr A$ (the Paulis are a \emph{1-design}) to find
\be
\frac{1}{2^n} \Re \tr U_1 \l( \frac{I}{2^n} \tr U_1^\dagger U_2 \r) U_2^\dagger \ge 1 - \delta^2
\ee
which simplified gives
\be
\l| \frac{\tr U_1 U_2^\dagger}{2^n} \r|^2 \ge 1 - \delta^2
\ee
giving the desired result.
\end{proof}
Now we show how to go from distances for just the generators $G$ to distances for the whole of $\hat{\cP}$:
\begin{lemma}
\label{lem:CloseGenAreClosePauli}
If for all $\sigma_g \in G$
\be
D^+(U_1 \sigma_g U_1^\dagger, U_2 \sigma_g U_2^\dagger) \le \delta
\ee
then for all $\sigma_p \in \hat{\cP}$
\be
D^+(U_1 \sigma_p U_1^\dagger, U_2 \sigma_p U_2^\dagger) \le 2n\delta
\ee
\end{lemma}
\begin{proof}
The proof is by induction on the number of generators required to make $\sigma_p$, using the triangle inequality for $D^+$.
\end{proof}

\newpage

\addcontentsline{toc}{chapter}{Bibliography}

\newcommand{\etalchar}[1]{$^{#1}$}


\begin{thebibliography}{{Low}09b}

\bibitem[{Aar}07]{AaronsonLearnability}
S.~{Aaronson}.
\newblock {The learnability of quantum states}.
\newblock {\em Proc. R. Soc. A}, 463:3089--3114, 2007.
\newblock arXiv:quant-ph/0608142.

\bibitem[{Aar}09]{Aaronson07}
S.~{Aaronson}.
\newblock {Quantum Copy-Protection and Quantum Money}.
\newblock {\em IEEE Conference on Computational Complexity 2009}, 2009.

\bibitem[ABI86]{ExactkWiseIndep}
N.~{Alon}, L.~{Babai}, and A.~{Itai}.
\newblock {A fast and simple randomized parallel algorithm for the maximal
  independent set problem}.
\newblock {\em J. Algorithms}, 7(4):567--583, 1986.

\bibitem[ABW09]{TamperResistance}
A.~{Ambainis}, J.~{Bouda}, and A.~{Winter}.
\newblock {Nonmalleable encryption of quantum information}.
\newblock {\em J. Math. Phys.}, 50(4):042106, 2009.
\newblock arXiv:0808.0353.

\bibitem[ADHW06]{ADHW06}
A.~{Abeyesinghe}, I.~{Devetak}, P.~{Hayden}, and A.~{Winter}.
\newblock {The mother of all protocols: Restructuring quantum information's
  family tree}, 2006.
\newblock arXiv:quant-ph/0606225.

\bibitem[AE07]{AmbainisEmerson07}
A.~{Ambainis} and J.~{Emerson}.
\newblock {Quantum t-designs: t-wise independence in the quantum world}.
\newblock {\em IEEE Conference on Computational Complexity 2007}, pages
  129--140, 2007.
\newblock arXiv:quant-ph/0701126v2.

\bibitem[AG09]{AaronsonGottesmanStabilisers}
S.~{Aaronson} and D.~{Gottesman}, 2009.
\newblock {Unpublished}.

\bibitem[AK62]{ArnoldKrylov62}
V.~I. {Arnold} and A.~L. {Krylov}.
\newblock {Uniform distribution of points on a sphere and some ergodic
  properties of solutions of linear ordinary differential equations in a
  complex domain}.
\newblock {\em Soviet Math. Dokl.}, 4(1), 1962.

\bibitem[ALM{\etalchar{+}}98]{PCPTheorem}
S.~{Arora}, C.~{Lund}, R.~{Motwani}, M.~{Sudan}, and M.~{Szegedy}.
\newblock {Proof verification and the hardness of approximation problems}.
\newblock {\em J. ACM}, 45(3):501--555, 1998.

\bibitem[AMTd00]{AMTW00}
A.~{Ambainis}, M.~{Mosca}, A.~{Tapp}, and R.~{de Wolf}.
\newblock {Private Quantum Channels}.
\newblock {\em 41st Annual IEEE Symposium on Foundations of Computer Science},
  pages 547--553, 2000.

\bibitem[AN08]{AlonNussboim08}
N.~{Alon} and A.~{Nussboim}.
\newblock {k-Wise Independent Random Graphs}.
\newblock {\em 49th Annual IEEE Symposium on Foundations of Computer Science},
  0:813--822, 2008.
\newblock arXiv:0804.1268.

\bibitem[AS04]{AmbainisSmith04}
A.~{Ambainis} and A.~{Smith}.
\newblock {Small Pseudo-Random Families of Matrices: Derandomizing Approximate
  Quantum Encryption}.
\newblock {\em Proceedings of RANDOM 2004, LNCS}, 3122:249--260, 2004.
\newblock arXiv:quant-ph/0404075.

\bibitem[{Aub}09]{AubrunRandomizingChannels}
G.~{Aubrun}.
\newblock {On Almost Randomizing Channels with a Short Kraus Decomposition}.
\newblock {\em Comm. Math. Phys.}, 288:1103--1116, 2009.
\newblock arXiv:0805.2900.

\bibitem[{Bar}02]{Barnum02}
H.~{Barnum}.
\newblock {Information-disturbance tradeoff in quantum measurement on the
  uniform ensemble and on the mutually unbiased bases}, 2002.
\newblock arXiv:quant-ph/0205155.

\bibitem[BBD{\etalchar{+}}97]{BBDEJM97}
A.~{Barenco}, A.~{Berthiaume}, D.~{Deutsch}, A.~{Ekert}, R.~{Jozsa}, and
  C.~{Macchiavello}.
\newblock {Stabilization of Quantum Computations by Symmetrization}.
\newblock {\em SIAM J. Comput.}, 26(5):1541--1557, 1997.
\newblock arXiv:quant-ph/9604028.

\bibitem[BFP{\etalchar{+}}72]{LinearTimeMedian}
M.~{Blum}, R.~W. {Floyd}, V.~{Pratt}, R.~L. {Rivest}, and R.~E. {Tarjan}.
\newblock {Linear time bounds for median computations}.
\newblock {\em 4th Annual ACM Symposium on Theory of Computing}, pages
  119--124, 1972.

\bibitem[BG06]{BG06}
J.~Bourgain and A.~Gamburd.
\newblock {New results on expanders}.
\newblock {\em C. R. Acad. Sci. Paris, Ser. I}, 342:717--721, 2006.

\bibitem[BH08]{HooryBrodsky04}
A.~{Brodsky} and S.~{Hoory}.
\newblock {Simple Permutations Mix Even Better}.
\newblock {\em {Random Struct. Algorithms}}, 32(3):274--289, 2008.
\newblock arXiv:math/0411098.

\bibitem[BHL{\etalchar{+}}05]{RemoteStatePreparation05}
C.H. Bennett, P.~Hayden, D.W. Leung, P.W. Shor, and A.~Winter.
\newblock {Remote preparation of quantum states}.
\newblock {\em IEEE Trans. Inform. Theory}, 51(1):56--74, 2005.

\bibitem[BL01]{BarnumLinden01}
H.~{Barnum} and N.~{Linden}.
\newblock {Monotones and invariants for multi-particle quantum states}.
\newblock {\em J. Phys. A}, 34(35):6787--6805, 2001.
\newblock arXiv:quant-ph/0103155.

\bibitem[BMW09]{MostStatesUselessMBQCBMW}
M.~J. {Bremner}, C.~{Mora}, and A.~{Winter}.
\newblock {Are Random Pure States Useful for Quantum Computation?}
\newblock {\em Phys. Rev. Lett.}, 102(19):190502, 2009.
\newblock arXiv:0812.3001.

\bibitem[BR94]{BellareRompel}
M.~{Bellare} and J.~{Rompel}.
\newblock {Randomness-efficient oblivious sampling}.
\newblock {\em 35th Annual IEEE Symposium on Foundations of Computer Science},
  pages 276--287, 1994.

\bibitem[BT07]{QExpandersEntropyDifference}
A.~{Ben-Aroya} and A.~{Ta-Shma}.
\newblock {Quantum expanders and the quantum entropy difference problem}, 2007.
\newblock arXiv:quant-ph/0702129.

\bibitem[BW92]{SuperdenseCoding}
C.~H. {Bennett} and S.~J. {Wiesner}.
\newblock {Communication via one- and two-particle operators on
  Einstein-Podolsky-Rosen states}.
\newblock {\em Phys. Rev. Lett.}, 69(20):2881--2884, 1992.

\bibitem[CN97]{ChuangNielsen97}
I.~L. {Chuang} and M.~A. {Nielsen}.
\newblock {Prescription for experimental determination of the dynamics of a
  quantum black box}.
\newblock {\em Journal of Modern Optics}, 44:2455--2467, 1997.
\newblock arXiv:quant-ph/9610001.

\bibitem[CRSS97]{QECGeometry}
A.~R. Calderbank, E.~M. Rains, P.~W. Shor, and N.~J.~A. Sloane.
\newblock {Quantum Error Correction and Orthogonal Geometry}.
\newblock {\em Phys. Rev. Lett.}, 78(3):405--408, 1997.

\bibitem[CRSS98]{Calderbank98quantumerror}
A.~R. Calderbank, E.~M. Rains, P.~W. Shor, and N.~J.~A. Sloane.
\newblock {Quantum Error Correction Via Codes Over GF(4)}.
\newblock {\em IEEE Trans. Inform. Theory}, 44:1369--1387, 1998.

\bibitem[{Dan}05]{Dankert05}
C.~{Dankert}.
\newblock {Efficient Simulation of Random Quantum States and Operators}.
\newblock {\em MMath Thesis, University of Waterloo}, 2005.
\newblock arXiv:quant-ph/0512217.

\bibitem[DCEL06]{DCEL06}
C.~{Dankert}, R.~{Cleve}, J.~{Emerson}, and E.~{Livine}.
\newblock {Exact and Approximate Unitary 2-Designs: Constructions and
  Applications}, 2006.
\newblock arXiv:quant-ph/0606161.

\bibitem[DHL{\etalchar{+}}04]{Locking04}
D.~P. {DiVincenzo}, M.~{Horodecki}, D.~W. {Leung}, J.~A. {Smolin}, and B.~M.
  {Terhal}.
\newblock {Locking Classical Correlations in Quantum States}.
\newblock {\em Phys. Rev. Lett.}, 92(6):067902, 2004.
\newblock arXiv:quant-ph/0303088.

\bibitem[DLP01]{DArianoPresti01}
G.~M. D'Ariano and P.~Lo~Presti.
\newblock {Quantum Tomography for Measuring Experimentally the Matrix Elements
  of an Arbitrary Quantum Operation}.
\newblock {\em Phys. Rev. Lett.}, 86(19):4195--4198, 2001.
\newblock arXiv:quant-ph/0012071.

\bibitem[DLT02]{DLT02}
D.~{DiVincenzo}, D.~{Leung}, and B.~{Terhal}.
\newblock {Quantum Data Hiding}.
\newblock {\em IEEE Trans. Inform. Theory}, 48(3):580--598, 2002.
\newblock arXiv:quant-ph/0103098.

\bibitem[DN06]{DickinsonNayak06}
P.~A. {Dickinson} and A.~{Nayak}.
\newblock {Approximate Randomization of Quantum States With Fewer Bits of Key}.
\newblock {\em Quantum Computing: Back Action 2006}, 864:18--36, 2006.

\bibitem[DOP07]{DOP07}
O.~C.~O. {Dahlsten}, R.~{Oliveira}, and M.~B. {Plenio}.
\newblock {The emergence of typical entanglement in two-party random
  processes}.
\newblock {\em Journal of Physics A Mathematical General}, 40:8081--8108, 2007.
\newblock arXiv:quant-ph/0701125.

\bibitem[DP06]{DahlstenPlenio05}
O.~{Dahlsten} and M.~{Plenio}.
\newblock {Entanglement probability distribution of bipartite randomised
  stabilizer states}.
\newblock {\em Q. Info. Comp.}, 6(6):527--538, 2006.
\newblock arXiv:quant-ph/0511119.

\bibitem[DS93]{DiaconisSaloff-Coste93}
P.~{Diaconis} and L.~{Saloff-Coste}.
\newblock {Comparison Theorems for Reversible Markov Chains}.
\newblock {\em Ann. Appl. Probab.}, 3(3):696--730, 1993.

\bibitem[DS96]{DiaconisSaloff-Coste96}
P.~{Diaconis} and L.~{Saloff-Coste}.
\newblock {Logarithmic Sobolev inequalities for finite Markov chains}.
\newblock {\em Ann. Appl. Probab.}, 6(3):695--750, 1996.

\bibitem[ELL05]{ELL05}
J.~{Emerson}, E.~{Livine}, and S.~{Lloyd}.
\newblock {Convergence conditions for random quantum circuits}.
\newblock {\em Phys. Rev. A}, 72(060302), 2005.
\newblock arXiv:quant-ph/0503210.

\bibitem[FK94]{PagesConjectureProof94}
S.~K. {Foong} and S.~{Kanno}.
\newblock {Proof of Page's conjecture on the average entropy of a subsystem}.
\newblock {\em Phys. Rev. Lett.}, 72:1148--1151, 1994.

\bibitem[GAE07]{GAE07}
D.~{Gross}, K.~{Audenaert}, and J.~{Eisert}.
\newblock {Evenly distributed unitaries: On the structure of unitary designs}.
\newblock {\em J. Math. Phys.}, 48(052104), 2007.
\newblock arXiv:quant-ph/0611002.

\bibitem[GC99]{GottesmanChuang}
D.~{Gottesman} and I.~L. {Chuang}.
\newblock {Demonstrating the viability of universal quantum computation using
  teleportation and single-qubit operations}.
\newblock {\em Nature}, 402:390--393, 1999.
\newblock arXiv:quant-ph/9908010.

\bibitem[GE08]{QuantumMargulisExpanders}
D.~{Gross} and J.~{Eisert}.
\newblock {Quantum Margulis Expanders}.
\newblock {\em Q. Info. Comp.}, 8(8/9):722--733, 2008.
\newblock arXiv:0710.0651.

\bibitem[GFE09]{MostStatesUselessMBQCGFE}
D.~{Gross}, S.~T. {Flammia}, and J.~{Eisert}.
\newblock {Most Quantum States Are Too Entangled To Be Useful As Computational
  Resources}.
\newblock {\em Phys. Rev. Lett.}, 102(19):190501, 2009.
\newblock arXiv:0810.4331.

\bibitem[{Gir}07]{GiraudPurityMoments}
O.~{Giraud}.
\newblock {Distribution of bipartite entanglement for random pure states}.
\newblock {\em J. Phys. A}, 40:2793--2801, 2007.
\newblock arXiv:quant-ph/0611285.

\bibitem[Got98]{GottesmanFaultTolerance}
D.~Gottesman.
\newblock {Theory of fault-tolerant quantum computation}.
\newblock {\em Phys. Rev. A}, 57(1):127--137, 1998.
\newblock arXiv:quant-ph/9702029.

\bibitem[Gro97]{GroversAlgorithm}
L.~Grover.
\newblock {Quantum Mechanics Helps in Searching for a Needle in a Haystack}.
\newblock {\em Phys. Rev. Lett.}, 79(2):325--328, 1997.

\bibitem[GW86]{GrimmettWelsh86}
G.~{Grimmett} and D.~{Welsh}.
\newblock {\em {Probability: An Introduction}}.
\newblock {Oxford University Press}, {Oxford, UK}, 1986.

\bibitem[GW98]{GoodmanWallach98}
R.~{Goodman} and N.~{Wallach}.
\newblock {\em {Representations and Invariants of the Classical Groups}}.
\newblock Cambridge University Press, Cambridge, UK, 1998.

\bibitem[{Har}08]{AramExpanders07}
A.~W. {Harrow}.
\newblock {Quantum expanders from any classical Cayley graph expander}.
\newblock {\em Q. Info. Comp.}, 8(8/9):715--721, 2008.
\newblock arXiv:0709.1142.

\bibitem[HH08]{HH08}
S.~{Hallgren} and A.~W. {Harrow}.
\newblock {Superpolynomial Speedups Based on Almost Any Quantum Circuit}.
\newblock In {\em ICALP '08: Proceedings of the 35th international colloquium
  on Automata, Languages and Programming, Part I}, pages 782--795, Berlin,
  Heidelberg, 2008. Springer-Verlag.
\newblock arXiv:0805.0007.

\bibitem[HH09]{HastingsHarrow08}
M.~B. {Hastings} and A.~W. {Harrow}.
\newblock {Classical and Quantum Tensor Product Expanders}.
\newblock {\em Q. Info. Comp.}, 9:336, 2009.
\newblock arXiv:0804.0011.

\bibitem[HHH06]{HHM06}
A.~{Hayashi}, T.~{Hashimoto}, and M.~{Horibe}.
\newblock {Reexamination of optimal quantum state estimation of pure states}.
\newblock {\em Phys. Rev. A}, 72(032325), 2006.
\newblock arXiv:quant-ph/0410207.

\bibitem[HHL04]{SuperdenseCodingHHL}
A.~{Harrow}, P.~{Hayden}, and D.~{Leung}.
\newblock {Superdense Coding of Quantum States}.
\newblock {\em Phys. Rev. Lett.}, 92(18):187901, 2004.
\newblock arXiv:quant-ph/0307221.

\bibitem[HHR{\etalchar{+}}05]{HaffnerTomography}
H.~H\"{a}ffner, W.~H\"{a}nsel, C.~F. Roos, J.~Benhelm, D.~Chek-Al-Kar,
  M.~Chwalla, T.~K\"{o}rber, U.~D. Rapol, M.~Riebe, P.~O. Schmidt, C.~Becher,
  O.~G\"{u}hne, W.~D\"{u}r, and R.~Blatt.
\newblock {Scalable multiparticle entanglement of trapped ions}.
\newblock {\em Nature}, 438(7068):643--646, 2005.
\newblock arXiv:quant-ph/0603217.

\bibitem[HHYW07]{HHYW07}
P.~{Hayden}, M.~{Horodecki}, J.~{Yard}, and A.~{Winter}.
\newblock {A decoupling approach to the quantum capacity}.
\newblock {\em Open Syst. Inf. Dyn.}, 15:7--19, 2007.
\newblock arXiv:quant-ph/0702005.

\bibitem[HL09a]{TPE}
A.~W. {Harrow} and R.~A. {Low}.
\newblock {Efficient Quantum Tensor Product Expanders and $k$-Designs}.
\newblock {\em Proceedings of RANDOM 2009, LNCS}, 5687:548--561, 2009.
\newblock arXiv:0811.2597.

\bibitem[HL09b]{RandomCircuits}
A.~W. {Harrow} and R.~A. {Low}.
\newblock {Random Quantum Circuits are Approximate 2-designs}.
\newblock {\em Comm. Math. Phys.}, 291(1):257--302, 2009.
\newblock arXiv:0802.1919.

\bibitem[HLSW04]{RandomizingQuantumStates04}
P.~{Hayden}, D.~{Leung}, P.~W. {Shor}, and A.~{Winter}.
\newblock {Randomizing Quantum States: Constructions and Applications}.
\newblock {\em Comm. Math. Phys.}, 250:371--391, 2004.
\newblock arXiv:quant-ph/0307104.

\bibitem[HLW06]{AspectsOfGenericEntanglement}
P.~{Hayden}, D.~W. {Leung}, and A.~{Winter}.
\newblock {Aspects of Generic Entanglement}.
\newblock {\em Comm. Math. Phys.}, 265:95--117, 2006.
\newblock arXiv:quant-ph/0407049.

\bibitem[HP07]{HaydenPreskill07}
P.~{Hayden} and J.~{Preskill}.
\newblock {Black holes as mirrors: quantum information in random subsystems}.
\newblock {\em Journal of High Energy Physics}, 09(120), 2007.
\newblock arXiv:0708.4025.

\bibitem[HW06]{HowManyCopiesStateDesc}
A.~W. {Harrow} and A.~{Winter}.
\newblock {How many copies are needed for state discrimination?}, 2006.
\newblock arXiv:quant-ph/0606131.

\bibitem[IR06]{IblisdirRoland06}
S.~{Iblisdir} and J.~{Roland}.
\newblock {Optimal finite measurements and Gauss quadratures}.
\newblock {\em Phys. Lett. A}, 358:368--372, 2006.
\newblock arXiv:quant-ph/0410237.

\bibitem[Kah96]{KahaleMixing}
N.~Kahale.
\newblock {A Semidefinite Bound for Mixing Rates of Markov Chains}.
\newblock {\em Proceedings of the 5th International IPCO Conference on Integer
  Programming and Combinatorial Optimization}, pages 190--203, 1996.

\bibitem[{Kas}05]{Kassabov05}
M.~{Kassabov}.
\newblock {Symmetric Groups and Expanders}, 2005.
\newblock arXiv:math/0503204.

\bibitem[KN96]{KushilevitzNisan}
E.~{Kushilevitz} and N.~{Nisan}.
\newblock {\em {Communication Complexity}}.
\newblock Cambridge University Press, Cambridge, UK, 1996.

\bibitem[KNR09]{KNRkWiseIndepPerms}
E.~{Kaplan}, M.~{Naor}, and O.~{Reingold}.
\newblock {Derandomized Constructions of k-Wise (Almost) Independent
  Permutations}.
\newblock {\em Algorithmica}, 55(1):113--133, 2009.

\bibitem[KSV02]{KSV02}
A.~Yu. {Kitaev}, A.~H. {Shen}, and M.~N. {Vyalyi}.
\newblock {\em {Classical and Quantum Computation}}.
\newblock American Mathematical Society, Boston, MA, USA, 2002.

\bibitem[Lan92]{LandauerInfoPhysical}
R.~Landauer.
\newblock {Information is Physical}.
\newblock {\em 1992 Workshop on Physics and Computation}, pages 1--4, 1992.

\bibitem[{Led}01]{Ledoux}
M.~{Ledoux}.
\newblock {\em {The Concentration of Measure Phenomenon}}.
\newblock {American Mathematical Society}, {Providence, RI, USA}, 2001.

\bibitem[Leu03]{Leung03}
D.~W. Leung.
\newblock {Choi's proof as a recipe for quantum process tomography}.
\newblock {\em J. Math. Phys.}, 44(2):528--533, 2003.
\newblock arXiv:quant-ph/0201119.

\bibitem[{Low}09a]{LargeDeviationskDesigns}
R.~A. {Low}.
\newblock {Large Deviation Bounds for k-designs}.
\newblock {\em Proc. R. Soc. A}, 465(2111):3289--3308, 2009.
\newblock arXiv:0903.5236.

\bibitem[{Low}09b]{LearningCliffords}
R.~A. {Low}.
\newblock {Learning and Testing Algorithms for the Clifford Group}.
\newblock {\em Phys. Rev. A}, 80(5):052314, 2009.
\newblock arXiv:0907.2833.

\bibitem[MO08]{QBF}
A.~{Montanaro} and T.~J. {Osborne}.
\newblock {Quantum boolean functions}, 2008.
\newblock arXiv:0810.2435.

\bibitem[MP95]{MassarPopsecu95}
S.~Massar and S.~Popescu.
\newblock {Optimal Extraction of Information from Finite Quantum Ensembles}.
\newblock {\em Phys. Rev. Lett.}, 74(8):1259--1263, 1995.

\bibitem[MR95]{MotwaniRaghavan}
R.~{Motwani} and P.~{Raghavan}.
\newblock {\em {Randomized algorithms}}.
\newblock Cambridge University Press, Cambridge, UK, 1995.

\bibitem[MR00]{MartinRandallDecomposition}
R.A. Martin and D.~Randall.
\newblock {Sampling adsorbing staircase walks using a new Markov chain
  decomposition method}.
\newblock {\em 41st Annual IEEE Symposium on Foundations of Computer Science},
  pages 492--502, 2000.

\bibitem[MT06]{MontenegroTetali06}
R.~{Montenegro} and P.~{Tetali}.
\newblock {Mathematical aspects of mixing times in Markov chains}.
\newblock {\em Found. Trends Theor. Comput. Sci.}, 1(3):237--354, 2006.

\bibitem[NC00]{NielsenChuang}
M.~A. {Nielsen} and I.~L. {Chuang}.
\newblock {\em {Quantum Computation and Quantum Information}}.
\newblock {Cambridge University Press}, Cambridge, UK, 2000.

\bibitem[{Nie}02]{Nielsen02}
M.~A. {Nielsen}.
\newblock {A simple formula for the average gate fidelity of a quantum
  dynamical operation}.
\newblock {\em Phys. Lett. A}, 303:249--252, 2002.
\newblock arXiv:quant-ph/0205035.

\bibitem[NN90]{NaorNaorkWiseIndep}
J.~Naor and M.~Naor.
\newblock {Small-bias probability spaces: efficient constructions and
  applications}.
\newblock {\em 22nd Annual ACM Symposium on Theory of Computing}, pages
  213--223, 1990.

\bibitem[ODP07]{ODP06}
R.~{Oliveira}, O.~C.~O. {Dahlsten}, and M.~B. {Plenio}.
\newblock {Efficient Generation of Generic Entanglement}.
\newblock {\em Phys. Rev. Lett.}, 98(130502), 2007.
\newblock arXiv:quant-ph/0605126.

\bibitem[{Pag}93]{PagesConjecture}
D.~N. {Page}.
\newblock {Average entropy of a subsystem}.
\newblock {\em Phys. Rev. Lett.}, 71:1291, 1993.

\bibitem[PCZ97]{PoyatosCiracZoller97}
J.~F. Poyatos, J.~I. Cirac, and P.~Zoller.
\newblock {Complete Characterization of a Quantum Process: The Two-Bit Quantum
  Gate}.
\newblock {\em Phys. Rev. Lett.}, 78(2):390--393, 1997.
\newblock arXiv:quant-ph/9611013.

\bibitem[PSW06]{ThermalisationPSW}
S.~{Popescu}, A.~J. {Short}, and A.~{Winter}.
\newblock {Entanglement and the Foundations of Statistical Mechanics}.
\newblock {\em Nature Physics}, 2:754--758, 2006.
\newblock arXiv:quant-ph/0511225.

\bibitem[{Rot}64]{RotaMobius}
G.-C. {Rota}.
\newblock {On the foundations of combinatorial theory I. Theory of M\"{o}bius
  Functions}.
\newblock {\em Probability Theory and Related Fields}, 2(4):340--368, 1964.

\bibitem[{San}95]{PagesConjectureProof95}
J.~{Sanchez-Ruiz}.
\newblock {Simple proof of Page's conjecture on the average entropy of a
  subsystem}.
\newblock {\em Phys. Rev. E}, 52:5653, 1995.

\bibitem[{Sen}05]{Sen05}
P.~{Sen}.
\newblock {Random measurement bases, quantum state distinction and applications
  to the hidden subgroup problem}.
\newblock {\em IEEE Conference on Computational Complexity 2006}, pages
  274--287, 2005.
\newblock arXiv:quant-ph/0512085.

\bibitem[{Shi}95]{GeometricEntanglementShimony}
A.~{Shimony}.
\newblock {Degree of Entanglement}.
\newblock {\em Ann. N.Y. Acad. Sci.}, 755:675, 1995.

\bibitem[Sho94]{ShorsAlgorithm}
P.~W. Shor.
\newblock {Algorithms for quantum computation: Discrete logarithms and
  factoring}.
\newblock {\em 35th Annual IEEE Symposium on Foundations of Computer Science},
  pages 124--134, 1994.

\bibitem[Sho96]{ShorFaultTolerance}
P.~W. Shor.
\newblock {Fault-tolerant quantum computation}.
\newblock {\em 37th Annual IEEE Symposium on Foundations of Computer Science},
  pages 56--65, 1996.
\newblock arXiv:quant-ph/9605011.

\bibitem[SL06]{SmithLeung06}
G.~{Smith} and D.~{Leung}.
\newblock {Typical entanglement of stabilizer states}.
\newblock {\em Phys. Rev. A}, 74(6):062314, 2006.
\newblock arXiv:quant-ph/0510232.

\bibitem[SS08]{SekinoSusskind08}
Y.~{Sekino} and L.~{Susskind}.
\newblock {Fast scramblers}.
\newblock {\em Journal of High Energy Physics}, 10:65, 2008.
\newblock arXiv:0808.2096.

\bibitem[{Sta}86]{EnumerativeCombinatorics}
R.~{Stanley}.
\newblock {\em {Enumerative Combinatorics}}.
\newblock {Cambridge University Press}, {Cambridge, UK}, 1986.

\bibitem[TDL01]{TDL01}
B.~M. {Terhal}, D.~P. {Divincenzo}, and D.~W. {Leung}.
\newblock {Hiding Bits in Bell States}.
\newblock {\em Phys. Rev. Lett.}, 86:5807--5810, 2001.
\newblock arXiv:quant-ph/0011042.

\bibitem[{van}02]{VanDamThesis}
W.~K. {van Dam}.
\newblock {On Quantum Computation Theory}.
\newblock {\em PhD Thesis, University of Amsterdam}, 2002.

\bibitem[{Zni}07]{Znidaric07}
M.~{Znidaric}.
\newblock {Optimal two-qubit gate for generation of random bipartite
  entanglement}.
\newblock {\em Phys. Rev. A}, 76(012318), 2007.
\newblock arXiv:quant-ph/0702240.

\end{thebibliography}
\end{document}